\documentclass[12pt]{article}
\usepackage{amsfonts,amsmath,amssymb,epsf,bbm,color,array,colortbl,xypic,longtable,extarrows}
\usepackage{tikz}
\usetikzlibrary{arrows,decorations.pathmorphing,positioning,matrix,scopes}
\usepackage{ifpdf}
\ifpdf
\usepackage{epstopdf,hyperref}
\else
\usepackage[hypertex]{hyperref}
\fi
\usepackage{pifont}
\allowdisplaybreaks[1]

\advance\textheight by \topskip
\usepackage[all]{xy}

\usepackage{amsthm} 
\newtheorem{Thm}{Theorem}

\newtheorem{Lem}[Thm]{Lemma}

\theoremstyle{definition}

\numberwithin{equation}{section}
\numberwithin{Thm}{section}
                   
\newcommand\be            {\begin{equation}}
\newcommand\bea           {\begin{equation}\begin{array}l\displaystyle}
\newcommand\bearll        {\begin{array}{ll}\displaystyle}
\newcommand\ee            {\end{equation}}
\newcommand\eear          {\end{array}}
\newcommand\enl           {\\[1em]\displaystyle}
\newcommand\etb           {& \displaystyle}

\newcommand\labl[1]       {\label{#1}\ee}

\newcommand\eps           {\varepsilon}
\newcommand\Hom           {\mathrm{Hom}}
\newcommand\id            {{\rm id}}
\newcommand\Kdd            {\mathrm{K}_0^b}

\newcommand\one           {{\bf1}}

\newcommand\Rep           {\mathrm{Rep}}

\newcommand\Cb            {\mathbb{C}}

\newcommand\Zb            {\mathbb{Z}}

\newcommand\Cc            {\mathcal{C}}
\newcommand\Hc            {\mathcal{H}}
\newcommand\Nc            {\mathcal{N}}

\newcommand\Qc            {\mathcal{Q}}
\renewcommand\Rc            {\mathcal{R}}
\newcommand\Sc           {\mathcal{S}}

\newcommand\Wc           {\mathcal{W}}
\newcommand\Pc           {\mathcal{P}}

\begin{document}
\thispagestyle{empty}
\def\thefootnote{\fnsymbol{footnote}}
\begin{flushright}
\scriptsize ZMP-HH/10-18\\
\scriptsize Hamburger Beitr\"age zur Mathematik 387\\
\href{http://arxiv.org/abs/1008.0082}{[1008.0082 [hep-th]]}
\end{flushright}
\vskip 2.5em
\begin{center}\LARGE
A modular invariant bulk theory \\
for the $c=0$ triplet model
\end{center}\vskip 2em
\begin{center}\large
  Matthias R. Gaberdiel%
  $^{a}$\footnote{Email: {\tt gaberdiel@itp.phys.ethz.ch}}, 
  Ingo Runkel%
  $^{b}$\footnote{Email: {\tt ingo.runkel@uni-hamburg.de}} \,%
  and
  Simon Wood%
  $^{a}$\footnote{Email: {\tt swood@itp.phys.ethz.ch}}
\end{center}
\begin{center}\it$^a$
Institute for Theoretical Physics, ETH Z\"urich \\
8093 Z\"urich, Switzerland
\end{center}
\begin{center}\it$^b$
Department Mathematik, Universit\"at Hamburg \\
Bundesstra\ss e 55, 20146 Hamburg, Germany
\end{center}
\vskip 1em
\begin{center}
July 2010
\end{center}
\vskip 1em
\begin{abstract}
A proposal for the bulk space of the logarithmic $\mathcal{W}_{2,3}$-triplet model 
at central charge zero is made. The construction is based on the idea that
one may reconstruct the bulk theory of a rational conformal field theory from
its boundary theory. The resulting bulk space is a quotient 
of the direct sum of projective representations, which is isomorphic,
as a vector space, to the direct sum of tensor products
of the irreducible representations with their projective covers. 
As a consistency check of our analysis we show that 
the partition function of the bulk theory is modular invariant, and 
that the boundary state analysis is compatible with the proposed
annulus partition functions of this model.
\end{abstract}

\setcounter{footnote}{0}
\def\thefootnote{\arabic{footnote}}

\newpage

\tableofcontents

\newpage

\section{Introduction and summary}

Logarithmic conformal field theories \cite{Gurarie:1993xq} with $c=0$ 
appear quite generically in models with a non-trivial fixed point theory
but a trivial generating function \cite{Cardy:gen}. Prime examples of such
theories are systems with quenched disorder \cite{Gurarie:1999yx}
that include, among others, percolation 
\cite{Cardy:1991cm,Watts:1996yh,Flohr:2005ai,Rasmussen:2007pc,Mathieu:2007pe,%
Rasmussen:2008ii,Ridout:2008cv}
and polymers  \cite{Saleur:1991hk,Read:2001pz,Pearce:2006we,DJS}. 
Another area where  $c=0$ conformal field theories have recently
made an appearance is as the dual theory for chiral gravity in AdS$_3$
\cite{Li:2008dq}, and also in this context there are indications that the conformal
field theory is logarithmic \cite{Grumiller:2008es,Maloney:2009ck}.

Much has been learned over the last fifteen years about logarithmic conformal
field theories. In particular, the triplet theory at $c=-2$ \cite{Kausch:1990vg} has
been understood in great detail \cite{Gaberdiel:1996np,Gaberdiel:1998ps,Kawai:2001ur,Gaberdiel:2006pp},
as have been the WZW models with supergroup target spaces
\cite{Rozansky:1992rx,Zirnbauer:1999ua,Schomerus:2005bf,Gotz:2006qp,Saleur:2006tf,Quella:2007hr}.
There has also been some progress towards understanding 
WZW models at fractional level that  show logarithmic
behaviour  \cite{Gaberdiel:2001ny,Lesage:2002ch,Lesage:2003kn,Ridout:2010qx},
logarithmic extensions of minimal models \cite{Pearce:2006sz},
and the structure of general indecomposable 
Virasoro-representations \cite{Kytola:2009ax}.
However, most of the models that have been understood behave quite differently from the 
theory at $c=0$. Indeed, $c=0$ is special since the Virasoro field is 
`null' at $c=0$ in the sense that it has a vanishing $2$-point function. As a consequence
the vacuum representation of the chiral algebra is not irreducible, 
and a number of novel phenomena appear. Many of these are also shared by 
the whole family of logarithmic minimal ${\cal W}_{p,q}$-models with $p,q\geq 2$
\cite{Feigin:2006iv} that have attracted a lot of attention during the last
four years \cite{Eberle:2006zn,Feigin:2006xa,Semikhatov:2007qp,Rasmussen:2008ii,%
Rasmussen:2008xi,Rasmussen:2008ez,Gaberdiel:2009ug,Rasmussen:2009zt,Wood:2009ub,%
AM1,Rasmussen:2009xa,AM}.
Significant progress has been made for these theories, and the fusion rules as well as the modular properties
of the chiral representations are now fairly well understood. In \cite{Gaberdiel:2009ug}
a consistent boundary theory was proposed for the simplest member of this family of models,
the theory with $c_{2,3}=0$,
and the generalisation to all ${\cal W}_{p,q}$-models was sketched in \cite{Wood:2009ub}. 
However, so far an understanding of the corresponding bulk theory is missing. 
\smallskip

In this paper we shall construct the bulk space of states  for the ${\cal W}_{2,3}$-model
that corresponds to the boundary theory of  \cite{Gaberdiel:2009ug}. Our strategy is to reconstruct
the bulk theory, starting from the boundary ansatz. This idea goes back to 
\cite{Runkel:1998pm,Felder:1999mq,tft1} (see also \cite{Petkova:2001ag} and 
\cite{unique,Kong:2008ci}), and was successfully applied to the logarithmic $\Wc_{1,p}$ triplet 
models  (and in particular to the theory at $c=-2$) in \cite{Gaberdiel:2007jv}. There are a number
of subtleties that arise in the present context, and as a consequence 
the analysis is technically considerably harder than in  \cite{Gaberdiel:2007jv}. However, the resulting
space of bulk states still has all properties one should expect it to  have. In particular, 
the partition function is modular invariant, and it allows for the construction
of boundary states in agreement with the prediction of \cite{Gaberdiel:2009ug}. 
Both of these are non-trivial consistency checks, 
and we are therefore confident that our bulk ansatz defines a consistent bulk conformal field theory. 

Many structural features are similar to what happens for the $\Wc_{1,p}$-models
(and in particular at $c=-2$). For example the space $\Hc_\text{bulk}$ of bulk states is of 
the form
\be\label{eq:Hbulk-decomp-WxP}
  \Hc_\text{bulk} = {\textstyle \bigoplus_{i \in \text{Irr}}} \Wc(i) \otimes_\Cb \bar\Pc(i)
  \qquad
  \text{as $(L_0^\text{diag},\bar L_0^\text{diag})$- graded vector space}\ ,
\ee
where the sum runs over all irreducible $\Wc_{2,3}$-representations $\Wc(i)$, 
$\Pc(i)$ denotes the projective cover of $\Wc(i)$, and $L_0^\text{diag}$ (resp.\ $\bar L_0^\text{diag}$) refers to the
diagonal part of the action of $L_0$ (resp.\ $\bar L_0$).
However, there are also a number of remarkable new phenomena that appear
for the $\Wc_{2,3}$-model at $c=0$. 
Firstly, in distinction to all rational CFTs and the $\Wc_{1,p}$-models, 
the bulk theory does not contain $\Wc_{2,3} \otimes_\Cb \bar\Wc_{2,3}$ as a 
sub-representation, but only as a sub-quotient.
The second surprising feature is that in order to
reproduce the annulus amplitudes one only needs ten Ishibashi states, while 
the characters of the representations labelling the boundary states of \cite{Gaberdiel:2009ug} 
span a $12$-dimensional space. The actual boundary states must involve
additional Ishibashi states in order for the bulk-boundary correlators to be non-degenerate 
in the bulk entry, but these additional Ishibashi states do not contribute to the annulus diagrams
(without any additional insertions), and are hence invisible from the point of view of the usual
annulus analysis.
\bigskip

The paper is organised as follows. 
In the remainder of this introduction we shall first review some of the results from
\cite{Gaberdiel:2009ug} that we shall need in the rest of the paper. In particular, 
Section~\ref{sec:intro-chiral} reviews the representation theory of the $ \Wc_{2,3}$-algebra 
in question, while Section~\ref{sec:intro-bnd} summarises the boundary theory that was proposed 
in \cite{Gaberdiel:2009ug}. In Section~\ref{sec:intro-bulkbnd} we sketch the general strategy for 
the construction of the bulk theory starting from this boundary ansatz. Finally, 
Section~\ref{sec:intro-bndstate} describes the main features of the corresponding boundary 
state analysis. These sections are meant to give a non-technical account of these considerations;
the actual details are then spelled out in Sections~\ref{sec:bulk} and~\ref{sec:bnd}, 
respectively, while Section~\ref{sec:conclusion} contains our conclusions. Some of the tables, 
diagrams and technical constructions have been delegated into appendices.

\smallskip

\subsection[The chiral ${\cal W}_{2,3}$-model]{The chiral $\boldsymbol{{\cal W}_{2,3}}$-model}\label{sec:intro-chiral}

The ${\cal W}_{p,q}$-models have central charge 
\begin{align}
  c_{p,q}=1-6\frac{(p-q)^2}{pq} \ ,
\end{align}
where $p$ and $q$ are a pair of positive coprime integers. In the following
we shall only consider the case  $(p,q)=(2,3)$ for which $c_{2,3}=0$. 
The chiral algebra $\Wc\equiv \Wc_{2,3}$ is generated by the Virasoro algebra, as well 
as by a triplet of fields of conformal weight $h=15$. In the vacuum representation 
$L_{-1}\Omega=0$, but $T=L_{-2}\Omega\neq 0$. Since the positive modes annihilate $T$
(in particular, $c=0$ implies that $L_2 T=0$), the 
vacuum representation  contains a proper subrepresentation, namely the highest
weight representation $\Wc(2)$ generated from $T$ by the action of the negative modes. 
Thus $\Wc$ is not irreducible, but has the schematic structure
\begin{equation}\label{Wemb}
  \rule{0pt}{1.5em}
  \xymatrix@R=0em{
  \Omega \ar@{.>}[r] \ar@/^1.2em/[rr]&\times&T\ar@{->}[rr]|{\Wc(2)}&& \\
  h=0 & h=1 & h=2 
  }
\end{equation}
\smallskip

\noindent
The irreducible representations of the $\Wc$ algebra are described by the finite Kac table
\cite{Feigin:2006iv}
\be
  \begin{tabular}{c|ccc}
    & $s=1$ & $s=2$ & $s=3$ \\
    \hline
    $r=1$ & \colorbox[gray]{.8}{$0,\,2,\,7$}~&~ \colorbox[gray]{.8}{$0,\,1,\,5$}~
    &$\frac{1}{3},\,\frac{10}{3}$\\[.5em]
    $r=2$ & $\frac58,\,\frac{33}{8}$& $\frac18,\,\frac{21}{8}$& $\frac{-1}{24},\,\frac{35}{24}$
  \end{tabular}
\labl{eq:W23-irreps}
Here each entry $h$ is the conformal dimension of the highest
weight states of an irreducible representation, which we shall denote by $\Wc(h)$. 
There is only one representation corresponding to $h=0$, 
namely the one-dimensional vacuum representation $\Wc(0)$, spanned 
by the vacuum vector $\Omega$. The significance of the grey boxes will be explained 
below in Section~\ref{sec:intro-bnd}.

As is familiar from other logarithmic theories, the $13$ irreducible representations 
in \eqref{eq:W23-irreps} do 
not close among themselves under fusion. 
The minimal set of representations which includes \eqref{eq:W23-irreps} and 
is closed under fusion and taking conjugates
involves in addition the $22$ indecomposable representations
\bea
  \colorbox[gray]{.8}{\rule[-2pt]{0pt}{1.1em}$\Wc$}~,~~
  \colorbox[gray]{.8}{\rule[-2pt]{0pt}{1.1em}$\Wc^*$}~,~~
  \colorbox[gray]{.8}{\rule[-2pt]{0pt}{1.1em}$\Qc$}~,~~
  \colorbox[gray]{.8}{\rule[-2pt]{0pt}{1.1em}$\Qc^*$}~,~~
  \Rc^{(2)}(0,2)_7~,~~
  \Rc^{(2)}(2,7)~,~~
  \Rc^{(2)}(0,1)_5~,~~
  \Rc^{(2)}(1,5)~,~~
\\[.3em]\displaystyle
  \Rc^{(2)}(0,2)_5~,~~
  \Rc^{(2)}(2,5)~,~~
  \Rc^{(2)}(0,1)_7~,~~
  \Rc^{(2)}(1,7)~,~~
  \Rc^{(2)}(\tfrac13,\tfrac13)~,~~
\\[.3em]\displaystyle
  \Rc^{(2)}(\tfrac13,\tfrac{10}{3})~,~~
  \Rc^{(2)}(\tfrac58,\tfrac58)~,~~
  \Rc^{(2)}(\tfrac58,\tfrac{21}{8})~,~~
  \Rc^{(2)}(\tfrac18,\tfrac18)~,~~
  \Rc^{(2)}(\tfrac18,\tfrac{33}{8})~,~~
\\[.3em]\displaystyle 
  \Rc^{(3)}(0,0,1,1)~,~~
  \Rc^{(3)}(0,0,2,2)~,~~
  \Rc^{(3)}(0,1,2,5)~,~~
  \Rc^{(3)}(0,1,2,7) \ ,
\eear\labl{eq:indec-W-rep}
whose structure (including characters and embedding diagrams) 
are described in some detail in \cite[App.\,A]{Gaberdiel:2009ug} and Appendix~\ref{app:W23-reps} 
below. Note that in addition to the chiral algebra $\Wc$ also
its conjugate representation $\Wc^*$ appears; because of the structure of (\ref{Wemb})
it is not isomorphic to $\Wc$. (Indeed, the embedding structure of  $\Wc^*$ can be obtained
from that of $\Wc$ by reversing the direction of the arrow $\Omega\rightarrow T$.)

The associative fusion rules of all of these $13+22=35$ representations were given
explicitly in  \cite{Gaberdiel:2009ug}. They were derived based on a direct calculation of the 
fusion rules for the Virasoro subalgebra (compare \cite{Eberle:2006zn}), and agree with what was 
obtained based on a lattice analysis in  \cite{Rasmussen:2008ii,Rasmussen:2008ez}. 
 
For the construction of the bulk theory the projective covers of the irreducible representations 
play an important role. Let us denote the projective cover of the representation $\Wc(h)$
 by $\Pc(h)$. Then we have the identifications (this was already suggested in \cite{Gaberdiel:2009ug})
 \be
 \begin{array}{rlll}
 & \Pc(1) = {\cal R}^{(3)}(0,0,1,1) \ , \qquad
 & \Pc(2) = {\cal R}^{(3)}(0,0,2,2) \ ,  \qquad
 & \Pc(5) = {\cal R}^{(3)}(0,1,2,5) \ , \qquad  \\
  & \Pc(7) = {\cal R}^{(3)}(0,1,2,7) \ , \qquad
 & \Pc(\tfrac{1}{3}) = {\cal R}^{(2)} (\tfrac{1}{3},\tfrac{1}{3})  \ ,  \qquad
 & \Pc(\tfrac{10}{3}) = {\cal R}^{(2)} (\tfrac{1}{3},\tfrac{10}{3})  \ ,   \\
  & \Pc(\tfrac{5}{8}) = {\cal R}^{(2)} (\tfrac{5}{8},\tfrac{5}{8}) \ , \qquad
 & \Pc(\tfrac{21}{8}) = {\cal R}^{(2)} (\tfrac{5}{8},\tfrac{21}{8}) \ , \qquad
 & \Pc(\tfrac{1}{8}) = {\cal R}^{(2)} (\tfrac{1}{8},\tfrac{1}{8}) \ , \qquad  \\
 & \Pc(\tfrac{33}{8}) = {\cal R}^{(2)} (\tfrac{1}{8},\tfrac{33}{8}) \ ,  \qquad
 & \Pc(\tfrac{-1}{24}) = \Wc(\tfrac{-1}{24})  \ , \qquad
 & \Pc(\tfrac{35}{24}) = \Wc(\tfrac{35}{24})  \ .
 \end{array}
 \labl{eq:proj-cover-list}
As was also explained in \cite{Gaberdiel:2009ug}, none of 
the representations in \eqref{eq:W23-irreps} and \eqref{eq:indec-W-rep} can be the
projective cover of $\Wc(0)$. On the other hand, on general 
grounds \cite{Huang:2007ir,Huang:2007mj,Huang:2009} 
one may expect that every irreducible 
representation has a projective cover. We shall propose the structure of the projective cover 
$\Pc(0)$ for $\Wc(0)$ in Section~\ref{sec:p0structure} below. 
In Appendix~\ref{app:P0-cover} we shall derive some further properties of $\Pc(0)$, and 
deduce its fusion rules with most of the representations in \eqref{eq:W23-irreps} and 
\eqref{eq:indec-W-rep}.

\subsection{The boundary theory}\label{sec:intro-bnd}

For non-logarithmic rational conformal field theories 
with charge-conjugation modular invariant (the `Cardy case')
the boundary conditions that preserve
the chiral symmetry are in one-to-one correspondence to the representations of the
chiral algebra \cite{Cardy:1989ir}. The situation is similar
for the $\Wc_{1,p}$-models, but as was explained in 
\cite{Gaberdiel:2009ug}, the situation is more complicated for the $\Wc_{2,3}$-model at 
$c=0$. Indeed, consistent boundary conditions can only be associated to a subset of the
$35$ representations above, namely to those $26$ representations that are 
{\em not in grey boxes}.\footnote{There may, however, be additional boundary conditions 
associated to other classes of representations.}
These allowed representations are characterised by the property that 
the conjugate representation agrees with the dual representation. This condition guarantees
that the space of boundary fields
has a non-degenerate bilinear pairing, {\it i.e.}\ that the 
boundary $2$-point functions are non-degenerate (for more details see \cite{Gaberdiel:2009ug}).
In particular, this then also implies that the space of boundary fields of any consistent boundary
condition cannot just be $\Wc_{2,3}$ itself.

The space of open string states $\Hc_{\Rc\rightarrow\mathcal{S}}$ between two boundary 
conditions labelled by ${\cal R}$ and ${\cal S}$
is given by the fusion product  $\Hc_{\Rc\rightarrow\mathcal{S}} = {\cal S}\otimes_f {\cal R}^*$, 
where $\otimes_f$ denotes the fusion product of $\Wc$-representations and
${\cal R}^*$ is the conjugate representation
to ${\cal R}$; this is as in the usual Cardy case. 

A convenient tool to analyse the annulus partition functions, {\it i.e.}\ the characters of the 
$\Wc$-representation $\Hc_{\Rc\rightarrow\mathcal{S}}$, is provided by the (additive) 
Grothendieck group $\mathrm{K}_0$. It consists of equivalence classes $[\Rc]$ of 
representations $\Rc$, where two representations $\Rc$ and $\Rc'$ are equivalent 
iff they have the same character\footnote{This description of $\mathrm{K}_0$ is correct if 
the irreducible representations of $\Wc$ have linearly independent characters, which is the case 
for $\Wc_{2,3}$.}, $\chi_\Rc(q) = \chi_{\Rc'}(q)$. Thus, giving the 
element $[\Rc]$ in $\mathrm{K}_0$ for a representation $\Rc$ is the same as specifying 
its character, and we can describe the annulus partition functions as
elements $[\Hc_{\Rc\rightarrow\mathcal{S}}] \in \mathrm{K}_0$.

The additive group $\mathrm{K}_0$ is generated by 13 elements, 
corresponding to the characters of the irreducible representations.
The representations which lead to consistent boundary conditions, {\it i.e.}\ those 
not in grey boxes, generate a subgroup $\Kdd$ of $\mathrm{K}_0$ spanned 
(for example) by the following 12 independent elements \cite[Sect.\,2.4]{Gaberdiel:2009ug},
\bea
   \Kdd = \text{span}_{\Zb}\Big(  [\Wc(h)] \big| h = 0,\tfrac{1}{3}, \tfrac{10}{3}, \tfrac{5}{8}, \tfrac{33}{8}, 
   \tfrac{1}{8}, \tfrac{21}{8}, \tfrac{-1}{24}, \tfrac{35}{24} \Big) \\[.5em]\displaystyle
   \hspace{4em}
   \oplus ~ 2 \Zb \big( [\Wc(2)]{+}[\Wc(7)] ) 
   \, \oplus \, 2 \Zb \big( [\Wc(1)]{+}[\Wc(5)] ) 
   \, \oplus \, 2 \Zb \big( [\Wc(2)]{+}[\Wc(5)] ) ~. 
\eear\labl{eq:K0dd-basis}
The fusion product makes $\Kdd$ (but not $\mathrm{K}_0$) into a ring.
By definition, this means that the annulus partition function for boundary conditions $\Rc$ and 
$\mathcal{S}$ can be written as
\be
  [\Hc_{\Rc\rightarrow\mathcal{S}}] = [\Sc] \cdot [\mathcal{R}^*] \ .
\ee
Clearly, two representations $\Rc$ and $\Rc'$ lead to the same annulus partition function
$[\Hc_{\Rc\rightarrow\mathcal{S}}] = [\Hc_{\Rc'\rightarrow\mathcal{S}}]$ if they have the
same character, because then $[\Rc] = [\Rc']$. However, the converse is not true: the ring $\Kdd$ 
contains a $2$-dimensional null ideal spanned by 
\be\label{nullideal}
N_1=    [\Wc(\tfrac{5}{8})] - [\Wc(\tfrac{33}{8})] 
- [\Wc(\tfrac{1}{8})] + [\Wc(\tfrac{21}{8})] - [\Wc(\tfrac{-1}{24})] + [\Wc(\tfrac{35}{24})]
~,~~  N_2 = [\Wc(0)] \ .
\ee
These two elements have the property that $N_1 \cdot C = 0 = N_2 \cdot C$ for all $C \in \Kdd$.
Thus any two representations $\Rc$ and $\Rc'$ with the property that 
$[\Rc]-[\Rc'] \in \text{span}_{\Zb}\big( [N_1] , [N_2] \big)$
also lead to identical annulus partition functions. For example,
the two representations $\Rc =  \Wc(\tfrac{5}{8}) \oplus \Wc(\tfrac{21}{8}) \oplus \Wc(\tfrac{35}{24})$
and $\Rc' = \Wc(\tfrac{33}{8}) \oplus \Wc(\tfrac{1}{8}) \oplus \Wc(\tfrac{-1}{24})$ have the property
that $[\Rc]-[\Rc']=[N_1]$. As a consequence their annulus amplitudes with any boundary
condition $\mathcal{S}$ agrees,
$[\Hc_{\Rc\rightarrow\mathcal{S}}] = [\Hc_{\Rc'\rightarrow\mathcal{S}}]$.
This does not imply that $\Rc$ and $\Rc'$ have identical boundary states 
$\|\Rc\rangle\!\rangle$ and $\|\Rc'\rangle\!\rangle$; 
indeed, one may expect that suitable correlators of bulk fields on a disc can distinguish two boundary conditions
$\Rc$ and $\Rc'$ whenever $[\Rc] \neq [\Rc']$ in $\Kdd$.
However, it does mean that $10$ Ishibashi states are 
sufficient to reproduce all annulus partition functions as overlaps 
$\langle\!\langle \Rc \| q{}^{L_0} \bar q{}^{\bar L_0} \|\mathcal{S}\rangle\!\rangle$.

\subsection{From boundary to bulk}\label{sec:intro-bulkbnd}

For non-logarithmic rational conformal field theories one can reconstruct the 
bulk theory from one consistent boundary condition \cite{Felder:1999mq,tft1}. The same
method was also successfully applied in  \cite{Gaberdiel:2007jv} to the case of the
logarithmic  $\Wc_{1,p}$-models. The idea behind the construction is as follows. 

Suppose we are given a boundary condition together with its space of boundary fields and
the associated operator product expansion (OPE). We consider the disc correlator involving one boundary field
and one field from the yet-to-be-constructed space of bulk states. This correlator 
defines  a pairing between the bulk and boundary degrees of freedom
\begin{equation}\label{bulkbdy}
{\cal H}_{\rm bulk}^{({\rm ans})} \times {\cal H}_{\rm bnd} \rightarrow {\mathbb C}: 
(\psi,\phi) \mapsto \big\langle V(\psi,0)\, V(\phi,u) \big\rangle_{\!\text{disc}} \ ,
\end{equation}
where $|u|=1$ is a point on the perimeter of the disc, and the index `(ans)' indicates
that at this stage we can only make an ansatz for the bulk space of states. 
If $\phi\in {\cal R} \otimes_\Cb \bar{\cal S}$ is in a representation of the
left- and right-moving chiral algebra, then the map (\ref{bulkbdy}) is 
uniquely determined up to some constants (one for each allowed fusion
channel) since, by the usual doubling trick, the correlator can be thought of 
as a chiral $3$-point function.  These constants encode the bulk-boundary OPE, and they
are constrained by two necessary requirements. First of all, it 
follows from general principles (non-degeneracy of the bulk $2$-point function on the sphere) 
that the map (\ref{bulkbdy}) has to be non-degenerate in the first (bulk) entry. Furthermore, given
the OPE of the boundary fields, it also defines correlators involving more than
one boundary field, and these have to satisfy the appropriate locality conditions,
see Section~\ref{sec:bulk-boundary-map}.
The `correct' space of bulk states is then simply
the {\em largest possible} space compatible with these requirements. It was proven in \cite{unique,Kong:2008ci} 
that in the non-logarithmic setting, this construction reproduces indeed the 
unique space of bulk states that is compatible with the given boundary condition. For the logarithmic
$\Wc_{1,p}$-models the construction was also shown to lead to sensible results; in particular,
the known $c=-2$ bulk theory of \cite{Gaberdiel:1998ps} was correctly reproduced by this 
method in  \cite{Gaberdiel:2007jv}.

In all of these constructions the analysis is simplest if the space of boundary fields just consists of
the chiral algebra (or VOA) itself. This is the case for the `identity Cardy brane'
and it leads to the charge-conjugation modular invariant theory. Indeed, such a brane exists
for all the $\Wc_{1,p}$-models, and it was taken as the starting point for the analysis of 
\cite{Gaberdiel:2007jv}. As already noted at the beginning of Section~\ref{sec:intro-bnd}, 
one of the complications for $c=0$ is that such a boundary condition
does not exist since the non-degeneracy of the boundary 2-point function forbids the space of 
boundary fields to consist just of the chiral algebra itself. However, as we shall propose
in Section~\ref{proposal}, for the purpose of analysing the consistency of the bulk ansatz, 
one may assume that the space of boundary fields just  consists of $\Wc^*$, the conjugate
of the VOA $\Wc$, so that we can proceed very similarly  to \cite{Gaberdiel:2007jv}. 
In particular, the solution to the maximality condition is
\begin{equation}\label{ansa}
{\cal H}_{\rm bulk}  =  \Bigl( \bigoplus_j \Pc(j) \otimes_\Cb \bar\Pc(j)^* \Bigr)/ {\cal N} \ , 
\end{equation}
where the sum runs over all irreducible representations, $\Pc(j)$ denotes the corresponding 
projective cover, and ${\cal N}$ is a certain subspace that can be calculated as in 
\cite{Gaberdiel:2007jv}, and that guarantees that the bulk-boundary map is non-degenerate in
the bulk entry. For the $\Wc_{2,3}$-model we will see in Section~\ref{proposal} that the space
of bulk states has the form
\begin{equation}\label{bulk}
{\cal H}_{\rm bulk} = {\cal H}_0 \oplus  {\cal H}_{\frac{1}{8}} \oplus {\cal H}_{\frac{5}{8}}
\oplus {\cal H}_{\frac{1}{3}} \oplus {\cal H}_{\frac{-1}{24}}  \oplus {\cal H}_{\frac{35}{24}}\ ,
\end{equation}
where we have labelled the individual blocks ${\cal H}_h$ by the conformal weight of the
lowest state. The corresponding characters are explicitly given by
\be\begin{array}{rl}\displaystyle
\mathrm{tr}_{\Hc_0}\big(q{}^{L_0} \bar q{}^{\bar L_0}\big) 
\etb\!\!\!= 2\bigl| \chi_{\Wc(0)}(q)\bigr|^2 + 
\big| \chi_{\Wc(0)}(q) + 2\chi_{\Wc(1)}(q) + 2\chi_{\Wc(2)}(q) 
+ 2\chi_{\Wc(5)}(q) + 2\chi_{\Wc(7)}(q) \big|^2 \ ,
\enl
\mathrm{tr}_{\Hc_{\frac{1}{8}}}\big(q{}^{L_0} \bar q{}^{\bar L_0}\big) 
\etb\!\!\!= 2 \big| \chi_{\Wc(\frac18)}(q) + \chi_{\Wc(\frac{33}8)}(q) \big|^2 \ ,
\enl
\mathrm{tr}_{\Hc_{\frac{5}{8}}}\big(q{}^{L_0} \bar q{}^{\bar L_0}\big) 
\etb\!\!\!= 2 \big| \chi_{\Wc(\frac58)}(q) + \chi_{\Wc(\frac{21}8)}(q) \big|^2 \ ,
\enl
\mathrm{tr}_{\Hc_{\frac{1}{3}}}\big(q{}^{L_0} \bar q{}^{\bar L_0}\big) 
\etb\!\!\!= 2 \big| \chi_{\Wc(\frac13)}(q) + \chi_{\Wc(\frac{10}3)}(q) \big|^2 \ ,
\enl
\mathrm{tr}_{\Hc_{\frac{-1}{24}}}\big(q{}^{L_0} \bar q{}^{\bar L_0}\big) 
\etb\!\!\!= \big| \chi_{\Wc(-\frac1{24})}(q) \big|^2 \quad, \quad \text{and}  \quad
\mathrm{tr}_{\Hc_{\frac{35}{24}}}\big(q{}^{L_0} \bar q{}^{\bar L_0}\big) 
= \big| \chi_{\Wc(\frac{35}{24})}(q) \big|^2 \ ,
\eear\label{eq:block-chars}
\ee
where $\chi_{\Wc(h)}(q)$ denotes the character of the irreducible representation
$\Wc(h)$, see \cite{Feigin:2006iv} or \cite[App.\,A.1]{Gaberdiel:2009ug} for explicit expressions.
Our construction, however, contains much more information than just these characters; in
fact, our analysis leads to a description of the $\Hc_h$ as a
representation of $\Wc \otimes_\Cb \bar\Wc$. 

A convenient way to represent the structure
of these (not fully reducible) representations is the composition series, which is defined as follows.
Starting from a representation $M_1$, one finds the largest sub-representation $R_1$ which 
can be written as a direct sum of irreducible representations --- these are called composition factors.
Then one takes the quotient of $M_1$ by $R_1$ and repeats the procedure with $M_2=M_1/R_1$. 
In other words, one constructs a chain of sub-representations
\be
  M_1 = A_n \supset A_{n-1} \supset \cdots \supset A_2 \supset A_1 = R_1
\ee
such that $R_i = A_i / A_{i-1}$ is a direct sum of irreducible representations.\footnote
  {It is more common to require each quotient $A_i / A_{i-1}$ to be irreducible, rather
  than fully reducible. In this case the composition series is only unique up to permutations
  of its composition factors.}
We represent the quotients $R_i$ of a composition series as
\be
 R_n \rightarrow R_{n-1} \rightarrow \cdots \rightarrow R_2 \rightarrow R_1  \ .
\ee
The action of ${\cal W}\otimes_\Cb {\bar{\cal W}}$ either maps states within a representation $R_j$ 
into one another, or moves them along arrows in the composition series. 

With this technology at hand, we can now describe the composition series of the representations
appearing in (\ref{bulk}). For $\Hc_{\frac{-1}{24}}$ and $\Hc_{\frac{35}{24}}$ the composition series consists 
just of a single term
\be\label{eq:irred-sector-comp-series}
 \Hc_{\frac{-1}{24}} ~:~ \Wc(\tfrac{-1}{24}) \otimes_\Cb \bar\Wc(\tfrac{-1}{24}) \ , \qquad
 \Hc_{\frac{35}{24}} ~:~ \Wc(\tfrac{35}{24}) \otimes_\Cb \bar\Wc(\tfrac{35}{24}) \ .
\ee
For ${\cal H}_{\frac{1}{8}}$ it has the same structure as in the $\Wc_{1,p}$-models (see \cite{Quella:2007hr} 
and \cite{Gaberdiel:2007jv}) 
\be
 {\cal H}_{\frac{1}{8}} ~ : \quad
 \begin{array}{rcl} 
   \Wc(\tfrac{1}{8}) \otimes_\Cb \bar\Wc(\tfrac{1}{8})
 &\oplus& \Wc(\tfrac{33}{8}) \otimes_\Cb \bar\Wc(\tfrac{33}{8})
\\
 &\downarrow&
\\
 2 \cdot \Wc(\tfrac{1}{8}) \otimes_\Cb \bar\Wc(\tfrac{33}{8})
 &\oplus& 2 \cdot \Wc(\tfrac{33}{8}) \otimes_\Cb \bar\Wc(\tfrac{1}{8})
\\
 &\downarrow&
\\
 \Wc(\tfrac{1}{8}) \otimes_\Cb \bar\Wc(\tfrac{1}{8})
 &\oplus& \Wc(\tfrac{33}{8}) \otimes_\Cb \bar\Wc(\tfrac{33}{8})\ ,
\end{array}
\quad 
\ee  
where we have written the composition series vertically. This is easier to visualise if we 
represent each direct sum by a little table where we indicate the multiplicity of each term 
$\Wc(h) \otimes_\Cb \bar\Wc(\bar h)$. For example, the composition series for 
${\cal H}_{1/8}$ is then written as
\be\label{eq:1/8_compseries}
\begin{tabular}{c|c|c|}
& $\tfrac{1}{8}$ & $\tfrac{33}{8}$ \\ 
\hline
$\tfrac{1}{8}$  & 1 & 0 \\
\hline
$\tfrac{33}{8}$ & 0 & 1 \\
\hline
\end{tabular}
\quad \longrightarrow \quad
\begin{tabular}{c|c|c|}
& $\tfrac{1}{8}$ & $\tfrac{33}{8}$ \\ 
\hline
$\tfrac{1}{8}$  & 0 & 2 \\
\hline
$\tfrac{33}{8}$ & 2 & 0 \\
\hline
\end{tabular}
\quad \longrightarrow \quad
\begin{tabular}{c|c|c|}
& $\tfrac{1}{8}$ & $\tfrac{33}{8}$ \\ 
\hline
$\tfrac{1}{8}$  & 1 & 0 \\
\hline
$\tfrac{33}{8}$ & 0 & 1 \\
\hline
\end{tabular}\ .
\ee
Here the horizontal direction gives $h$ and the vertical direction $\bar h$. 
The picture for ${\cal H}_{{5}/{8}}$ and ${\cal H}_{1/3}$ looks the same,
with $\{ \tfrac18, \tfrac{33}8 \}$ replaced by $\{ \tfrac58, \tfrac{21}8 \}$ and 
$\{ \tfrac13, \tfrac{10}3 \}$, respectively. For $\Hc_0$ the composition series is more 
complicated, 
\def\stab{\hspace*{-.6em}&\hspace*{-.4em}}
\begin{align}\label{eq:0_compseries}
\begin{tikzpicture}[>=latex,baseline=(F)]
  \node (A) at (0,0) {\scriptsize\begin{tabular}{c|c|c|c|c|c|}
 &\hspace*{-.6em} 0 \stab 1 \stab 2 \stab 5 \stab 7 \hspace*{-.6em}\\
\hline
0 &\hspace*{-.6em}   1 \stab   \stab   \stab   \stab   \hspace*{-.6em}\\
\hline
1 &\hspace*{-.6em}   \stab 1 \stab   \stab   \stab   \hspace*{-.6em}\\
\hline
2 &\hspace*{-.6em}   \stab   \stab 1 \stab   \stab   \hspace*{-.6em}\\
\hline
5 &\hspace*{-.6em}   \stab   \stab   \stab 1 \stab   \hspace*{-.6em}\\
\hline
7 &\hspace*{-.6em}   \stab   \stab   \stab   \stab 1 \hspace*{-.6em}\\
\hline
\end{tabular}};
  \node (B) at (-5,0) {\scriptsize\begin{tabular}{c|c|c|c|c|c|}
 &\hspace*{-.6em} 0 \stab 1 \stab 2 \stab 5 \stab 7 \hspace*{-.6em}\\
\hline
0 &\hspace*{-.6em}    \stab 1  \stab 1  \stab   \stab   \hspace*{-.6em}\\
\hline
1 &\hspace*{-.6em} 1 \stab  \stab   \stab  2 \stab 2  \hspace*{-.6em}\\
\hline
2 &\hspace*{-.6em} 1  \stab   \stab  \stab 2  \stab 2  \hspace*{-.6em}\\
\hline
5 &\hspace*{-.6em}   \stab  2 \stab 2  \stab  \stab   \hspace*{-.6em}\\
\hline
7 &\hspace*{-.6em}   \stab 2  \stab 2  \stab   \stab  \hspace*{-.6em}\\
\hline
\end{tabular}};
  \node (C) at (2.5,4) {\scriptsize\begin{tabular}{c|c|c|c|c|c|}
 &\hspace*{-.6em} 0 \stab 1 \stab 2 \stab 5 \stab 7 \hspace*{-.6em}\\
\hline
0 &\hspace*{-.6em} 1 \stab   \stab   \stab 2 \stab 2 \hspace*{-.6em}\\
\hline
1 &\hspace*{-.6em}  \stab 2 \stab 4 \stab   \stab   \hspace*{-.6em}\\
\hline
2 &\hspace*{-.6em}  \stab 4 \stab 2 \stab  \stab  \hspace*{-.6em}\\
\hline
5 &\hspace*{-.6em} 2 \stab  \stab  \stab 2 \stab 4 \hspace*{-.6em}\\
\hline
7 &\hspace*{-.6em} 2 \stab  \stab  \stab 4 \stab 2 \hspace*{-.6em}\\
\hline
\end{tabular}};
  \node (D) at (-2.5,4) {\scriptsize\begin{tabular}{c|c|c|c|c|c|}
 &\hspace*{-.6em} 0 \stab 1 \stab 2 \stab 5 \stab 7 \hspace*{-.6em}\\
\hline
0 &\hspace*{-.6em}    \stab 1  \stab 1  \stab   \stab   \hspace*{-.6em}\\
\hline
1 &\hspace*{-.6em} 1 \stab  \stab   \stab  2 \stab 2  \hspace*{-.6em}\\
\hline
2 &\hspace*{-.6em} 1  \stab   \stab  \stab 2  \stab 2  \hspace*{-.6em}\\
\hline
5 &\hspace*{-.6em}   \stab  2 \stab 2  \stab  \stab   \hspace*{-.6em}\\
\hline
7 &\hspace*{-.6em}   \stab 2  \stab 2  \stab   \stab  \hspace*{-.6em}\\
\hline
\end{tabular}};
  \node (E) at (-7.5,4) {\scriptsize\begin{tabular}{c|c|c|c|c|c|}
 &\hspace*{-.6em} 0 \stab 1 \stab 2 \stab 5 \stab 7 \hspace*{-.6em}\\
\hline
0 &\hspace*{-.6em}  1  \stab   \stab   \stab   \stab   \hspace*{-.6em}\\
\hline
1 &\hspace*{-.6em}   \stab 1 \stab   \stab   \stab   \hspace*{-.6em}\\
\hline
2 &\hspace*{-.6em}   \stab   \stab 1 \stab   \stab   \hspace*{-.6em}\\
\hline
5 &\hspace*{-.6em}   \stab   \stab   \stab 1 \stab   \hspace*{-.6em}\\
\hline
7 &\hspace*{-.6em}   \stab   \stab   \stab   \stab 1 \hspace*{-.6em}\\
\hline
\end{tabular}};
\node (F) at (-5,2) {};
\node (G) at (2.5,2) {};
\draw[<-] (A) -- (B);
\draw[<-] (B) -- (F.center) -- (G.center) -- (C);
\draw[<-] (C) -- (D);
\draw[<-] (D) -- (E);
\node at (1.5,0) {.};
\end{tikzpicture}
\end{align}
All empty entries are equal to `0'.

Adding all the tables in a composition series reproduces the multiplicities given in 
the partition function for each $\Hc_h$ from \eqref{eq:block-chars}.  The complete 
partition function turns out to be modular invariant, as must be the case for a consistent 
conformal field theory. This represents a non-trivial consistency check on our analysis. 

It is also worth mentioning that a non-degenerate bulk two-point function requires that $\Hc_\text{bulk}$ 
is isomorphic  to its conjugate representation $\Hc_\text{bulk}^*$. A necessary condition for this is 
that the  composition series does not change when reversing all arrows, which indeed holds for the 
series given above and provides another consistency check of our construction.

\subsection{The boundary states}\label{sec:intro-bndstate}

With the detailed knowledge of the proposed bulk theory at hand we can study whether
the boundary conditions of \cite{Gaberdiel:2009ug} can actually be described in terms of appropriate 
boundary states. More specifically, we can ask whether we can reproduce the annulus partition functions
of \cite{Gaberdiel:2009ug}  in terms of suitable boundary states of our proposed bulk theory. 

The first step of this analysis can be done without any detailed knowledge of the bulk theory. The 
proposal of \cite{Gaberdiel:2009ug} for the boundary conditions makes a prediction for the various
annulus partition functions. Because of the two-dimensional null ideal, see eq.\ (\ref{nullideal}), we
expect that these annulus amplitudes can be reproduced by  boundary states that are linear
combinations of  $10$ (rather than $12$) Ishibashi states. This turns out to be correct:
if we label the necessary Ishibashi states in the various sectors of the bulk space \eqref{bulk} as
\begin{eqnarray}\label{eq:Ishibashistates}
{\cal H}_{\frac{-1}{24}} \ , \ {\cal H}_{\frac{35}{24}} & : &
| \tfrac{-1}{24}\rangle\!\rangle \ , \quad
| \tfrac{35}{24}\rangle\!\rangle  \ , \nonumber \\
{\cal H}_{\frac{1}{8}} \ , \ {\cal H}_{\frac{5}{8}} \ , \ {\cal H}_{\frac{1}{3}} & : & 
|  \tfrac{1}{8}, A \rangle\!\rangle  \ , \ 
|  \tfrac{1}{8}, B \rangle\!\rangle  \ , \ 
|  \tfrac{5}{8}, A \rangle\!\rangle \ , \ 
|  \tfrac{5}{8}, B \rangle\!\rangle \ , \ 
|  \tfrac{1}{3}, A \rangle\!\rangle \ , \
|  \tfrac{1}{3}, B \rangle\!\rangle \ , \label{Ishiused} \\
{\cal H}_0 & : &  
|  0, + \rangle\!\rangle   \ , \ 
|  0, - \rangle\!\rangle 
\ , \nonumber 
\end{eqnarray}
and assume --- this question will be addressed momentarily --- that their overlaps can be 
chosen to be (with $q=e^{2\pi i \tau}$)
\begin{align}\label{eq:Ishibashioverlaps}
  \langle\!\langle \tfrac{-1}{24} |\, q^{L_0 + \bar{L}_0}\, | \tfrac{-1}{24} \rangle\!\rangle
  &= \sqrt{3}\,\chi_{{\cal W}(\tfrac{-1}{24})}(q)
\nonumber \\
  \langle\!\langle \tfrac{35}{24} |\, q^{L_0 + \bar{L}_0}\, | \tfrac{35}{24} \rangle\!\rangle
  &= -\sqrt{3}\,\chi_{{\cal W}(\tfrac{35}{24})}(q) \nonumber \\
  \langle\!\langle \tfrac{1}{8},A |\, q^{L_0 + \bar{L}_0}\, | \tfrac{1}{8},A \rangle\!\rangle
  & = -\frac{i \tau}{3} \Bigl( \chi_{{\cal W}(\tfrac{1}{8})}(q)  -  2\, \chi_{{\cal W}(\tfrac{33}{8})}(q)  \Bigr)
\nonumber \\
\langle\!\langle \tfrac{1}{8},A |\, q^{L_0 + \bar{L}_0}\, | \tfrac{1}{8},B \rangle\!\rangle
& =  \frac{2}{\sqrt{3}} \Bigl( \chi_{{\cal W}(\tfrac{33}{8})}(q) + \chi_{{\cal W}(\tfrac{1}{8})}(q)  \Bigr)
\nonumber \\
\langle\!\langle \tfrac{5}{8},A |\, q^{L_0 + \bar{L}_0}\, | \tfrac{5}{8},A \rangle\!\rangle
& =  -\frac{i \tau2}{3} \Bigl( \chi_{{\cal W}(\tfrac{5}{8})}(q) - \tfrac{1}{2} \, \chi_{{\cal W}(\tfrac{21}{8})}(q) \Bigr)
\nonumber \\
\langle\!\langle \tfrac{5}{8},A |\, q^{L_0 + \bar{L}_0}\, | \tfrac{5}{8},B \rangle\!\rangle
& =  \frac{2}{\sqrt{3}} \Bigl( \chi_{{\cal W}(\tfrac{21}{8})}(q) + \chi_{{\cal W}(\tfrac{5}{8})}(q)  \Bigr)
 \\
\langle\!\langle \tfrac{1}{3},A |\, q^{L_0 + \bar{L}_0}\, | \tfrac{1}{3},A \rangle\!\rangle
& =  i \tau2\sqrt{3}
\Bigl(   \chi_{{\cal W}(\tfrac{1}{3})}(q)  -  \chi_{{\cal W}(\tfrac{10}{3})}(q)  \Bigr)
\nonumber \\
\langle\!\langle \tfrac{1}{3},A |\, q^{L_0 + \bar{L}_0}\, | \tfrac{1}{3},B \rangle\!\rangle
& =  \frac{2}{\sqrt{3}} \Bigl( \chi_{{\cal W}(\tfrac{1}{3})}(q) + \chi_{{\cal W}(\tfrac{10}{3})}(q)  \Bigr)
\nonumber \\
\langle\!\langle 0,\pm |\, q^{L_0 + \bar{L}_0}\, | 0 ,\pm \rangle\!\rangle
& = \frac12  \chi_{{\cal W}(0)}(q) = \frac12 \nonumber \\
\langle\!\langle 0,\pm |\, q^{L_0 + \bar{L}_0}\, |0,\mp \rangle\!\rangle
& =  \frac{1}{2}\Bigl(\chi_{{\cal W}(0)}(q)  + 
2 \chi_{{\cal W}(1)}(q) + 2 \chi_{{\cal W}(2)}(q) + 2 \chi_{{\cal W}(5)}(q) + 2 \chi_{{\cal W}(7)}(q) \Bigr)\ ,
\nonumber 
\end{align}
the annulus amplitudes are reproduced by 
\begin{align}\label{eq:boundary}
  \|\Wc(-\tfrac{1}{24})\rangle\!\rangle&=
  |\tfrac{-1}{24}\rangle\!\rangle-|\tfrac{35}{24}\rangle\!\rangle
  +3|\tfrac{1}{3},B\rangle\!\rangle
  +|\tfrac{5}{8},B\rangle\!\rangle-|\tfrac{1}{8},B\rangle\!\rangle \nonumber\\
  \|\Wc(\tfrac{35}{24})\rangle\!\rangle&=
  |\tfrac{-1}{24}\rangle\!\rangle-|\tfrac{35}{24}\rangle\!\rangle
  -3|\tfrac{1}{3},B\rangle\!\rangle
  +|\tfrac{5}{8},B\rangle\!\rangle-|\tfrac{1}{8},B\rangle\!\rangle
  \nonumber\\
  \|\Wc(\tfrac{1}{3})\rangle\!\rangle&=
  \frac{1}{2}\Big[|\tfrac{-1}{24}\rangle\!\rangle+|\tfrac{35}{24}\rangle\!\rangle\Big]
  +\frac{1}{2}|\tfrac{1}{3},A\rangle\!\rangle
  +\frac{1}{2}\Big[|\tfrac{5}{8},B\rangle\!\rangle+|\tfrac{1}{8},B\rangle\!\rangle\Big]
  -\frac{2}{\sqrt{3}}|0+\rangle\!\rangle\nonumber\\
  \|\Wc(\tfrac{10}{3})\rangle\!\rangle&=
  \frac{1}{2}\Big[|\tfrac{-1}{24}\rangle\!\rangle+|\tfrac{35}{24}\rangle\!\rangle\Big]
  -\frac{1}{2}|\tfrac{1}{3},A\rangle\!\rangle
  +\frac{1}{2}\Big[|\tfrac{5}{8},B\rangle\!\rangle+|\tfrac{1}{8},B\rangle\!\rangle\Big]
  +\frac{2}{\sqrt{3}}|0+\rangle\!\rangle\nonumber\\
  \|\Wc(\tfrac{5}{8})\rangle\!\rangle&=
  \frac{1}{3}\Big[|\tfrac{-1}{24}\rangle\!\rangle-|\tfrac{35}{24}\rangle\!\rangle\Big]
  +|\tfrac{1}{3},B\rangle\!\rangle
  +|\tfrac{5}{8},A\rangle\!\rangle+|\tfrac{1}{8},A\rangle\!\rangle
  -\frac{1}{12}\Big[|\tfrac{5}{8},B\rangle\!\rangle-|\tfrac{1}{8},B\rangle\!\rangle\Big]
  -|0-\rangle\!\rangle\nonumber\\
  \|\Wc(\tfrac{33}{8})\rangle\!\rangle&
  =\frac{1}{3}\Big[|\tfrac{-1}{24}\rangle\!\rangle-|\tfrac{35}{24}\rangle\!\rangle\Big]
  -|\tfrac{1}{3},B\rangle\!\rangle
  +|\tfrac{5}{8},A\rangle\!\rangle+|\tfrac{1}{8},A\rangle\!\rangle
  -\frac{1}{12}\Big[|\tfrac{5}{8},B\rangle\!\rangle-|\tfrac{1}{8},B\rangle\!\rangle\Big]
  +|0-\rangle\!\rangle\nonumber\\
  \|\Wc(\tfrac{1}{8})\rangle\!\rangle&=
  \frac{2}{3}\Big[|\tfrac{-1}{24}\rangle\!\rangle-|\tfrac{35}{24}\rangle\!\rangle\Big]
  -2|\tfrac{1}{3},B\rangle\!\rangle
  -|\tfrac{5}{8},A\rangle\!\rangle-|\tfrac{1}{8},A\rangle\!\rangle
  -\frac{5}{12}\Big[|\tfrac{5}{8},B\rangle\!\rangle-|\tfrac{1}{8},B\rangle\!\rangle\Big]
  -|0-\rangle\!\rangle\nonumber\\
  \|\Wc(\tfrac{21}{8})\rangle\!\rangle&=
  \frac{2}{3}\Big[|\tfrac{-1}{24}\rangle\!\rangle-|\tfrac{35}{24}\rangle\!\rangle\Big]
  +2|\tfrac{1}{3},B\rangle\!\rangle
  -|\tfrac{5}{8},A\rangle\!\rangle-|\tfrac{1}{8},A\rangle\!\rangle
  -\frac{5}{12}\Big[|\tfrac{5}{8},B\rangle\!\rangle-|\tfrac{1}{8},B\rangle\!\rangle\Big]
  +|0-\rangle\!\rangle\nonumber\\
  \|R^{(2)}(0,2)_7\rangle\!\rangle&=
  \frac{2}{3}\Big[|\tfrac{-1}{24}\rangle\!\rangle+|\tfrac{35}{24}\rangle\!\rangle\Big]
  +2\Big[|\tfrac{5}{8},A\rangle\!\rangle-|\tfrac{1}{8},A\rangle\!\rangle\Big]
  -\frac{1}{6}\Big[|\tfrac{5}{8},B\rangle\!\rangle+|\tfrac{1}{8},B\rangle\!\rangle\Big]
  \nonumber\\
  \|\Rc^{(2)}(2,7)\rangle\!\rangle&= 
 \|R^{(2)}(0,2)_7\rangle\!\rangle
  \nonumber\\
  \|\Rc^{(2)}(1,5)\rangle\!\rangle&=
  \frac{4}{3}\Big[|\tfrac{-1}{24}\rangle\!\rangle+|\tfrac{35}{24}\rangle\!\rangle\Big]
  -2\Big[|\tfrac{5}{8},A\rangle\!\rangle-|\tfrac{1}{8},A\rangle\!\rangle\Big]
  -\frac{5}{6}\Big[|\tfrac{5}{8},B\rangle\!\rangle+|\tfrac{1}{8},B\rangle\!\rangle\Big]
  \nonumber\\
  \|\Rc^{(2)}(2,5)\rangle\!\rangle&=
  |\tfrac{-1}{24}\rangle\!\rangle+|\tfrac{35}{24}\rangle\!\rangle
  +|\tfrac{1}{3},A\rangle\!\rangle
  -\frac{1}{2}\Big[|\tfrac{5}{8},B\rangle\!\rangle+|\tfrac{1}{8},B\rangle\!\rangle\Big]
  +\frac{2}{\sqrt{3}}|0+\rangle\!\rangle \ . 
\end{align}
This solution was found using the same technique as in
\cite[Sect.\,5.2]{Gaberdiel:2007jv},
{\it i.e.}\ by recursively constructing the boundary states using the knowledge of the
annulus partition functions.  
The $\Wc_{2,3}$-representations corresponding to the above boundary states 
generate all of $\Kdd$ (see \cite[Sect.\,2.4]{Gaberdiel:2009ug}); thus the boundary
state $\|\Rc\rangle\!\rangle$ for an arbitrary representation in $\Kdd$
is a linear combination of the above states with integral coefficients. We should also
stress that in the above boundary states we have only included components 
that have non-zero overlap with at least one other boundary state --- these are the only 
ones that can be detected in annulus amplitudes. The full boundary state must contain 
additional contributions, see Section~\ref{sec:bulk-int-kac}.

Up to now we have not used any knowledge about the bulk spectrum of the theory. 
Next we want to check whether the above  Ishibashi states (\ref{Ishiused}) (out of which our
boundary states were formed) actually exist within our bulk space $\Hc_\text{bulk}$.
It is not difficult to classify all Ishibashi states following \cite{Gaberdiel:2007jv},  and we find 
\begin{eqnarray} \label{Ishicounting}
{\cal H}_{\frac{-1}{24}} \ , \ {\cal H}_{\frac{35}{24}} & : & \hbox{1 Ishibashi state each } \nonumber \\
{\cal H}_{\frac{1}{8}} \ , \ {\cal H}_{\frac{5}{8}} \ , \ {\cal H}_{\frac{1}{3}} & : & 
\hbox{3  Ishibashi states each} \label{Ishigen} \\
{\cal H}_0 & : &  \hbox{9  Ishibashi states.} \nonumber
\end{eqnarray}
The Ishibashi states \eqref{Ishiused}
can be expressed in terms of linear combinations of the general Ishibashi states in 
\eqref{Ishicounting}. Furthermore, the structure of the bulk theory constrains the form 
of the corresponding overlaps, and this is perfectly compatible with \eqref{eq:Ishibashioverlaps}. 
This is a highly non-trivial consistency check of our analysis.

\section{The construction of the bulk space}\label{sec:bulk}

\subsection{From boundary to bulk --- the original construction}

Let us begin by reviewing how the original construction for the space of bulk fields
starting from a given boundary condition works. This construction was developed 
for  rational conformal field theories in \cite{Runkel:1998pm,Felder:1999mq,tft1,unique}, 
and then successfully applied to the logarithmic $\Wc_{1,p}$-models in
\cite{Gaberdiel:2006pp,Gaberdiel:2007jv} 
(for the logarithmic analogue of the
`charge-conjugation modular invariant', {\it i.e.}\ the Cardy case). 

The discussion in Sections~\ref{sec:boundaryfields}--\ref{sec:maximal-nondeg} is generic
and does not rely on $\Wc=\Wc_{2,3}$. On the other hand, the arguments from 
Section~\ref{sec:p0structure} onwards, are specific to $\Wc = \Wc_{2,3}$.

\subsubsection{The algebra of boundary fields}\label{sec:boundaryfields}

We want to explain more precisely what we mean by a consistent boundary theory. First we recall that 
for each $\Wc$-representation $U$ there exists a conjugate representation $U^*$ that is defined on 
the graded dual of $U$ \cite[Def.\,2.35]{Huang:2007ir}. The operation of taking the conjugate of 
a representation defines a contravariant functor from $\Rep(\Wc)$ to itself which we denote by 
$(-)^*$. 

With this definition we can then explain what we mean by a 
$\Wc$-{\em symmetric algebra of boundary fields} (or {\em boundary algebra} for short): it is a 
collection of data $(\Hc_\text{bnd}, m, \eta, \eps)$ where (see also \cite[Sect.\,1.2]{Gaberdiel:2007jv})
\begin{enumerate}
\item 
$\Hc_\text{bnd}$ is the space of boundary fields and forms a $\Wc$-representation.
\item 
$m \in \Hom(\Hc_\text{bnd} \otimes_f \Hc_\text{bnd} , \Hc_\text{bnd})$ is an intertwiner
describing the associative boundary OPE\footnote{
Here $\otimes_f$ denotes the fusion product. By construction of $\otimes_f$, for three 
$\Wc$-representations $R,S,U$, the vector space of intertwining operators 
$V(-,x) : R \otimes_\Cb S \to U$  (we suppress the formal variable and algebraic completion 
in the notation of the target vector space) is canonically isomorphic to
the space of $\Wc$-intertwiners $\Hom(R \otimes_f S,U)$, see
\cite[Def.\,3.18\,\&\,Prop\,4.15]{Huang:2007ir}.}.
That is, if we denote the corresponding intertwining operator by $V_m : \Hc_\text{bnd} \otimes_\Cb \Hc_\text{bnd} \to \Hc_\text{bnd}$,
for each $\psi \in \Hc_\text{bnd}$ we obtain a linear map $V_m(\psi,x)$ from $\Hc_\text{bnd}$ to (a completion of) $\Hc_\text{bnd}$, and
the OPE $\psi(x) \psi'(0)$ of two boundary fields $\psi, \psi' \in \Hc_\text{bnd}$ is 
written as $V_m(\psi,x)\psi'$.
\item $\eta : \Wc \rightarrow \Hc_\text{bnd}$ is an injective $\Wc$-intertwiner
such that $\eta(\Omega)$ is the identity field in $\Hc_\text{bnd}$. Since $\eta$ is injective,
the entire VOA $\Wc$ is contained in $\Hc_\text{bnd}$, as is required for a $\Wc$-symmetric
boundary condition.
\item 
$\eps : \Hc_\text{bnd} \rightarrow \Wc^*$ is a $\Wc$-intertwiner describing one-point functions of 
boundary fields. It has to give rise to a non-degenerate and symmetric two-point function on the boundary.
\end{enumerate}
We will sometimes just write $\Hc_\text{bnd}$ in place of $(\Hc_\text{bnd}, m, \eta, \eps)$.
Given a boundary algebra, the $n$-point function of boundary fields $\psi_1,\dots,\psi_n \in \Hc_\text{bnd}$ 
on the upper half plane is given by
\be
  \big\langle\, \psi_1(x_1) \cdots \psi_n(x_n) \,\big\rangle_{\!\text{uhp}}
  = \big\langle\, \eps\big( V_m(\psi_1,x_1) \cdots  V_m(\psi_{n-1},x_{n-1}) \psi_n \big)\,,\, \Omega \,\big\rangle \ ,
\ee
where we take $x_n=0$ and $x_1>x_2>\dots>x_n$.

\subsubsection{Bulk-boundary maps} \label{sec:bulk-boundary-map}

For the rational theories and the $\Wc_{1,p}$-models one could define 
the bulk space as a solution to a universal property. In order to formulate this
construction, we need the notion of a bulk-boundary map.

Fix a boundary algebra $\Hc_\text{bnd}$ and
let $\Hc$ be a $\Wc \otimes_\Cb \bar\Wc$-representation. 
By a {\em bulk-boundary map} we mean a linear map
$\beta(-,y) : \Hc \rightarrow \Hc_\text{bnd}$ (here $y>0$, and we again suppress the algebraic completion 
of the target from the notation) which satisfies the following two conditions:
\begin{itemize}
\item[(i)] {\em $\beta$ is compatible with the $\Wc$-symmetry:}
Given $\phi \in \Hc$, the element $\beta(\phi,y)$ is to be thought of as an insertion of the bulk
field $\phi$ on the upper half plane at position $iy$, expanded in terms of boundary fields
at position $0$. The compatibility condition is formulated as follows.
Write $\Hc$ as a quotient $\Hc = \hat\Hc/Q$, where $\hat\Hc = \bigoplus_{a} M_a \otimes_\Cb \bar N_a$ 
for some index set $a$ and $Q$ is a sub-representation of $\hat\Hc$. Denote by $\beta_a$ the composition
\be
\beta_a(-,y) = M_a \otimes_\Cb \bar N_a \hookrightarrow \hat\Hc 
\twoheadrightarrow \Hc \xrightarrow{\beta(-,y)} \Hc_\text{bnd}\ .
\ee
Then we require that
\be
  \beta_a(\phi \otimes_\Cb \phi' , y)
  = V_l(\phi,iy)V_r(\phi',-iy) \Omega \ ,
\labl{eq:bb-W-sym}
for appropriate intertwining operators 
$V_r(-,x) : N_a \times \Wc \rightarrow N_a$ and
$V_l(-,x) : M_a \times N_a \rightarrow \Hc_\text{bnd}$. Here, $V_r(-,x)\Omega$ can be 
taken to be just translation by $x$. In particular, the space of all maps $\beta_a$ is isomorphic 
to $\Hom(M_a \otimes_f N_a,\Hc_\text{bnd})$. We denote by 
$\tilde \beta_a \in \Hom(M_a \otimes_f N_a,\Hc_\text{bnd})$ the image of $\beta_a$ under that
isomorphism.

\item[(ii)] {\em $\beta$ is central:} Since $\beta(\phi,y)$ corresponds to a bulk field $\phi$ inserted at $iy$, an 
$n$-point correlator of boundary fields involving $\beta(\phi,y)$ should be continuous when a boundary 
field is taken past zero. This gives rise to the centrality condition: write $\Hc = \hat\Hc/Q$ as in (i). 
The map $\beta(-,y)$ is called {\em central} iff for all $\psi \in \Hc_\text{bnd}$
\be
  \lim_{x \searrow 0} V_m(\psi,x)\beta(\phi,y)
  =
  \lim_{x \searrow 0} V_m( \beta(\phi,y), x )\psi \ .
\labl{eq:bb-central}
This can also be expressed in terms of the braiding $c_{U,V} : U \otimes_f V \rightarrow V \otimes_f U$ 
in the category of $\Wc$-representations, using the component maps $\tilde \beta_a$ introduced in (i),
\be
  m \circ (\id_{\Hc_\text{bnd}} \otimes_f \tilde \beta_a) \circ ( c_{M_a,\Hc_\text{bnd}} \otimes_f \id_{N_a})
  =
  m \circ (\tilde \beta_a \otimes_f \id_{\Hc_\text{bnd}}) \circ (\id_{M_a} \otimes_f c_{\Hc_\text{bnd},N_a}) \ .
\labl{eq:bb-central-cat}
\end{itemize}

A graphical representation of condition \eqref{eq:bb-central-cat} is\footnote{Here we use the
usual graphical notation of braided tensor categories \cite{Joyal:1993}; 
our diagrams are read from bottom to top.}
\begin{align}
\label{eq:bb-central-cat-pic}
\begin{tikzpicture}[baseline=(eq)]
  \node (Mal) at (0,0) {$M_a$};
  \node (Hbl) at (1,0) {$\Hc_\text{bnd}$};
  \node (Nal) at (2,0) {$N_a$};
  \node (Htl) at (1,4) {$\Hc_\text{bnd}$};
  \node[shape=rectangle,minimum width=0.5cm] (mlinv) at (1,3) {\phantom{$m$}};
  \node[shape=rectangle,minimum width=0.5cm] (blinv) at (1.5,1.5) {\phantom{$\tilde\beta_a$}};
  \draw (Nal) .. controls +(up:8mm) and +(down:8mm) .. (blinv.south east);
  \draw (Hbl) .. controls +(up:8mm) and +(down:8mm) .. (0.5,1.6)
  ..controls +(up:8mm) and +(down:8mm) .. (mlinv.south west);
  \draw[-,line width=4pt,white] (Mal) .. controls +(up:8mm) and +(down:8mm) .. 
  (blinv.south west);
  \draw (Mal) .. controls +(up:8mm) and +(down:8mm) .. (blinv.south west);
  \draw (blinv.north) .. controls +(up:8mm) and +(down:8mm) .. (mlinv.south east);
  \draw (mlinv.north) .. controls +(up:8mm) and +(down:8mm) .. (Htl);
  \node[shape=rectangle,draw,minimum width=1cm] (bl) at (1.5,1.5) {$\tilde\beta_a$};
  \node[shape=rectangle,draw,minimum width=1cm] (ml) at (1,3) {$m$};
  \node (Mar) at (4,0) {$M_a$};
  \node (Hbr) at (5,0) {$\Hc_\text{bnd}$};
  \node (Nar) at (6,0) {$N_a$};
  \node (Htr) at (5,4) {$\Hc_\text{bnd}$};
  \node[shape=rectangle,minimum width=0.5cm] (mrinv) at (5,3) {\phantom{$m$}};
  \node[shape=rectangle,minimum width=0.5cm] (brinv) at (4.5,1.5) {\phantom{$\tilde\beta_a$}};
  \draw (Nar) .. controls +(up:8mm) and +(down:8mm) .. (brinv.south east);
  \draw (Mar) .. controls +(up:8mm) and +(down:8mm) .. (brinv.south west);
  \draw[-,line width=4pt,white] (Hbr) .. controls +(up:8mm) and +(down:8mm) .. (5.5,1.6)
  ..controls +(up:8mm) and +(down:8mm) .. (mrinv.south east);
  \draw (Hbr) .. controls +(up:8mm) and +(down:8mm) .. (5.5,1.6)
  ..controls +(up:8mm) and +(down:8mm) .. (mrinv.south east);
  \draw (brinv.north) .. controls +(up:8mm) and +(down:8mm) .. (mrinv.south west);
  \draw (mrinv.north) .. controls +(up:8mm) and +(down:8mm) .. (Htr);
  \node[shape=rectangle,draw,minimum width=1cm] (br) at (4.5,1.5) {$\tilde\beta_a$};
  \node[shape=rectangle,draw,minimum width=1cm] (mr) at (5,3) {$m$};
  \node (eq) at (3,1.5) {$=$};
\end{tikzpicture}
\end{align}
We will sometimes also refer to $\tilde \beta_a$ as a `bulk-boundary map'. In the non-logarithmic 
rational case, bulk-boundary maps where first systematically studied in \cite{Cardy:1991tv,Lewellen:1991tb}; 
the centrality condition amounts to \cite[Fig.\,9d]{Lewellen:1991tb}. In the framework of vertex 
operator algebras and intertwining operators, the bulk-boundary map and its categorical 
expression \eqref{eq:bb-central-cat} have appeared in \cite[Sect.\,1.8\,\&\,3.2]{Kong2006b}.

\subsubsection{Characterising the bulk space} \label{sec:maximality-condition}

After these definitions we can now explain how the bulk space of states could be 
characterised for the rational theories  \cite{Runkel:1998pm,Felder:1999mq,tft1,unique}, 
and for the $\Wc_{1,p}$-models in
the case of the `charge-conjugation modular invariant'  \cite{Gaberdiel:2006pp,Gaberdiel:2007jv}, 
the logarithmic analogue of the Cardy case

Let us fix a boundary algebra $\Hc_\text{bnd}$. We consider pairs $(\Hc,\beta)$, where $\Hc$ is a 
$\Wc \otimes_\Cb \bar\Wc$-representation and $\beta : \Hc \rightarrow \Hc_\text{bnd}$ is a 
bulk-boundary map. By an arrow $f : (\Hc,\beta) \rightarrow (\Hc',\beta')$ between two such pairs we 
mean a $\Wc \otimes_\Cb \bar\Wc$-intertwiner $f : \Hc \rightarrow \Hc'$ such that the diagram
\be
\raisebox{-2em}{
\begin{xy} 
(0,15)*+{\Hc}="a"; (25,15)*+{\Hc'}="b";%
(12,0)*+{\Hc_\text{bnd}}="c";%
{\ar@{->} "a";"b"}?*!/_2mm/{f}; 
{\ar@{->} "a";"c"}?*!/_2mm/{\beta}; 
{\ar@{->} "b";"c"}?*!/^2mm/{\beta'}; 
\end{xy}}
\labl{eq:bulk-space-terminal-arrows}
commutes. The {\em space of bulk fields} $\Hc_\text{bulk}$ and the {\em bulk-boundary OPE} 
$\beta_{\text{bb}}(-,y)$ for the boundary algebra $\Hc_{\text{bnd}}$ are then defined to be a terminal 
object in the category  formed by these pairs and arrows. In other words, the pair 
$(\Hc_\text{bulk},\beta_{\text{bb}})$ has the property 
that for all other pairs $(\Hc',\beta')$ there exists a unique $\Wc \otimes_\Cb \bar\Wc$-intertwiner 
$f : \Hc' \rightarrow \Hc_\text{bulk}$ such that we have
$\beta' = \beta_{\text{bb}} \circ f : \Hc' \rightarrow \Hc_\text{bnd}$.

A terminal object need not exist, but if it does, it is unique up to unique isomorphism. 
Indeed, for another terminal object $(\Hc_\text{bulk}',\beta'_{\text{bb}})$, there are unique intertwiners 
$\Hc_\text{bulk} \rightarrow \Hc_\text{bulk}'$ and $\Hc_\text{bulk}' \rightarrow \Hc_\text{bulk}$ 
which -- again by uniqueness -- have to compose to the identity. 

The bulk-boundary OPE $\beta_{\text{bb}}(-,y)$ is necessarily injective. To see this, 
note that if $f : M \rightarrow \Hc_\text{bulk}$ is any $\Wc \otimes_\Cb \bar\Wc$-intertwiner 
such that $\beta_{\text{bb}} \circ f = 0$, then $f$ is an arrow from 
$(M,0)$ to $(\Hc_\text{bulk},\beta)$. As $0$ is also an arrow between these pairs, 
by uniqueness we have $f=0$. This shows that the present formulation is equivalent to 
the one given in \cite[Sect.\,3.1]{Gaberdiel:2007jv}.

For the actual computations the simplest case (and in many situations the only
tractable case) is to take as the boundary condition the `identity' boundary condition,
{\it i.e.}\ the boundary condition whose space of states just consists of the $\Wc$-algebra itself.
Such a boundary condition exists for the $\Wc_{1,p}$-models \cite{Gaberdiel:2006pp,Gaberdiel:2007jv}. 
In that case, the objects of the category are $(\Hc,\beta$), where $\beta$ is a bulk-boundary map 
$\Hc\rightarrow \Wc$. Since the image lies in $\Wc$, the centrality condition (ii) in the definition of the 
bulk-boundary map is trivial, and the problem simplifies considerably. This approach was successfully 
applied to the $\Wc_{1,p}$-models in \cite{Gaberdiel:2007jv}. 

\subsection[The proposal for the bulk space of the $\Wc_{2,3}$-model]{The proposal for the bulk space of the $\boldsymbol{\Wc_{2,3}}$-model}\label{proposal}

Next we want to repeat a similar analysis for the case of the ${\cal W}_{2,3}$-model. 
Unfortunately, as was already mentioned earlier, the ${\cal W}_{2,3}$-model does not have
an `identity brane' whose space of states only consists of $\Wc_{2,3}$ itself. Thus we
cannot directly apply the same method that worked for the $\Wc_{1,p}$-models.
One could obviously try to apply the general formalism to a different boundary condition,
but then the centrality condition (ii) plays an important role, and we have not been able
to characterise the most general bulk-boundary maps in any useful way.

Instead we shall proceed slightly differently. By assumption, any $\Wc$-symmetric
boundary condition has $\Wc$ as a subalgebra in its space of boundary fields $\Hc_{\text{bnd}}$.
Since $\Hc_{\text{bnd}}$ is self-conjugate, it follows that $\Wc^\ast$ can be obtained 
as a quotient of $\Hc_{\text{bnd}}$. It seems reasonable to assume that there is at least 
one boundary condition, for which
the disc correlator of a bulk and a boundary field is already non-degenerate in the bulk entry 
if we take the boundary insertions only from $\Wc$. (This is the natural analogue of the `identity 
brane case'.) The two-point functions on the boundary define a natural pairing between $\Wc$ and 
$\Wc^\ast$. The above supposition therefore implies that the bulk-boundary map is still injective after 
projecting from $\Hc_{\text{bnd}}$ to $\Wc^\ast$. Obviously, for the 
$\Wc_{1,p}$-models $\Wc^\ast\cong \Wc$, and the distinction between $\Wc$ and $\Wc^\ast$ is 
immaterial. 

We thus propose that the natural analogue of the `identity brane' situation is 
to consider bulk-boundary maps $\beta:\Hc\rightarrow \Wc^*$ whose image is $\Wc^*$, rather 
than $\Wc$, and that the terminal
object within this category is our desired bulk space 
$(\Hc_{\text{bulk}},\beta_{\text{bb}})$. As we shall explain in Section~\ref{properties}, 
it satisfies a number of fairly non-trivial consistency conditions.

\subsubsection{Characterising the terminal object}\label{sec:maximal-nondeg}

As advertised above, let us consider the category of pairs $(\Hc,\beta)$, where $\Hc$ is a 
$\Wc \otimes_\Cb \bar\Wc$-representation, and $\beta: \Hc \rightarrow \Wc^*$ is compatible with 
the $\Wc$-symmetry, that is, it satisfies condition (i) of a bulk-boundary map as stated in 
Section~\ref{sec:bulk-boundary-map}. The arrows are intertwiners 
$f : \Hc \rightarrow \Hc'$ which make the diagram
\be
\raisebox{-2em}{
\begin{xy} 
(0,15)*+{\Hc}="a"; (25,15)*+{\Hc'}="b";%
(12,0)*+{\Wc^*}="c";%
{\ar@{->} "a";"b"}?*!/_2mm/{f}; 
{\ar@{->} "a";"c"}?*!/_2mm/{\beta}; 
{\ar@{->} "b";"c"}?*!/^2mm/{\beta'}; 
\end{xy}}
\labl{eq:bulk-space-terminal-W*}
commute.
We want to find a terminal object $(\Hc_{\text{bulk}},\beta_{\text{bb}})$ in this category.
As in Section~\ref{sec:bulk-boundary-map}, the terminal object need not exist, but if it does 
it is unique and $\beta_{\text{bb}}$ is injective. 
\medskip

It was argued in \cite[Sect.\,5.1]{Gaberdiel:2009ug} that 
there is an isomorphism 
$\Hom(U,V^*) \rightarrow \Hom(U \otimes_f V, \Wc^*)$.
The image of the 
natural isomorphism $U \rightarrow U^{**}$ under this map provides a non-degenerate pairing 
\be
  ev_{U} : U \otimes_f U^* \longrightarrow \Wc^*
\ee
for all representations $U$. Let $V_{ev_U}(-,x) : U \times U^* \rightarrow \Wc^*$ be the 
corresponding intertwining operator. Non-degeneracy means that for all $u \in U$ there is a 
$u' \in U^*$ such that $V_{ev_U}(u,x)u' \neq 0$, and vice versa.

Suppose that $\Wc$ has a finite number of irreducible representations and that each of 
these has a projective cover $\Pc(a)$. Let us denote by $\hat\Hc$ the space
$\hat\Hc = \bigoplus_a \Pc(a) \otimes_\Cb \bar \Pc(a)^*$,
where the direct sum is over all irreducible representations,
and define the map $\hat\beta = \bigoplus_a V_{ev_{\Pc(a)}} : \hat\Hc \rightarrow \Wc^*$. 
Let $\Hc_{\text{bulk}} = \hat\Hc/ \ker(\hat\beta)$ and let 
$\beta_{\text{bb}} : \Hc_{\text{bulk}} \rightarrow \Wc^*$ 
be the map induced by $\hat\beta$ on the quotient. We claim
\begin{quote}
   $(\Hc_{\text{bulk}},\beta_{\text{bb}})$ is a terminal object  in the above category.
\end{quote}
The proof proceeds analogously to the one in \cite[Sect.\,3.5]{Gaberdiel:2007jv}. 
Every representation $\Hc'$ can be written as
\be\label{ansat}
 \Hc' = \hat  \Hc' /Q'
 \qquad ; \quad \hat \Hc' = \bigoplus_{a,b} N_{ab} \, \Pc(a) \otimes_\Cb \bar \Pc(b)  \ ,
\ee
where $N_{ab}$ are multiplicities and $Q'$ is a sub-representation of the direct sum. 
Suppose that $(\Hc',\beta')$ is a pair in our category, where we write 
$\Hc'$ as in (\ref{ansat}). 
Denote the corresponding projection by $\pi' : \hat \Hc' \rightarrow \Hc'$. If we define 
$\hat \beta' = \beta' \circ \pi'$, then $\pi'$ describes an arrow from 
$(\hat \Hc', \hat \beta')$ to $(\Hc',\beta')$. 

Our first aim is to construct an arrow 
$(\hat \Hc', \hat \beta')\rightarrow (\Hc_{\text{bulk}},\beta_{\text{bb}})$. 
To this end we denote the restriction of $\hat \beta'$ to the $k$'th summand 
$\Pc(a) \otimes_\Cb \Pc(b) \hookrightarrow \hat \Hc'$ by $\hat \beta_k'$. By definition 
of $ev_{\Pc(a)}$, there exist intertwiners $f_{a,b,k} : \Pc(b) \rightarrow \Pc(a)^*$ such that
\be
  \hat\beta_k' = V_{ev_{\Pc(a)}} \circ (\id_{\Pc(a)} \otimes_\Cb f_{a,b,k}) \ .
\ee
Then \smash{$f = \bigoplus_a \sum_{b,k} \id_{\Pc(a)} \otimes_\Cb f_{a,b,k} : \hat \Hc' \rightarrow \hat\Hc$}
provides an arrow $f : (\hat \Hc', \hat \beta') \rightarrow (\hat\Hc,\hat\beta)$. 
We compose this arrow with the projection
$(\hat\Hc,\hat\beta) \rightarrow (\Hc_{\text{bulk}},\beta_{\text{bb}})$ to give the desired arrow
$f' : (\hat \Hc', \hat \beta') \rightarrow (\Hc_{\text{bulk}},\beta_{\text{bb}})$. 

Next we want to show that this gives rise to an arrow on $(\Hc',\beta')$. First we note
that $Q'$ is equal to $\ker(\pi')$, and thus also $\hat \beta'$ vanishes on $Q'$. 
But $\hat \beta' = \beta_{\text{bb}} \circ f'$ (as $f' :\hat \Hc' \rightarrow \Hc_{\text{bulk}}$ is an arrow in our
category), and since $\beta_{\text{bb}}$ is injective, $f'$ vanishes on $Q'$. Thus there is a 
well-defined map $g : \Hc' \rightarrow\Hc_{\text{bulk}}$ on the quotient, which provides the sought-after 
arrow $(\Hc',\beta') \rightarrow (\Hc_{\text{bulk}},\beta_{\text{bb}})$.
Since $\beta_{\text{bb}}$ is injective, this arrow is unique.
\medskip

It remains to find a practical way to compute the kernel $\Nc = \ker \hat\beta$ of 
$\hat\beta : \hat\Hc \rightarrow \Wc^*$ to obtain $\Hc_{\text{bulk}}$ as a quotient. This can be done by 
the construction in  \cite[Sect.\,3.3]{Gaberdiel:2007jv}, which we briefly review. There we 
described $\Nc$ as the image of all possible maps of projectives 
$\Pc(i) \otimes_\Cb \bar \Pc(j)$, for all $i,j$, into $\Nc$. For fixed $i,j$ such a map 
takes the form $h : \Pc(i) \otimes_\Cb \bar \Pc(j) \rightarrow \hat\Hc$, 
\smash{$h = \bigoplus_a f_a \otimes_\Cb \bar g_a^*$}, where 
$f_a : \Pc(i) \rightarrow \Pc(a)$ and 
$g_a : \Pc(a) \rightarrow \Pc(j)^*$. Since
\be\label{2.17}
  \hat\beta \circ h = \bigoplus_a V_{ev_{\Pc(a)}} \circ (f_a \otimes_\Cb \bar g_a^*) 
   = \sum_a V_{ev_{\Pc(i)}} \circ \big(\id_{\Pc(i)}  \otimes_\Cb (\overline{g_a \circ f_a})^*\big)\ ,
\ee
$\hat\beta \circ h$ is zero if and only if  \smash{$\sum_a g_a \circ f_a = 0$}. Since 
$\Nc$ itself is a quotient of a direct sum of 
$\Pc(i) \otimes_\Cb \Pc(j)$ with some multiplicities, the above construction 
produces the entire kernel.

If $\Pc(j)^* \cong \Pc(j)$, a generic example of such maps $f_a$ and $g_a$ is found as follows. 
Let $e : \Pc(i) \rightarrow \Pc(j)$ be any morphism. Then
\be\label{eq:seesaw}
  \big(\Pc(i) \xrightarrow{e} \Pc(j) \xrightarrow{\id} \Pc(j) \cong \Pc(j)^*\big)
  ~+~
  \big(\Pc(i) \xrightarrow{\id} \Pc(i) \xrightarrow{-e} \Pc(j) \cong \Pc(j)^*\big)
  = 0 \ .
\ee
Thus, the image of $\Pc(i) \otimes_\Cb \Pc(j)^*$ under 
$e \otimes_\Cb \id - \id \otimes_\Cb \bar e^*$ lies in $\Nc$.

\subsubsection{Projective covers}\label{sec:p0structure}

We now restrict our attention again to the case $\Wc = \Wc_{2,3}$.  
As we have just seen, the construction of $\Hc_\text{bulk}$ relies
on the existence of projective covers for all irreducible representations. 
For $\Wc(h)$ with $h$  taking all values in \eqref{eq:W23-irreps} except $h=0$, these were given
in \cite{Gaberdiel:2009ug} to be \eqref{eq:proj-cover-list}. The argument in \cite{Gaberdiel:2009ug} 
was based on the assumption that  the projective cover of the irreducible representation 
$\Wc(\frac{-1}{24})$ is $\Wc(\frac{-1}{24})$ itself.
Hence for every representation $U$ that has a dual, also $\Wc(\frac{-1}{24}) \otimes_f U$ is projective 
(see Appendix \ref{app:P0xP0} for details). By the fusion rules in \cite[App.\,A.4]{Gaberdiel:2009ug}, this 
implies that the indecomposable representations listed in \eqref{eq:proj-cover-list} are necessarily 
projective. From their composition series as given in Appendix~\ref{sec:embedding} one sees that 
they provide the projective covers of all $\Wc(h)$ except for $h=0$. 

Regarding $\Pc(0)$, one first notes that
the only representations $R$ in \eqref{eq:W23-irreps} and \eqref{eq:indec-W-rep} which 
allow for a surjection $R \rightarrow \Wc(0)$ are $\Wc(0)$, $\Wc$ and $\Qc$. It is verified in 
\cite[App.\,A.3]{Gaberdiel:2009ug} that these are not projective. On the other hand, in 
\cite[Thm.\,3.24]{Huang:2007mj} it is stated that under certain conditions on $\Wc$, 
including $C_1$-cofiniteness, every irreducible $\Wc$-representation has a projective cover. 
We do not know if $\Wc_{2,3}$ satisfies these conditions, but it seems plausible to us that it 
does (see also \cite{AM1}). We will hence assume that $\Wc(0)$ does have a projective cover. Using these
assumptions as well as the result of the recent calculation of Zhu's algebra \cite{AM}
it follows that $\Pc(0)$ has the composition series
\begin{equation}\label{eq:P0st}
 \Wc(0) \rightarrow \Wc(1)\oplus\Wc(2) \rightarrow \Wc(0)\oplus2\Wc(5)\oplus2\Wc(7) 
  \rightarrow \Wc(1)\oplus\Wc(2)
  \rightarrow \Wc(0) \ .
\end{equation}
The derivation is given in Appendix~\ref{app:P0-cover}, and the corresponding 
embedding diagram can be found in Appendix~\ref{sec:embedding}.

\subsubsection{Computing the Kernel}\label{sec:constr_kernel}

Next we want to work out $(\Hc_{\text{bulk}},\beta_{\text{bb}})$ as defined in 
the previous subsection explicitly. Recall from section~\ref{sec:maximal-nondeg} 
that $\hat\Hc$ is defined as
$\hat\Hc=\bigoplus_a \Pc(a)\otimes_{\mathbb{C}}\bar{\Pc}(a)^\ast$.
We want to determine the quotient space $\Hc_{\text{bulk}} = \hat\Hc / \text{ker}(\hat\beta)$, where 
\begin{align}
\hat\beta=\bigoplus_a V_{ev_{\Pc(a)}}:\hat\Hc\rightarrow \Wc^\ast\ .
\end{align}
First we will show that one may choose representatives of vectors in
$\Hc_{\text{bulk}} $ in 
\begin{equation}\label{H(0)}
\Hc_{\text{bulk}} ^{(0)} =  \bigoplus_a \Wc(a)\otimes_{\mathbb{C}}\bar{\Pc}(a)^\ast \  ,
\end{equation}
where $\Wc(a)$ is the top factor in the composition series of $\Pc(a)$. 
(This is very similar to what happened in  \cite{Quella:2007hr} and \cite{Gaberdiel:2007jv}.)  
Consider some element 
$\Wc(i)\otimes_{\mathbb C} \bar\Wc(j) \subset \Pc(a)\otimes_{\mathbb C} \bar{\Pc}(a)^\ast$ in the composition
series of $\hat\Hc$. Since $\Pc(i)$ is the projective cover of $\Wc(i)$, it follows that there
exists an intertwiner $e: \Pc(i)\rightarrow  \Pc(a)$ such that the top factor $\Wc(i)$ in the composition 
series of $\Pc(i)$ is mapped to the given $\Wc(i)\subset \Pc(a)$. By \eqref{eq:seesaw} it then 
follows that the image of 
\begin{equation}\label{schaukel}
e\otimes_{\mathbb{C}}{\rm id}_{\bar{\Pc}(a)}   - {\rm id}_{\Pc(i)}\otimes_{\mathbb{C}}\bar e^\ast
\end{equation}
lies in the kernel of $\hat\beta$. Thus in $\Hc_{\text{bulk}}$,  
any vector in 
$\Wc(i)\otimes_{\mathbb C} \bar\Wc(j) \subset \Pc(a)\otimes_{\mathbb C} \bar{\Pc}(a)^\ast$,
where $\Wc(i)$ is not the top component of $\Pc(a)$,
lies in the same equivalence class as a vector in 
$\Wc(i)\otimes_{\mathbb C} \bar e^*(\bar\Wc(j)) \subset \Pc(i)\otimes_{\mathbb C} \bar{\Pc}(i)^\ast$,
where $\Wc(i)$ is the top component of $\Pc(i)$.
Thus we have shown that $\Hc_{\text{bulk}} $ is at most as big as $\Hc_{\text{bulk}} ^{(0)}$. 

\smallskip

We can collect the terms of $\hat\Hc$ into sectors
\begin{align}\label{hatHcdec}
\hat\Hc
= \hat\Hc_{0} \oplus \hat\Hc_{\frac{1}{8}} \oplus \hat\Hc_{\frac{5}{8}} \oplus
\hat\Hc_{\frac{1}{3}} \oplus \hat\Hc_{\frac{35}{24}} \oplus  \hat\Hc_{\frac{-1}{24}}  \  ,
\end{align}
with
\begin{align}\label{eq:seesawmaps}
\hat\Hc_{\frac{-1}{24}}&=\Wc(\tfrac{-1}{24})\otimes_{\mathbb{C}}\bar{\Wc}(\tfrac{-1}{24})^\ast&
\hat\Hc_{\frac{35}{24}}&=\Wc(\tfrac{35}{24})\otimes_{\mathbb{C}}\bar{\Wc}(\tfrac{35}{24})^\ast \nonumber \\[6pt]
\hat\Hc_{\frac13}&=\bigoplus_{a=\frac13,\frac{10}{3}}\Pc(a)\otimes_{\mathbb{C}}\bar{\Pc}(a)^\ast&
\hat\Hc_{\frac58}&=\bigoplus_{a=\frac58,\frac{21}{8}}\Pc(a)\otimes_{\mathbb{C}}\bar{\Pc}(a)^\ast \\
\hat\Hc_{\frac18}&=\bigoplus_{a=\frac18,\frac{33}{8}}\Pc(a)\otimes_{\mathbb{C}}\bar{\Pc}(a)^\ast&
\hat\Hc_{0}&=\bigoplus_{a=0,1,2,5,7}\Pc(a)\otimes_{\mathbb{C}}\bar{\Pc}(a)^\ast \ .\nonumber
\end{align}
It is easy to see that the kernel of $\hat\beta$ does not mix the different sectors in 
(\ref{hatHcdec}), and we may therefore consider each of them in turn. 
It follows from the 
arguments of \cite[Sect.\,4.3\,\&\,App.\,D]{Gaberdiel:2007jv} that in all sectors, except possibly for $\hat\Hc_{0}$, 
all relations of the quotient space $\Hc_{\text{bulk}}  = \hat\Hc / \text{ker}(\hat\beta)$ are taken into account
by \eqref{schaukel}
and that $\Hc_{\text{bulk}}$ is isomorphic (as an 
$(L_0^\text{diag},\bar L_0^\text{diag})$-graded vector space, see \eqref{eq:Hbulk-decomp-WxP}) to 
$\Hc_{\text{bulk}} ^{(0)}$. We believe that this will also be the case for $\hat\Hc_{0}$, 
but we have no proof.
We therefore conjecture that $\Hc_{\text{bulk}} ^{(0)}$ is not just an upper
bound for $\Hc_{\text{bulk}} $, but actually isomorphic to it.
This gives then the explicit 
description of $\Hc_{\text{bulk}}$ that was advertised before. (The bulk-boundary map
$\beta_{\text{bb}}$ is the one induced from $\hat\beta$ under the quotient map 
$\hat\Hc \twoheadrightarrow \Hc_{\text{bulk}}$.)

\subsection{Properties of the resulting bulk space}\label{properties}

Next we want to describe the resulting bulk space $\Hc_{\text{bulk}}$ in some more detail. 
First we want to explain how $\Wc\otimes_{\mathbb{C}}\bar{\Wc}$ acts on it. Given that 
we have a description of $\Hc_{\text{bulk}}$ as a quotient  by $\Nc$ of $\hat\Hc$, this can now be 
easily deduced. Since $\hat\Hc$ has a decomposition as in (\ref{hatHcdec}), we get a similar 
decomposition for $\Hc_{\text{bulk}}$, which we write as 
\begin{align}\label{Hcdec}
\Hc_{\text{bulk}} = \hat\Hc / \Nc 
= \Hc_{0} \oplus \Hc_{\frac{1}{8}} \oplus \Hc_{\frac{5}{8}} \oplus
\Hc_{\frac{1}{3}} \oplus \Hc_{\frac{-1}{24}} \oplus  \Hc_{\frac{35}{24}}  \  .
\end{align}
In order to describe the resulting structure we analyse, sector by sector, their composition 
series, following the method outlined in Section~\ref{sec:intro-bulkbnd}. This is to say, we 
identify first the largest direct sum of irreducible subrepresentations; then we quotient by 
these and find the largest direct sum of irreducible subrepresentations in the quotient, {\it etc}. 
The situation is obviously simplest for $\Hc_{- {1}/{24}}$ and $\Hc_{{35}/{24}}$ since 
they are already, by themselves, irreducible. As a consequence, the quotient is trivial
for these sectors, and the resulting composition series just consists of one term.

\subsubsection{The edge of the Kac table}\label{sec:compedge}

The situation is more interesting for those sectors that come from the `edge of the
Kac table', {\it i.e.}\ the sectors $\Hc_{1/8}$,   $\Hc_{1/3}$ and
$\Hc_{5/8}$. In the following we shall only show the calculation for the 
first case ($h=\tfrac{1}{8}$) --- the other cases then follow upon 
replacing $\{\tfrac18,\tfrac{33}{8}\}$ with $\{\tfrac58,\tfrac{21}{8}\}$ or 
$\{\tfrac13,\tfrac{10}{3}\}$, respectively.

As we have explained above, the representatives are described by $\Hc^{(0)}_{1/8}$,
see eq.\ (\ref{H(0)}). It is then not difficult to show that the 
maximal fully reducible subrepresentation of $M_1(\frac18)=\Hc_{{1}/{8}}$ is given by 
\begin{equation}
R_1(\tfrac{1}{8}) =\bigoplus_{i=\frac{1}{8},\frac{33}{8}}\Wc(i)\otimes_{\mathbb{C}}\bar{\Wc}(i) \ ,
\end{equation}
where $\bar{\Wc}(i)$ comes from the bottom entry in $\bar{\Pc}(i)$. 

Next, we consider the quotient $M_2(\frac18)=M_1(\frac18) / R_1(\frac18)$, and repeat the
analysis. The maximal fully reducible subrepresentation of $M_2(\frac18)$ is 
\begin{align}
R_2(\tfrac{1}{8})=2\, \Wc(\tfrac18)\otimes_{\mathbb{C}}\bar{\Wc}(\tfrac{33}{8})\; \oplus \; 
2\, \Wc(\tfrac{33}{8})\otimes_{\mathbb{C}}\bar{\Wc}(\tfrac{1}{8}) \ .
\end{align}
These arise from the middle lines in the embedding diagrams of the corresponding
projective covers $\bar{\Pc}(i)$. Finally, the maximal fully reducible subrepresentation 
of $M_3(\frac18)=M_2(\frac18) /R_2(\frac18)$ equals
\begin{align}
R_3(\tfrac{1}{8})=\bigoplus_{i=\frac{1}{8},\frac{33}{8}}\Wc(i)\otimes_{\mathbb{C}}\bar{\Wc}(i)\ ,
\end{align}
which comes from the top entries of $\bar{\Pc}(i)$. Thus we obtain precisely the 
composition series given in \eqref{eq:1/8_compseries}.

\subsubsection{The interior of the Kac table}

This leaves us with the sector $\Hc_0$, for which the analysis is more complicated.
Proceeding as before we find that the maximal fully reducible subrepresentation 
of $M_1(0)=\Hc_0$ is given by
\begin{align}\label{eq:compseries_first_factor}
R_1(0)&=\bigoplus_{i=0,1,2,5,7}\Wc(i)\otimes_{\mathbb{C}}\bar\Wc(i)\ .
\end{align}
As before, these states arise from the bottom component in $\bar{\Pc}(i)^\ast$.

Next, we find that the maximal fully reducible subrepresentation of $M_2(0)=M_1(0)/R_1(0)$ is
\begin{align}
R_2(0)&=\bigl(\Wc(1)\oplus\Wc(2)\bigr)\otimes_{\mathbb{C}}
    \bigl(2\, \bar{\Wc}(5)\oplus2\, \bar{\Wc}(7)\oplus\bar{\Wc}(0)\bigr) \nonumber \\
   &\quad \oplus \bigl(\Wc(5)\oplus\Wc(7)\bigr)\otimes_{\mathbb{C}}
     \bigl(2\, \bar{\Wc}(2)\oplus2\, \bar{\Wc}(1)\bigr)\\
   &\quad \oplus\Wc(0)\otimes_{\mathbb{C}}\bigl(\bar{\Wc}(1)\oplus \bar{\Wc}(2)\bigr)\ , \nonumber
\end{align}
while that of $M_3(0)=M_2(0)/R_2(0)$ is
\begin{align}
R_3(0)&=\Wc(1)\otimes_{\mathbb{C}}\bigl(2\, \bar{\Wc}(1)\oplus4\, \bar{\Wc}(2)\bigr) \nonumber\\
     &\quad \oplus\Wc(2)\otimes_{\mathbb{C}}\bigl(2\, \bar{\Wc}(2)\oplus\, 4 \bar{\Wc}(1)\bigr)\nonumber\\
      &\quad \oplus\Wc(0)\otimes_{\mathbb{C}}
          \bigl(\bar{\Wc}(0)\oplus2\, \bar{\Wc}(5)\oplus2\, \bar{\Wc}(7) \bigr)\\
      &\quad \oplus\Wc(5)\otimes_{\mathbb{C}}
          \bigl(2\, \bar{\Wc}(0)\oplus 2\, \bar{\Wc}(5)\oplus4\, \bar{\Wc}(7)\bigr)\nonumber\\
      &\quad \oplus\Wc(7)\otimes_{\mathbb{C}}
          \bigl(2\, \bar{\Wc}(0)\oplus2\, \bar{\Wc}(7)\oplus4\, \bar{\Wc}(5)\bigr)\nonumber \ . 
\end{align}
Similarly, the maximal fully reducible subrepresentation of $M_4(0)=M_3(0)/R_3(0)$ is
$R_4(0) = R_2(0)$, while that of $M_5(0)=M_4(0)/R_4(0)$ equals $R_5(0) = R_1(0)$.
This agrees precisely with what we claimed in \eqref{eq:0_compseries}.

\subsubsection{Consistency conditions}\label{consistency}

\noindent Finally, let us comment on the consistency conditions our answer satisfies.
None of the properties discussed below is built into our ansatz from the start, and they give strong 
support to our proposed bulk space.

\smallskip

\noindent {\em Modular invariance:}
Given the composition series \eqref{eq:irred-sector-comp-series}, \eqref{eq:1/8_compseries} 
and \eqref{eq:0_compseries}, it is straightforward to work out the 
partition function of $\Hc_{\text{bulk}}$, resulting in
\begin{equation}
\begin{array}{ll}
Z(q) & = \sum_i \chi_{\Wc(i)}(q) \cdot \chi_{\Pc(i)}(\bar q)
\\[.3em]\displaystyle  
 & = (q \bar q)^{-1/24} + 3 + 2(q \bar q)^{1/8} + 2(q \bar q)^{1/3} 
 + (q{+}\bar q)\cdot(q \bar q)^{-1/24} + 2(q{+}\bar q) + \cdots \ ,
\end{array}\label{eq:ansatz} 
\end{equation}
where the sum runs over all 13 irreducibles (including $\Wc(0)$).
This partition function was already proposed in \cite[Sect.\,4]{Gaberdiel:2009ug}, at least
up to the multiplicity of $(q\bar q)^0$, which was left as an arbitrary positive integer there. 
As was mentioned in \cite{Gaberdiel:2009ug}, $Z(q)$ is modular invariant (see also 
\cite{Feigin:2006iv}).

\medskip\noindent {\em Self-conjugacy:}
The space of bulk states $\Hc_{\text{bulk}}$ is actually isomorphic to
its conjugate representation, $\Hc_{\text{bulk}}^\ast$. This is necessary in order for the 
bulk theory to have a non-degenerate (bulk) two-point function. 

\medskip\noindent {\em Stress tensor:}
Curiously, the composition series \eqref{eq:0_compseries} of $\Hc_{\text{bulk}}$ shows that
$\Hc_{\text{bulk}}$ does {\em not} contain $\Wc \otimes_\Cb \bar\Wc$ as a sub-representation.
To see this we note that the composition series of $\Wc \otimes_\Cb \bar\Wc$ is
\be\label{eq:WxWbar-compseries}
\begin{tabular}{c|c|c|}
& $0$ & $2$ \\ 
\hline
$0$  & 1 & 0 \\
\hline
$2$ & 0 & 0 \\
\hline
\end{tabular}
\quad \longrightarrow \quad
\begin{tabular}{c|c|c|}
& $0$ & $2$ \\ 
\hline
$0$  & 0 & 1 \\
\hline
$2$ & 1 & 0 \\
\hline
\end{tabular}
\quad \longrightarrow \quad
\begin{tabular}{c|c|c|}
& $0$ & $2$ \\ 
\hline
$0$  & 0 & 0 \\
\hline
$2$ & 0 & 1 \\
\hline
\end{tabular} \ .
\ee
In order to embed $\Wc \otimes_\Cb \bar\Wc$ into $\Hc_{\text{bulk}}$, we need to map the 
level $0$ subspace $\Wc(0)\otimes_\Cb \bar\Wc(0)$  to the summand $\Wc(0)\otimes_\Cb \bar\Wc(0)$
of $\Hc_{\text{bulk}}$ at level $2$. However, the composition series \eqref{eq:P0-diagram} of 
$\Pc(0)$ shows that from there 
one can reach {\it e.g.}\ $\Wc(0)\otimes_\Cb \bar\Wc(1)$,
which is not contained in $\Wc \otimes_\Cb \bar\Wc$.
Nonetheless, $\Hc_{\text{bulk}}$ does still
contain a holomorphic field of weight $(2,0)$ and an anti-holomorphic field of weight $(0,2)$,
 which form the candidate stress tensor.
For example, a representative of the equivalence class of the state of  
weight $(2,0)$ at level $3$
of $\Hc_{\text{bulk}}$ (see \eqref{eq:0_compseries}) is
$\Wc(2) \otimes_\Cb \bar\Wc(0) \subset \Pc(2) \otimes_\Cb \bar\Pc(2)^*$,  
where $\bar\Wc(0)$ sits at level $3$.
Acting with $\bar{L}_{-1}$ gives a state of weight $(2,1)$, where the  
factor with weight $1$ has to be at level
$4$ of $\Pc(2)$, which, however, does not contain such a state. Thus
the state of weight $(2,0)$ really defines a holomorphic field.
\smallskip

As an aside, we note that this argument also suggests that the the state  
of weight $(1,0)$ at level
$3$ of $\Hc_{\text{bulk}}$ is {\em not} holomorphic, as acting with $ 
\bar{L}_{-1}$
on its representative
$\Wc(1) \otimes_\Cb \bar\Wc(0) \subset \Pc(1) \otimes_\Cb \bar\Pc(1)^*$,  
where $\bar\Wc(0)$
sits at level $3$, gives a state of weight $(1,1)$ in
$\Wc(1) \otimes_\Cb \bar\Wc(1) \subset \Pc(1) \otimes_\Cb \bar\Pc(1)^*$,
where $\bar\Wc(1)$ sits at level $4$, and it seems highly plausible  
that this state is nonzero.

\medskip\noindent {\em Boundary states:}
As a final consistency check, we can study the boundary states one can construct within
$\Hc_{\text{bulk}}$. In particular, as we shall show in the following section, we are
able to reproduce the boundary theory of \cite{Gaberdiel:2009ug} in this manner.

\section{The boundary state analysis}\label{sec:bnd}

In this section we want to study how the boundary conditions that were proposed in 
\cite{Gaberdiel:2009ug} fit into the present analysis. In particular, we want to show how
the corresponding boundary states can be constructed in our proposed bulk theory 
$\Hc_{\text{bulk}}$.

\subsection{Reconstructing the boundary states}\label{sec:boundary}

Let us begin by identifying the Ishibashi states and their overlaps that 
reproduce the annulus partition functions of  \cite{Gaberdiel:2009ug}; this first part
of the analysis does not require any detailed knowledge about the bulk space $\Hc_{\text{bulk}}$. The 
basic idea of the analysis is simple. Given the proposal of  \cite{Gaberdiel:2009ug}, we know the annulus 
partition functions in the `open string sector'. More specifically,  if $\Rc$ and $\Rc^\prime$ label two 
consistent boundary conditions, then the open string annulus amplitude is simply
\begin{equation}\label{open}
  \chi_{\Rc\otimes_f{\Rc^{\prime}}^\ast}(\tilde q) \ ,
\end{equation}
where $\chi_{\cal S}$ denotes the character of ${\cal S}$ and
$\tilde q = \exp(- 2 \pi i/\tau)$ is open string modular parameter.
Thus the overlap of the corresponding boundary states  
$\|\Rc\rangle\!\rangle$ and $\|\Rc'\rangle\!\rangle$ must equal
\begin{equation}\label{closed}
\langle\!\langle\Rc\|q^{L_0+\bar L_0}\|\Rc^\prime\rangle\!\rangle = 
\chi_{\Rc\otimes_f{\Rc^{\prime}}^\ast}(\tilde q) \ ,
\end{equation}
{\it i.e.}\ the modular $S$-transform of (\ref{open}) --- see appendix 
\ref{sec:S-matrix} for explicit formulae, or \cite{Feigin:2006iv,Semikhatov:2007qp} for 
the formula for general $\Wc_{p,q}$-models.

By considering the different powers of $q$ that appear in (\ref{closed}), it is clear which
contribution arises from which sector 
of $\Hc_\text{bulk}$ in \eqref{Hcdec}.
We can thus identify the overlaps of the various different 
Ishibashi states that appear. 
As was explained in Section~\ref{sec:intro-bndstate}, the  annulus amplitudes
can be described in terms of ten Ishibashi states, reflecting the fact that
the overlaps have a 2-dimensional null ideal, see eq.\ (\ref{nullideal}). 
We should stress, however, that this 
does not mean that the actual boundary states can be written in terms
of these ten Ishibashi states. It only means that for the calculation of the annulus amplitudes,
only the coupling to these ten Ishibashi states is required. (In fact, it seems natural to assume
that the actual boundary states will depend on two more Ishibashi states, but we have not
managed to prove this.)

So far, we have not used any detailed information about the structure of the bulk spectrum. 
The non-trivial consistency condition now arises from the requirement that the 
Ishibashi states (\ref{Ishiused}) and their overlaps (\ref{eq:Ishibashioverlaps})  can indeed be 
obtained within $\Hc_{\text{bulk}}$. In order to see that this is
possible we begin by classifying the most general Ishibashi states in $\Hc_{\text{bulk}}$.

\subsection{The Ishibashi states}

As was explained in \cite[Section~5.1]{Gaberdiel:2007jv}, we may think of the Ishibashi states in
terms of `Ishibashi morphisms'. Indeed, each Ishibashi state defines (and is defined by) a 
`bulk-boundary pairing' $\Hc_{\text{bulk}} \times \Wc \rightarrow {\mathbb C}$ that is compatible 
with the $\Wc$ action. Since the bulk space $\Hc_{\text{bulk}}$ is the quotient of $\hat\Hc$ by 
$\Nc$, the space of all Ishibashi intertwiners  consists of those bulk-boundary pairings on 
$\hat\Hc$  that vanish on $\Nc$. It was furthermore shown in \cite{Gaberdiel:2007jv} that every such
bulk boundary pairing can be written as 
\begin{equation}\label{brho}
b_\rho \bigl( u \otimes_\Cb \bar{v},w\bigr) = \sum_a  \bigl(\text{ev}_{\Pc(a)} (\rho_a(u_a),\bar v_a)\bigr) (w) \ , 
\end{equation}
where $\rho$ is an intertwiner \smash{$\rho=\bigoplus_a \rho_a : \Pc(a) \rightarrow \Pc(a)$}. 
Here \smash{$u \otimes_\Cb \bar{v} = \bigoplus_a u_a \otimes_\Cb \bar v_a \in \hat\Hc$}, and $w\in\Wc$. 
Furthermore, as also explained in \cite{Gaberdiel:2007jv}, the condition 
to vanish on $\Nc$ can be analysed by very similar methods as above in Section 
\ref{sec:maximal-nondeg}. The bulk boundary map $b_\rho$ vanishes on $\Nc$ if and 
only if
\be \label{eq:ish-morph-cond}
  \sum_a g_a\circ\rho_a\circ f_a=0
\ee  
for all $f_a:\Pc(i)\rightarrow\Pc(a)$ and $g_a:\Pc(a)\rightarrow\Pc(j)$ such that
\smash{$(\bigoplus_a f_a\otimes_{\mathbb{C}}\bar g_a^\ast)(\Pc(i)\otimes_{\mathbb{C}}
\bar\Pc(j)^\ast)$} lies in $\Nc$.

As is familiar from the usual rational case, we can analyse the Ishibashi states sector by sector. 
Again, the situation is easiest for the sectors $\Hc_{-{1}/{24}}$ and $\Hc_{{35}/{24}}$,
corresponding to the irreducible representations at the corner of the Kac table. In either case 
Schur's lemma implies that there is only a 
one-dimensional space of intertwiners, and since $\Nc$ 
only intersects these sectors in $0$,  there is no additional constraint to worry about. Thus we have 
one Ishibashi state from each of these two sectors.

\subsubsection{The edge of the Kac table}

The situation is more interesting for the sectors  from the `edge of the
Kac table', {\it i.e.}\ the sectors $\Hc_{1/8}$,   $\Hc_{5/8}$ and $\Hc_{1/3}$. For
concreteness, let us again only consider the case of $\Hc_{1/8}$; the situation in the
other sectors is completely analogous. First we recall that 
\begin{equation}\label{sumss}
\hat\Hc_{\frac{1}{8}} = \Bigl( \Pc(\tfrac{1}{8}) \otimes_{\mathbb C} \bar{\Pc} (\tfrac{1}{8})^\ast\Bigr)
 \, \oplus \, \Bigl(
 \Pc(\tfrac{33}{8}) \otimes_{\mathbb C} \bar{\Pc} (\tfrac{33}{8})^\ast \Bigr) \ .
 \end{equation}
 As was explained in \cite{Gaberdiel:2009ug},  the space of intertwiners 
 $\rho : \Pc(\frac18) \rightarrow \Pc(\frac18)$ is two-dimensional, and similarly for 
$ \Pc(\frac{33}8)$. Thus, in addition to the identity intertwiner $\text{id}_a$, there
is one linearly independent intertwiner which we denote by $n_a\equiv e^{a\rightarrow a}_{2}$;
our notation for the intertwiners is  explained in Appendix~\ref{sec:embedding}. 
The intertwiner $\text{id}_a$ acts as the identity on $\Pc(a)$, and as zero on 
the other summand in (\ref{sumss}), and similarly for $n_a$. In total, we therefore have four such 
intertwiners, and we write the most general ansatz as 
\begin{equation}
\rho =  \rho_{\frac{1}{8}} \oplus \rho_{\frac{33}{8}} = 
\bigoplus_{a=\frac{1}{8},\frac{33}{8}} A_a^{(1)} \, \text{id}^a + A_a^{(2)} n^a \ .
\end{equation}
The space $\Nc$ by which we have to quotient out $\hat\Hc_{1/8}$ is generated by
\begin{align}\label{eq:ishrel1}
\Bigl( e_{1;\alpha}^{a\rightarrow b}\otimes_{\mathbb{C}}{\rm id}-{\rm id}\otimes_{\mathbb{C}}
(e_{1;\alpha}^{a\rightarrow b})^\ast \Bigr) 
\left( \Pc(a) \otimes_{\mathbb C} \bar{\Pc} (b)^\ast \right) \ ,
\end{align}
where either ($a=\frac18$ and $b=\frac{33}8$) or ($a=\frac{33}8$ and $b=\frac18$), and
$\alpha$ denotes the two different choices of such intertwiners. (This is because the 
other relations of the form (\ref{eq:seesaw}), involving intertwiners of higher
degree, are a consequence of these.) 
By \eqref{eq:ish-morph-cond}, the intertwiner $\rho$ vanishes on \eqref{eq:ishrel1} if and 
only if  for each $\alpha$ we have on $\Pc(a)$
\begin{equation}
\rho_b\circ e^{a\rightarrow b}_{1;\alpha}- (e^{a\rightarrow b}_{1;\alpha})  \circ \rho_a = 0  \ .
\end{equation}
Thus there is a three-dimensional solution space with basis
\begin{equation}
\rho^{(1)} = \text{id}^{\frac{1}{8}} +  \text{id}^{\frac{33}{8}} \ , \qquad
\rho^{(a,2)} = n^{a} \ , \qquad a=\tfrac{1}{8}, \tfrac{33}{8} \ ,
\end{equation}
since $e^{a\rightarrow b}_{1;\alpha} \circ n^a = n^b \circ e^{a\rightarrow b}_{1;\alpha} = 0$. 
This analysis is completely analogous to the discussion in \cite{Gaberdiel:2007jv}.

\subsubsection{The interior of the Kac table}

This leaves us with analysing the Ishibashi morphisms in the sector $\Hc_0$, corresponding
to the `interior of the Kac table'. The following calculations make frequent reference to
the embedding diagrams \eqref{eq:P0-diagram}-\eqref{eq:P57-diagram} of the projective 
representations $\Pc(0),\Pc(1),\Pc(2),\Pc(5)$ and
$\Pc(7)$, and of our conventions regarding intertwiners described there. 
As is apparent from these embedding diagrams, the space of intertwiners of each of 
the projective representations $\Pc(a)$ with $a\in\{1,2,5,7\}$ is $4$-dimensional.  For 
each  such $\Pc(a)$ let $\text{id}^{a}$ be the identity intertwiner, 
while $e^{a\rightarrow a}_{4}$ denotes the unique intertwiner of degree four.  The remaining two
intertwiners have degree $2$, and will be denoted by $e^{a\rightarrow a}_{2;\alpha}$, where
$\alpha$ takes the appropriate two values out of $\{1,2,5,7\}$, as explained in 
Appendix~\ref{sec:embedding}. 
The most general ansatz for an intertwiner $\Pc(a)\rightarrow \Pc(a)$, $a=1,2,5,7$ is therefore
\begin{align}
  \rho_a=A^{\rm id}_a\, {\rm id}^{a}+\sum_\alpha A_a^\alpha\, e^{a\rightarrow a}_{2;\alpha}
		  +A^4_a\, e^{a\rightarrow a}_{4} \ .
\end{align}
For the case of $\Pc(0)$, the space of intertwiners from $\Pc(0)$ to itself is $3$-dimensional;
apart from the identity intertwiner $\text{id}^{0}$, 
there are the intertwiners $e^{0\rightarrow 0}_2$ and $e^{0\rightarrow 0}_4$ of degree $2$ and $4$, 
respectively. The most general ansatz for an intertwiner $\Pc(0)\rightarrow\Pc(0)$ is therefore
\begin{align}
  \rho_0= A^{\rm id}_0\, {\rm id}^0+A_0^2\, e^{0\rightarrow0}_2 +A^4_0\, e^{0\rightarrow0}_4\ .
\end{align}
Before taking  into account the constraints coming from $\Nc$, the number of intertwiners 
we need to consider in the sector $\hat\Hc_0$ is therefore $19=4\times 4 + 3$.
\smallskip

\noindent The intersection of $\Nc$ with $\hat\Hc_0$ is again generated by the relations of 
the form (\ref{eq:seesaw}) involving intertwiners of degree 1
\begin{align}
&\bigl( e_{1;\beta}^{a\rightarrow b}\otimes_{\mathbb{C}}{\rm id}^b
-{\rm id}^a\otimes_{\mathbb{C}} (\overline{e_{1;\beta}^{a\rightarrow b}})^\ast  \bigr) \, 
\bigl(\Pc(a) \otimes_{\mathbb C} \bar{\Pc}(b)^\ast \bigr) \nonumber \\
& \bigl(e_1^{0\rightarrow d}\otimes_{\mathbb{C}}{\rm id}^d
-{\rm id}^0\otimes_{\mathbb{C}} (\overline{e_1^{0\rightarrow d}})^\ast \bigr) \,
\bigl(\Pc(0) \otimes_{\mathbb C} \bar{\Pc}(d)^\ast \bigr) \\
& \bigl( e_1^{d\rightarrow 0}\otimes_{\mathbb{C}}{\rm id}^0
-{\rm id}^d\otimes_{\mathbb{C}} (\overline{e_1^{d\rightarrow 0}})^\ast \bigr)  \,
\bigl(\Pc(d) \otimes_{\mathbb C} \bar{\Pc}(0)^\ast \bigr) \ ,  \nonumber 
\end{align}
where $(a,b)=(1,5), (1,7), (2,5), (2,7)$ or $(a,b)=(5,1),(7,1),(5,2),(7,2)$ and 
$d\in\{1,2\}$. By 
\eqref{eq:ish-morph-cond},  the intertwiner $\rho$ vanishes on \eqref{eq:ishrel1} if and only if
\begin{align}\label{eq:isheqs}
\rho_b\circ e_{1;\beta}^{a\rightarrow b}-e_{1;\beta}^{a\rightarrow b}\circ\rho_a&=0
\quad \text{on $\Pc(a)$ ,} \nonumber\\
\rho_d\circ e_1^{0\rightarrow d}-e_1^{0\rightarrow d}\circ\rho_0&=0 \quad \text{on $\Pc(0)$ ,} \\
\rho_0\circ e_1^{d\rightarrow 0}- e_1^{d\rightarrow 0}\circ\rho_d&=0 \quad \text{on $\Pc(d)$ .}
\nonumber 
\end{align}
Evaluating the first equation of \eqref{eq:isheqs} we get
\begin{align}
  0&=\rho_b\circ e_{1;\beta}^{a\rightarrow b}-e_{1;\beta}^{a\rightarrow b}\circ\rho_a
  \nonumber \\
  &=\Big(A^{\rm id}_b\,{\rm id}^b+\sum_\eta A^\eta_b\,e^{b\rightarrow b}_{2;\eta}
  +A^4_b\,e^{b\rightarrow b}_4\Big)\circ e_{1;\beta}^{a\rightarrow b}
  -e_{1;\beta}^{a\rightarrow b}\circ\Big(A^{\rm id}_a{\rm id}^a
  +\sum_\eta A^\eta_a e^{a\rightarrow a}_{2;\eta}
  +A^4_a\,e^{a\rightarrow a}_4\Big)  \nonumber \\
  &=(A^{\rm id}_b-A^{\rm id}_a)e_{1;\beta}^{a\rightarrow b}
     +\sum_{\gamma} \Big( \sum_{\eta}A^\eta_b C_{\eta \beta}^{\gamma}
        - \sum_{\eta}A^\eta_a   \hat{C}_{\beta \eta}^{\gamma}  \Big) 
        e_{3;\gamma}^{a\rightarrow b}\ ,
\end{align}
where we have used that $e^{b\rightarrow b}_{4} \circ e^{a\rightarrow b}_{1;\beta}= 0 = 
e^{a\rightarrow b}_{1;\beta} \circ e^{a\rightarrow a}_{4}$, as well as 
defined the structure constants via (see Appendix~\ref{sec:embedding} for some examples)
\begin{equation}
e^{b\rightarrow b}_{2;\alpha} \circ e^{a\rightarrow b}_{1;\beta} = \sum_\gamma C^{\gamma}_{\alpha,\beta}\, 
e_{3;\gamma}^{a\rightarrow b}  \ , \qquad
e^{a\rightarrow b}_{1;\beta} \circ e^{a\rightarrow a}_{2;\alpha} = \sum_\gamma \hat{C}_{\beta,\alpha}^{\gamma}
\, e_{3;\gamma}^{a\rightarrow b} \ . 
\end{equation}
The second and third equations of \eqref{eq:isheqs} lead to equivalent relations. 
Concentrating on the former we obtain 
\begin{align}
  0&=\rho_d\circ e_1^{0\rightarrow d}-e_1^{0\rightarrow d}\circ\rho_0
  \nonumber\\
  &=(A^{\rm id}_d\,{\rm id}^d+A^5_d\,e^{d\rightarrow d}_{2,5}+A^{7}_d\,e^{d\rightarrow d}_{2,7}
  +A^4_d\,e^{d\rightarrow d}_4)\circ e_1^{0\rightarrow d}
  -e_1^{0\rightarrow d}\circ(A_0^{\rm id}\,{\rm id}^0+A^2_0\,e^{0\rightarrow0}_2
  +A^4_0\,e^{0\rightarrow0}_4)  \nonumber\\
  &=(A_d^{\rm id}-A_0^{\rm id})e_1^{0\rightarrow d}+(A^5_d+A_d^7-A_0^2)e^{0\rightarrow d}_3\ ,
\end{align}
where we have again used that 
$e_{4}^{d\rightarrow d} \circ e_1^{0\rightarrow d} = 0 = e_1^{0\rightarrow d} \circ e_4^{0\rightarrow 0}$. 
 A basis for the space of solutions is given by the following nine Ishibashi morphisms\footnote{We believe
that by an appropriate rescaling of the intertwiners, the coefficients $C^{\gamma}_{\alpha,\beta}$ and 
$\hat{C}^{\gamma}_{\beta, \alpha}$ can be chosen to be either zero or one.
This was used to arrive at \eqref{eq:ishbasis}.}
\begin{align}
  \rho^{({\rm id})}&={\rm id}^0+{\rm id}^1+{\rm id}^2+{\rm id}^5+{\rm id}^7 \nonumber \\
  \rho^{(\mu)}&= e_2^{0\rightarrow0}+e_{2,7}^{1\rightarrow1}
  +e_{2,5}^{2\rightarrow2}+e_{2,2}^{5\rightarrow5}
  +e_{2,1}^{7\rightarrow7}\nonumber\\
  \rho^{(\nu)}&= e_2^{0\rightarrow0}+e_{2,5}^{1\rightarrow1}
  +e_{2,7}^{2\rightarrow2}+e_{2,1}^{5\rightarrow5}
  +e_{2,2}^{7\rightarrow7}  \label{eq:ishbasis}  \\
  \rho^{(\delta)}&=e_{2,5}^{2\rightarrow2}-e_{2,7}^{2\rightarrow2}+e_{2,2}^{5\rightarrow5}-e_{2,2}^{7\rightarrow7}
  \nonumber\\
  \rho^{(\sigma_i)}&=e^{i\rightarrow i}_4\qquad i=0,1,2,5,7\ .\nonumber
\end{align}

\subsection{Overlaps of Ishibashi states}

In order to identify the Ishibashi states (\ref{eq:Ishibashistates}) 
with linear combinations of the Ishibashi states corresponding to the intertwiners 
\eqref{eq:ishbasis}, we need to determine the different overlaps between the latter.
Obviously, the overlaps between Ishibashi states from different sectors vanish, but since there is 
generically more than one Ishibashi state in each sector, the relative overlaps between them
are more complicated. The first step in identifying the structure of these overlaps is to
understand the relation between the Ishibashi morphisms $\rho$ and the corresponding Ishibashi 
states $|\rho\rangle\!\rangle$, thought of as elements in a
completion of $\Hc_{\text{bulk}}$. 
Suppose that $w\in \Hc_{\text{bulk}}$ is an arbitrary bulk state, then we have 
\begin{align}\label{Ishimor}
b_\rho(w,\Omega)=B(|\rho\rangle\!\rangle,w)\ ,
\end{align}
where $B(-,-)$ is the bulk 2-point function and $b_\rho$ is defined by 
\eqref{brho}. 
That $|\rho\rangle\!\rangle$ is annihilated by 
$(W_m -  (-1)^{h_W} \bar W_{-m})$ is related to the fact that $b_\rho(-,\Omega)$ is a chiral
2-point block on the sphere (see {\it e.g.} \cite[App.\,A]{Gaberdiel:2007jv}). 
Since the bulk 2-point function is non-degenerate, it follows that $\rho$
determines uniquely the corresponding Ishibashi state $|\rho\rangle\!\rangle$, and
vice versa. Given the knowledge of the corresponding Ishibashi state we can then
work out the cylinder overlaps since we have
\begin{equation}
\langle\!\langle \rho_1 | q^{L_0+\bar{L}_0} | \rho_2 \rangle\!\rangle = 
B\bigl( |\rho_1\rangle\!\rangle, q^{L_0+\bar{L}_0} | \rho_2 \rangle\!\rangle \bigr) \ .
\end{equation}
It follows from the definition of \eqref{brho} that $b_\rho(u \otimes_\Cb \bar{v} , \Omega)$ 
is only nonzero provided that $\rho_a(u_a)$ and $\bar{v}_a$ contain summands that are
conjugate to one another. 
Here we have written $u \otimes_\Cb \bar v = \sum_a u_a\otimes_\Cb \bar v_a$ with 
$u_a\otimes_\Cb \bar v_a\in \Pc(a)\otimes_{\mathbb{C}}\bar\Pc(a)^\ast$.

As regards the bulk 2-point function $B(-,-)$, we choose the following convention. 
Fixing a bulk 2-point function is equivalent to giving an isomorphism 
$\Hc_{\text{bulk}} \rightarrow \Hc_{\text{bulk}}^*$, which one can then precompose
with the canonical pairing $\Hc_{\text{bulk}}\times \Hc_{\text{bulk}}^\ast \rightarrow {\mathbb C}$. 
There is no unique choice, but there is a preferred such isomorphism:
if we keep the vector space decomposition implicit in \eqref{eq:0_compseries}, 
the conjugation map flips the direction of all arrows, and effectively turns the composition diagram upside down.
We pick the isomorphism which maps the level 0 states of $\Hc_{\text{bulk}}$ to states of $\Hc_{\text{bulk}}^\ast$
which have components only at level 0 (and not also at levels 2 and 4). With this choice,
the bulk 2-point function is non-vanishing only for combinations of states  at opposite
points in the bulk composition series.

Combining these two considerations we can then determine
the components of the bulk composition series in which the Ishibashi state has non-trivial
components, as we shall now see.

\subsubsection{The corner of the Kac table}

Let us begin with the Ishibashi states from the sectors $\Hc_{-1/24}$ and
$\Hc_{35/24}$. For these sectors the Ishibashi morphisms are proportional to the 
identity and the Ishibashi states are just those built upon the highest weight states. 
We can then normalise the two resulting Ishibashi states 
$| \tfrac{-1}{24} \rangle\!\rangle$ and $| \tfrac{35}{24} \rangle\!\rangle$
such that their overlaps are
\be
\langle\!\langle \tfrac{-1}{24} |\, q^{L_0 + \bar{L}_0}\, | \tfrac{-1}{24} \rangle\!\rangle
=  \sqrt{3} \,\chi_{{\cal W}(\tfrac{-1}{24})}(q)
\qquad \text{and} \qquad
\langle\!\langle \tfrac{35}{24} |\, q^{L_0 + \bar{L}_0}\, | \tfrac{35}{24} \rangle\!\rangle
= - \sqrt{3} \,\chi_{{\cal W}(\tfrac{35}{24})}(q)\ ,
\ee
in agreement with \eqref{eq:Ishibashioverlaps}.

\subsubsection{The edge of the Kac table}

For the sectors along the `edge of the Kac table', let us first work out the non-trivial components
of the Ishibashi states. As in the previous sections we shall do this explicitly in the $\Hc_{1/8}$ sector;
the analysis in the other edge sectors is similar. The Ishibashi morphism $b_{\rho^{\rm id}}$ 
is non-vanishing on $u \otimes_\Cb \bar{v} \in \Hc_{1/8}$ 
provided that the components $u_a \otimes_\Cb \bar{v}_a$
are conjugate to one another in $\Pc(a)\otimes_{\mathbb{C}}\bar\Pc(a)$. 
If we choose the representatives $u_a\otimes_\Cb \bar{v}_a$ according to \eqref{H(0)},
then $u_a$ lies in the top $\Wc(a)$ at level $0$ in $\Pc(a)$. In order for $\bar v_a$ to be 
conjugate to $u_a$ it must then lie in the bottom $\Wc(a)$, {\it i.e.}\ at level $2$ in $\bar \Pc(a)^\ast$. 
In terms of the analysis of Section~\ref{sec:compedge}, it then follows that $u \otimes_\Cb \bar{v}$ must lie
in the grey components of the composition series,
\be\label{eq:1/8_compseriesIshi}
b_{\rho^{({\rm id})}} \neq 0 \text{ only on} \quad
\begin{tabular}{c|c|c|}
& $\tfrac{1}{8}$ & $\tfrac{33}{8}$ \\ 
\hline
$\tfrac{1}{8}$  & 1 & 0 \\
\hline
$\tfrac{33}{8}$ & 0 & 1 \\
\hline
\end{tabular}
\quad \longrightarrow \quad
\begin{tabular}{c|c|c|}
& $\tfrac{1}{8}$ & $\tfrac{33}{8}$ \\ 
\hline
$\tfrac{1}{8}$  & 0 & 2 \\
\hline
$\tfrac{33}{8}$ & 2 & 0 \\
\hline
\end{tabular}
\quad \longrightarrow \quad
\begin{tabular}{c|c|c|}
& $\tfrac{1}{8}$ & $\tfrac{33}{8}$ \\ 
\hline
$\tfrac{1}{8}$  & \cellcolor[gray]{.8}1 & 0 \\
\hline
$\tfrac{33}{8}$ & 0 & \cellcolor[gray]{.8}1 \\
\hline
\end{tabular} \ ,
\ee
where we have used the same notation as in (\ref{eq:1/8_compseries}). The corresponding 
Ishibashi state (that we shall denote by $|{\rm id}\rangle\!\rangle$) then has non-zero components in 
the conjugate sectors. 
\be \label{eq:|id>-in-1/8-sector}
|{\rm id}\rangle\!\rangle \in \quad
\begin{tabular}{c|c|c|}
& $\tfrac{1}{8}$ & $\tfrac{33}{8}$ \\ 
\hline
$\tfrac{1}{8}$  & \cellcolor[gray]{.8}1 & 0 \\
\hline
$\tfrac{33}{8}$ & 0 & \cellcolor[gray]{.8}1 \\
\hline
\end{tabular}
\quad \longrightarrow \quad
\begin{tabular}{c|c|c|}
& $\tfrac{1}{8}$ & $\tfrac{33}{8}$ \\ 
\hline
$\tfrac{1}{8}$  & 0 & 2 \\
\hline
$\tfrac{33}{8}$ & 2 & 0 \\
\hline
\end{tabular}
\quad \longrightarrow \quad
\begin{tabular}{c|c|c|}
& $\tfrac{1}{8}$ & $\tfrac{33}{8}$ \\ 
\hline
$\tfrac{1}{8}$  & 1 & 0 \\
\hline
$\tfrac{33}{8}$ & 0 & 1 \\
\hline
\end{tabular} \ .
\ee
\medskip

The analysis for the Ishibashi morphism \smash{$\rho^{(a,2)}$} is similar. Again, \smash{$\rho^{(a,2)}$}
is non-vanishing on $u \otimes_\Cb \bar{v}\in \Hc_{1/8}$ provided that $n_a(u)$ is conjugate to $\bar{v}$. 
Since $n_a(u)$ maps the top $\Wc(a)$ to the bottom $\Wc(a)$ in $\Pc(a)$, it follows that 
$u \otimes_\Cb \bar{v}$ has to have non-zero components in the sectors
\be\label{eq:1/8_compseriesIshi18}
b_{\rho^{(1/8,2)}} \neq 0 \text{ only on} \quad
\begin{tabular}{c|c|c|}
& $\tfrac{1}{8}$ & $\tfrac{33}{8}$ \\ 
\hline
$\tfrac{1}{8}$  & \cellcolor[gray]{.8}1 & 0 \\
\hline
$\tfrac{33}{8}$ & 0 & 1 \\
\hline
\end{tabular}
\quad \longrightarrow \quad
\begin{tabular}{c|c|c|}
& $\tfrac{1}{8}$ & $\tfrac{33}{8}$ \\ 
\hline
$\tfrac{1}{8}$  & 0 & 2 \\
\hline
$\tfrac{33}{8}$ & 2 & 0 \\
\hline
\end{tabular}
\quad \longrightarrow \quad
\begin{tabular}{c|c|c|}
& $\tfrac{1}{8}$ & $\tfrac{33}{8}$ \\ 
\hline
$\tfrac{1}{8}$  & 1 & 0 \\
\hline
$\tfrac{33}{8}$ & 0 & 1 \\
\hline
\end{tabular} 
\ee
and
\be\label{eq:1/8_compseriesIshi338}
b_{\rho^{(33/8,2)}} \neq 0 \text{ only on} \quad
\begin{tabular}{c|c|c|}
& $\tfrac{1}{8}$ & $\tfrac{33}{8}$ \\ 
\hline
$\tfrac{1}{8}$  & 1 & 0 \\
\hline
$\tfrac{33}{8}$ & 0 & \cellcolor[gray]{.8}1 \\
\hline
\end{tabular}
\quad \longrightarrow \quad
\begin{tabular}{c|c|c|}
& $\tfrac{1}{8}$ & $\tfrac{33}{8}$ \\ 
\hline
$\tfrac{1}{8}$  & 0 & 2 \\
\hline
$\tfrac{33}{8}$ & 2 & 0 \\
\hline
\end{tabular}
\quad \longrightarrow \quad
\begin{tabular}{c|c|c|}
& $\tfrac{1}{8}$ & $\tfrac{33}{8}$ \\ 
\hline
$\tfrac{1}{8}$  & 1 & 0 \\
\hline
$\tfrac{33}{8}$ & 0 & 1 \\
\hline
\end{tabular} \ ,
\ee
respectively. The corresponding Ishibashi states will be denoted by $|n_{1/8}\rangle\!\rangle$
and $|n_{33/8}\rangle\!\rangle$, and they must have non-trivial components in the conjugate
sectors, {\it i.e.}\  in
\be
\label{eq:1/8_compseriesIshi181}
|n_{1/8}\rangle\!\rangle \in \quad
\begin{tabular}{c|c|c|}
& $\tfrac{1}{8}$ & $\tfrac{33}{8}$ \\ 
\hline
$\tfrac{1}{8}$  & 1 & 0 \\
\hline
$\tfrac{33}{8}$ & 0 & 1 \\
\hline
\end{tabular}
\quad \longrightarrow \quad
\begin{tabular}{c|c|c|}
& $\tfrac{1}{8}$ & $\tfrac{33}{8}$ \\ 
\hline
$\tfrac{1}{8}$  & 0 & 2 \\
\hline
$\tfrac{33}{8}$ & 2 & 0 \\
\hline
\end{tabular}
\quad \longrightarrow \quad
\begin{tabular}{c|c|c|}
& $\tfrac{1}{8}$ & $\tfrac{33}{8}$ \\ 
\hline
$\tfrac{1}{8}$  & \cellcolor[gray]{.8}1 & 0 \\
\hline
$\tfrac{33}{8}$ & 0 & 1 \\
\hline
\end{tabular} 
\ee
and
\be
\label{eq:1/8_compseriesIshi3381}
|n_{33/8}\rangle\!\rangle \in \quad
\begin{tabular}{c|c|c|}
& $\tfrac{1}{8}$ & $\tfrac{33}{8}$ \\ 
\hline
$\tfrac{1}{8}$  & 1 & 0 \\
\hline
$\tfrac{33}{8}$ & 0 & 1 \\
\hline
\end{tabular}
\quad \longrightarrow \quad
\begin{tabular}{c|c|c|}
& $\tfrac{1}{8}$ & $\tfrac{33}{8}$ \\ 
\hline
$\tfrac{1}{8}$  & 0 & 2 \\
\hline
$\tfrac{33}{8}$ & 2 & 0 \\
\hline
\end{tabular}
\quad \longrightarrow \quad
\begin{tabular}{c|c|c|}
& $\tfrac{1}{8}$ & $\tfrac{33}{8}$ \\ 
\hline
$\tfrac{1}{8}$  & 1 & 0 \\
\hline
$\tfrac{33}{8}$ & 0 & \cellcolor[gray]{.8}1 \\
\hline
\end{tabular} \ .
\ee

Given this information, we can then directly work out the non-trivial overlaps between
these Ishibashi states. For example, for the overlaps of $|n_{1/8}\rangle\!\rangle$ 
(or $|n_{33/8}\rangle\!\rangle$) with themselves vanish, since both of these states 
only have a component in the `bottom' factor, but none in the conjugate `top' factor. 
On the other hand, the overlap with  $|\text{id}\rangle\!\rangle$ is non-zero, and we 
can choose the normalisation of $|n_{1/8}\rangle\!\rangle$ and $|n_{33/8}\rangle\!\rangle$
such that
\begin{align}\label{3.28}
  \langle\!\langle{\rm id}|q^{L_0+\bar L_0}|n_{1/8}\rangle\!\rangle&=\chi_{\Wc(\tfrac{1}{8})}(q) &
  \langle\!\langle{\rm id}|q^{L_0+\bar L_0}|n_{33/8}\rangle\!\rangle&=\chi_{\Wc(\tfrac{33}{8})} (q) \ .
\end{align}

Let us decompose the action of $L_0$ on a representation ${\cal R}$ as 
$L_0 = L_0^\text{diag} + L_0^\text{nil}$, where $L_0^\text{nil}$ denotes the nilpotent part.
Since for a mode $W_m$ of a homogeneous generator $W$ of $\Wc_{2,3}$ we have
$[L_0,W_m] = -m W_m$ as well as $[L_0^\text{diag},W_m] = -m W_m$ (after all, $L_0^\text{diag}$ 
just gives the grading of ${\cal R}$ and $W_m$ is a map of degree $-m$), it follows that
$[L_0^\text{nil},W_m]=0$, {\em i.e.}\ the nilpotent part of $L_0$ is an intertwiner from ${\cal R}$ to 
itself. In particular, if $|i\rangle\!\rangle$ is an Ishibashi state, so is 
$(L_0^\text{nil} + \bar L_0^\text{nil})^n|i\rangle\!\rangle$
for any $n \ge 0$.

The above discussion, together with the observations in \eqref{eq:|id>-in-1/8-sector}, 
\eqref{eq:1/8_compseriesIshi181} and \eqref{eq:1/8_compseriesIshi3381} shows that
$(L_0^\text{nil} + \bar L_0^\text{nil})|{\rm id}\rangle\!\rangle$ is again an Ishibashi state, and 
it has to be a linear combination of $|n_{1/8}\rangle\!\rangle$ and $|n_{33/8}\rangle\!\rangle$. 
We can choose the normalisation of $|{\rm id}\rangle\!\rangle$ such that
\begin{align} \label{eq:alpha-condition}
\frac{1}{2\pi}\,  
(L_0^\text{nil} + \bar L_0^\text{nil}) |{\rm id}\rangle\!\rangle = 
  |n_{1/8}\rangle\!\rangle + \alpha_{33/8} |n_{33/8}\rangle\!\rangle\ ,
\end{align}
for some constant $\alpha_{33/8} \in \Cb$.
We should note that $\alpha_{33/8}$ is in principle determined by the structure of 
$\Hc_{\text{bulk}}$, but that with our current limited understanding of the latter
we cannot actually calculate it from first principles. In any case it follows that
\begin{align}\label{3.30}
  \langle\!\langle{\rm id}|q^{L_0+\bar L_0}|{\rm id}\rangle\!\rangle=
   i   \tau 
\Big( 
\chi_{\Wc(\tfrac{1}{8})}(q) +  \alpha_{33/8} \cdot \chi_{\Wc(\tfrac{33}{8})}(q) \Big)\ .
\end{align}
If we now set
\be\label{18Ishi}
  |\tfrac{1}{8},A\rangle\!\rangle=\frac{i}{\sqrt{3}}\, |{\rm id}\rangle\!\rangle
  \quad , \qquad
  |\tfrac{1}{8},B\rangle\!\rangle=-2i \,\big(\,|n_{1/8}\rangle\!\rangle+|n_{33/8}\rangle\!\rangle\big)
  \quad , \qquad
  \alpha_{33/8}=-2\ ,
\ee
we precisely reproduce the overlaps \eqref{eq:Ishibashioverlaps}.

We should also mention that the two Ishibashi states in (\ref{18Ishi}) are characterised
by the property that their overlaps do not lead to any $\tilde\tau = -1/\tau$ 
terms in the open string. 
Indeed, the additional linear combination, which we could take to be 
$|\tfrac{1}{8},C\rangle\!\rangle=|n_{1/8}\rangle\!\rangle-|n_{33/8}\rangle\!\rangle$
cannot enter the boundary state construction since it leads, in the open string 
loop diagram, to a term proportional to $\tilde\tau$. For example, we find
\begin{align}
  \langle\!\langle \tfrac{1}{8},A | q^{L_0+\bar{L}_0}  |\tfrac{1}{8},C\rangle\!\rangle
  &= \frac{i}{\sqrt{3}} \Big( \chi_{\Wc(\tfrac18)}(q)-\chi_{\Wc(\tfrac{33}{8})}(q) \Big) \\
  &= -\frac{i}{18}\Big(\chi_{\Wc(\tfrac18)}(\tilde q)
    +\chi_{\Wc(\tfrac{33}{8})}(\tilde q)\Big)
  +\frac{\tilde\tau}{3 \sqrt{3}} \Big(\chi_{\Wc(\tfrac{1}{8})}(\tilde q)
    - 2 \chi_{\Wc(\tfrac{33}8)}(\tilde q)\Big) \nonumber\\[.3em]
   &\phantom{=} ~ +\left(\text{contributions from other sectors}\right) \ ,\nonumber
\end{align}
where $\tilde\tau = -1/\tau$ is the modular parameter in the open string channel.
The situation is therefore completely analogous to what was found in \cite{Gaberdiel:2007jv}
for the $\Wc_{1,p}$-models. 

The analysis in the other two sectors is essentially identical. Indeed, the analogues of 
\eqref{3.28} and \eqref{3.30} hold also for the Ishibashi states in $\Hc_{5/8}$ and $\Hc_{1/3}$, 
respectively, and this determines their overlaps up to the constants $\alpha_{21/8}$ and 
$\alpha_{10/3}$ as in \eqref{eq:alpha-condition}. The Ishibashi states of (\ref{eq:Ishibashistates}) can then
be identified with
\begin{align}
  |\tfrac58,A\rangle\!\rangle&= i \, \sqrt{\tfrac23}\, |{\rm id}\rangle\!\rangle&
  |\tfrac58,B\rangle\!\rangle&=-i \,\sqrt{2}\, \bigl(|n_{5/8}\rangle\!\rangle+|n_{21/8}\rangle\!\rangle\bigr)\\
    |\tfrac13,A\rangle\!\rangle&=3^{\frac14} \, \sqrt{2} \, |{\rm id}\rangle\!\rangle&\nonumber
  |\tfrac13,B\rangle\!\rangle&= 3^{-\frac34}\,\sqrt{2}  \, 
  \bigl(|n_{1/3}\rangle\!\rangle+|n_{10/3}\rangle\!\rangle\bigr) \ ,
\end{align}
and we reproduce precisely the overlaps  (\ref{eq:Ishibashioverlaps}) provided that
\be
\alpha_{21/8} =-\tfrac{1}{2} \ , \qquad  \alpha_{10/3} =-1 \ .
\ee
The third Ishibashi state in each of these sectors cannot contribute to the boundary states
since it would lead to terms proportional to $\tilde\tau$ in the open string channel.

\subsubsection{The interior of the Kac table}\label{sec:bulk-int-kac}

It therefore only remains to identify the Ishibashi states with components in $\Hc_0$, the 
`interior of the Kac table'. The Ishibashi morphism $b_{\rho^{({\rm id})}}$ is 
non-vanishing for $u \otimes_\Cb \bar v\in\Hc_{0}$ 
provided that the left and right tensor factor of the representatives in
$\hat\Hc_0$ are conjugate to one another. By the same logic as above, this 
can only be the case provided that $u \otimes_\Cb \bar v$ lies in the grey components 
\begin{align}
  \begin{tikzpicture}
    [>=angle 90,factor/.style={rectangle, draw, inner sep=0pt, minimum size=1cm}
    ,baseline=(fact1)]
    \node[factor] (fact1) {};
    \node[factor] (fact2) [right=of fact1]{};
    \node[factor] (fact3) [right=of fact2]{};
    \node[factor] (fact4) [right=of fact3]{};
    \node (fact5) [right=of fact4]{
      \scriptsize\begin{tabular}{c|c|c|c|c|c|}
        &\hspace*{-.6em} 0 \stab 1 \stab 2 \stab 5 \stab 7 \hspace*{-.6em}\\
        \hline
        0 &\hspace*{-.6em} \cellcolor[gray]{.8}1 \stab   \stab   \stab   \stab   \hspace*{-.6em}\\
        \hline
        1 &\hspace*{-.6em}   \stab \cellcolor[gray]{.8}1 \stab   \stab   \stab   \hspace*{-.6em}\\
        \hline
        2 &\hspace*{-.6em}   \stab   \stab \cellcolor[gray]{.8}1 \stab   \stab   \hspace*{-.6em}\\
        \hline
        5 &\hspace*{-.6em}   \stab   \stab   \stab \cellcolor[gray]{.8}1 \stab   \hspace*{-.6em}\\
        \hline
        7 &\hspace*{-.6em}   \stab   \stab   \stab   \stab \cellcolor[gray]{.8}1 \hspace*{-.6em}\\
        \hline
      \end{tabular}
      };
    \node (ish) [left=of fact1] {$b_{\rho^{({\rm id})}}\neq 0 \text{ only on} \hspace{-2em}$};
    \draw[->] (fact1) -- (fact2);
    \draw[->] (fact2) -- (fact3);
    \draw[->] (fact3) -- (fact4);
    \draw[->] (fact4) -- (fact5);
  \end{tikzpicture}\ ,
\end{align}
where the empty squares are the remaining composition factors of \eqref{eq:0_compseries},
which we have left blank to keep the notation compact.
The corresponding Ishibashi state $|{\rm id}\rangle\!\rangle$ 
therefore has non-vanishing components in the conjugate sectors
\begin{align}
  \begin{tikzpicture}
    [>=angle 90,factor/.style={rectangle, draw, inner sep=0pt, minimum size=1cm}
    ,baseline=(fact1)]
    \node (fact1) {
      \scriptsize\begin{tabular}{c|c|c|c|c|c|}
        &\hspace*{-.6em} 0 \stab 1 \stab 2 \stab 5 \stab 7 \hspace*{-.6em}\\
        \hline
        0 &\hspace*{-.6em} \cellcolor[gray]{.8}1 \stab   \stab   \stab   \stab   \hspace*{-.6em}\\
        \hline
        1 &\hspace*{-.6em}   \stab \cellcolor[gray]{.8}1 \stab   \stab   \stab   \hspace*{-.6em}\\
        \hline
        2 &\hspace*{-.6em}   \stab   \stab \cellcolor[gray]{.8}1 \stab   \stab   \hspace*{-.6em}\\
        \hline
        5 &\hspace*{-.6em}   \stab   \stab   \stab \cellcolor[gray]{.8}1 \stab   \hspace*{-.6em}\\
        \hline
        7 &\hspace*{-.6em}   \stab   \stab   \stab   \stab \cellcolor[gray]{.8}1 \hspace*{-.6em}\\
        \hline
      \end{tabular}};
    \node[factor] (fact2) [right=of fact1]{};
    \node[factor] (fact3) [right=of fact2]{};
    \node[factor] (fact4) [right=of fact3]{};
    \node[factor] (fact5) [right=of fact4]{};
    \node (ish) [left=of fact1] {$|{\rm id}\rangle\!\rangle \in $};
    \draw[->] (fact1) -- (fact2);
    \draw[->] (fact2) -- (fact3);
    \draw[->] (fact3) -- (fact4);
    \draw[->] (fact4) -- (fact5);
  \end{tikzpicture}\ .
\end{align}
The Ishibashi morphisms $b_{\rho^{(\sigma_i)}}$ are non-vanishing for $u \otimes_\Cb \bar v\in\Hc_{0}$
provided that $\rho^{(\sigma_i)}(u)$ is conjugate to $\bar v$, {\it i.e.} provided that 
$u \otimes_\Cb \bar v$ lies in the $\Wc(i)\otimes_{\mathbb{C}}\bar\Wc(i)$ component
at level~0 of the $\Hc_0$ composition series, {\it e.g.}\ for $i=2$
\begin{align}\label{sigsupp}
  \begin{tikzpicture}
    [>=angle 90,factor/.style={rectangle, draw, inner sep=0pt, minimum size=1cm}
    ,baseline=(fact1)]
    \node (fact1) {
      \scriptsize\begin{tabular}{c|c|c|c|c|c|}
        &\hspace*{-.6em} 0 \stab 1 \stab 2 \stab 5 \stab 7 \hspace*{-.6em}\\
        \hline
        0 &\hspace*{-.6em} 1 \stab   \stab   \stab   \stab   \hspace*{-.6em}\\
        \hline
        1 &\hspace*{-.6em}   \stab 1 \stab   \stab   \stab   \hspace*{-.6em}\\
        \hline
        2 &\hspace*{-.6em}   \stab   \stab \cellcolor[gray]{.8}1 \stab   \stab   \hspace*{-.6em}\\
        \hline
        5 &\hspace*{-.6em}   \stab   \stab   \stab 1 \stab   \hspace*{-.6em}\\
        \hline
        7 &\hspace*{-.6em}   \stab   \stab   \stab   \stab 1 \hspace*{-.6em}\\
        \hline
      \end{tabular}};
    \node[factor] (fact2) [right=of fact1]{};
    \node[factor] (fact3) [right=of fact2]{};
    \node[factor] (fact4) [right=of fact3]{};
    \node[factor] (fact5) [right=of fact4]{};
    \node (ish) [left=of fact1] {$b_{\rho^{(\sigma_2)}} \neq 0 \text{ only on} \hspace{-2em}$};
    \draw[->] (fact1) -- (fact2);
    \draw[->] (fact2) -- (fact3);
    \draw[->] (fact3) -- (fact4);
    \draw[->] (fact4) -- (fact5);
  \end{tikzpicture}\ .
\end{align}
The corresponding Ishibashi state $|\sigma_2\rangle\!\rangle$ must therefore have a 
non-vanishing component in the conjugate sector
\begin{align}
  \begin{tikzpicture}
    [>=angle 90,factor/.style={rectangle, draw, inner sep=0pt, minimum size=1cm}
    ,baseline=(fact1)]
    \node[factor] (fact1) {};
    \node[factor] (fact2) [right=of fact1]{};
    \node[factor] (fact3) [right=of fact2]{};
    \node[factor] (fact4) [right=of fact3]{};
    \node (fact5) [right=of fact4]{
      \scriptsize\begin{tabular}{c|c|c|c|c|c|}
        &\hspace*{-.6em} 0 \stab 1 \stab 2 \stab 5 \stab 7 \hspace*{-.6em}\\
        \hline
        0 &\hspace*{-.6em} 1 \stab   \stab   \stab   \stab   \hspace*{-.6em}\\
        \hline
        1 &\hspace*{-.6em}   \stab 1 \stab   \stab   \stab   \hspace*{-.6em}\\
        \hline
        2 &\hspace*{-.6em}   \stab   \stab \cellcolor[gray]{.8}1 \stab   \stab   \hspace*{-.6em}\\
        \hline
        5 &\hspace*{-.6em}   \stab   \stab   \stab 1 \stab   \hspace*{-.6em}\\
        \hline
        7 &\hspace*{-.6em}   \stab   \stab   \stab   \stab 1 \hspace*{-.6em}\\
        \hline
      \end{tabular}
      };
    \node (ish) [left=of fact1] {$|\sigma_2\rangle\!\rangle \in$};
    \draw[->] (fact1) -- (fact2);
    \draw[->] (fact2) -- (fact3);
    \draw[->] (fact3) -- (fact4);
    \draw[->] (fact4) -- (fact5);
  \end{tikzpicture}\ .
\end{align}
The analysis for the other Ishibashi states $|\sigma_i\rangle\!\rangle$ for $i=0,1,5,7$ is
similar, with $\Wc(2)\otimes_{\mathbb{C}}\bar\Wc(2)$ being replaced by 
$\Wc(i)\otimes_{\mathbb{C}}\bar\Wc(i)$.

Finally, we consider the Ishibashi morphisms $b_{\rho^{(\mu)}}$, $b_{\rho^{(\nu)}}$ and 
$b_{\rho^{(\delta)}}$. For them the analysis is more complicated because of 
the multiplicities of $2$ on the diagonal of the third factor of the composition series 
\eqref{eq:0_compseries}. In \eqref{H(0)} these two components come from
$\Wc(a)\otimes_{\mathbb C} \bar{\Wc}(a)$, where $\Wc(a)$ lies at level $0$, while 
$\bar\Wc(a)$ lies on the very left or very right at level $2$ in $\bar\Pc(a)^\ast$; we shall
therefore denote these sectors by $\ell$ and $r$, respectively. Note that the sector labelled by 
$\ell$ is conjugate to that labelled by $r$, and vice versa. With this notation, 
we have
\begin{align}
  \begin{tikzpicture}
    [>=angle 90,factor/.style={rectangle, draw, inner sep=0pt, minimum size=1cm}
    ,baseline=(fact1)]
    \node[factor] (fact1) {};
    \node[factor] (fact2) [right=of fact1]{};
    \node (fact3) [right=of fact2]{
      \scriptsize\begin{tabular}{c|c|c|c|c|c|}
        &\hspace*{-.6em} 0 \stab 1 \stab 2 \stab 5 \stab 7 \hspace*{-.6em}\\
        \hline
        0 &\hspace*{-.6em} \cellcolor[gray]{.8}1 \stab   \stab   \stab   \stab   \hspace*{-.6em}\\
        \hline
        1 &\hspace*{-.6em}   \stab \cellcolor[gray]{.8}$\ell$ \stab   \stab   \stab   \hspace*{-.6em}\\
        \hline
        2 &\hspace*{-.6em}   \stab   \stab \cellcolor[gray]{.8}$\ell$ \stab   \stab   \hspace*{-.6em}\\
        \hline
        5 &\hspace*{-.6em}   \stab   \stab   \stab \cellcolor[gray]{.8}$\ell$ \stab   \hspace*{-.6em}\\
        \hline
        7 &\hspace*{-.6em}   \stab   \stab   \stab   \stab \cellcolor[gray]{.8}$\ell$ \hspace*{-.6em}\\
        \hline
      \end{tabular}};
    \node[factor] (fact4) [right=of fact3]{};
    \node[factor] (fact5) [right=of fact4]{};
    \node (ish) [left=of fact1] {$b_{\rho^{(\mu)}} \neq 0 \text{ only on} \hspace{-2em}$};
    \draw[->] (fact1) -- (fact2);
    \draw[->] (fact2) -- (fact3);
    \draw[->] (fact3) -- (fact4);
    \draw[->] (fact4) -- (fact5);
  \end{tikzpicture}\ ,
\end{align}
\begin{align}
  \begin{tikzpicture}
    [>=angle 90,factor/.style={rectangle, draw, inner sep=0pt, minimum size=1cm}
    ,baseline=(fact1)]
    \node[factor] (fact1) {};
    \node[factor] (fact2) [right=of fact1]{};
    \node (fact3) [right=of fact2]{
      \scriptsize\begin{tabular}{c|c|c|c|c|c|}
          &\hspace*{-.6em} 0 \stab 1 \stab 2 \stab 5 \stab 7 \hspace*{-.6em}\\
        \hline
        0 &\hspace*{-.6em} \cellcolor[gray]{.8}1 \stab   \stab   \stab   \stab   \hspace*{-.6em}\\
        \hline
        1 &\hspace*{-.6em}   \stab \cellcolor[gray]{.8}$r$ \stab   \stab   \stab   \hspace*{-.6em}\\
        \hline
        2 &\hspace*{-.6em}   \stab   \stab \cellcolor[gray]{.8}$r$ \stab   \stab   \hspace*{-.6em}\\
        \hline
        5 &\hspace*{-.6em}   \stab   \stab   \stab \cellcolor[gray]{.8}$r$ \stab   \hspace*{-.6em}\\
        \hline
        7 &\hspace*{-.6em}   \stab   \stab   \stab   \stab \cellcolor[gray]{.8}$r$ \hspace*{-.6em}\\
        \hline
      \end{tabular}};
    \node[factor] (fact4) [right=of fact3]{};
    \node[factor] (fact5) [right=of fact4]{};
    \node (ish) [left=of fact1] {$b_{\rho^{(\nu)}} \neq 0 \text{ only on} \hspace{-2em}$};
    \draw[->] (fact1) -- (fact2);
    \draw[->] (fact2) -- (fact3);
    \draw[->] (fact3) -- (fact4);
    \draw[->] (fact4) -- (fact5);
  \end{tikzpicture}\ ,
\end{align}
\begin{align}
  \begin{tikzpicture}
    [>=angle 90,factor/.style={rectangle, draw, inner sep=0pt, minimum size=1cm}
    ,baseline=(fact1)]
    \node[factor] (fact1) {};
    \node[factor] (fact2) [right=of fact1]{};
    \node (fact3) [right=of fact2]{
      \scriptsize\begin{tabular}{c|c|c|c|c|c|}
          &\hspace*{-.6em} 0 \stab 1 \stab 2 \stab 5 \stab 7 \hspace*{-.6em}\\
        \hline
        0 &\hspace*{-.6em} 1 \stab   \stab   \stab   \stab   \hspace*{-.6em}\\
        \hline
        1 &\hspace*{-.6em}   \stab 2 \stab   \stab   \stab   \hspace*{-.6em}\\
        \hline
        2 &\hspace*{-.6em}   \stab   \stab \cellcolor[gray]{.8}2 \stab   \stab   \hspace*{-.6em}\\
        \hline
        5 &\hspace*{-.6em}   \stab   \stab   \stab \cellcolor[gray]{.8}$\ell$ \stab   \hspace*{-.6em}\\
        \hline
        7 &\hspace*{-.6em}   \stab   \stab   \stab   \stab \cellcolor[gray]{.8}$r$ \hspace*{-.6em}\\
        \hline
      \end{tabular}};
    \node[factor] (fact4) [right=of fact3]{};
    \node[factor] (fact5) [right=of fact4]{};
    \node (ish) [left=of fact1] {$b_{\rho^{(\delta)}} \neq 0 \text{ only on} \hspace{-2em}$};
    \draw[->] (fact1) -- (fact2);
    \draw[->] (fact2) -- (fact3);
    \draw[->] (fact3) -- (fact4);
    \draw[->] (fact4) -- (fact5);
  \end{tikzpicture}\ .
\end{align}
The above diagrams imply in turn that the corresponding Ishibashi states must 
have non-vanishing components in
\begin{align}
  \begin{tikzpicture}
    [>=angle 90,factor/.style={rectangle, draw, inner sep=0pt, minimum size=1cm}
    ,baseline=(fact1)]
    \node (state) {$|\mu\rangle\!\rangle \in$};
    \node[factor] (fact1) [right=of state]{};
    \node[factor] (fact2) [right=of fact1]{};
    \node (fact3) [right=of fact2]{
      \scriptsize\begin{tabular}{c|c|c|c|c|c|}
          &\hspace*{-.6em} 0 \stab 1 \stab 2 \stab 5 \stab 7 \hspace*{-.6em}\\
        \hline
        0 &\hspace*{-.6em} \cellcolor[gray]{.8}1 \stab   \stab   \stab   \stab   \hspace*{-.6em}\\
        \hline
        1 &\hspace*{-.6em}   \stab\cellcolor[gray]{.8}$r$\stab   \stab   \stab   \hspace*{-.6em}\\
        \hline
        2 &\hspace*{-.6em}   \stab   \stab\cellcolor[gray]{.8}$r$\stab   \stab   \hspace*{-.6em}\\
        \hline
        5 &\hspace*{-.6em}   \stab   \stab   \stab\cellcolor[gray]{.8}$r$\stab   \hspace*{-.6em}\\
        \hline
        7 &\hspace*{-.6em}   \stab   \stab   \stab   \stab\cellcolor[gray]{.8}$r$\hspace*{-.6em}\\
        \hline
      \end{tabular}};
    \node[factor] (fact4) [right=of fact3]{};
    \node[factor] (fact5) [right=of fact4]{};
    \draw[->] (fact1) -- (fact2);
    \draw[->] (fact2) -- (fact3);
    \draw[->] (fact3) -- (fact4);
    \draw[->] (fact4) -- (fact5);
  \end{tikzpicture}\ ,
\end{align}
\begin{align}
  \begin{tikzpicture}
    [>=angle 90,factor/.style={rectangle, draw, inner sep=0pt, minimum size=1cm}
    ,baseline=(fact1)]
    \node (state) {$|\nu\rangle\!\rangle \in$};
    \node[factor] (fact1) [right=of state]{};
    \node[factor] (fact2) [right=of fact1]{};
    \node (fact3) [right=of fact2]{
      \scriptsize\begin{tabular}{c|c|c|c|c|c|}
          &\hspace*{-.6em} 0 \stab 1 \stab 2 \stab 5 \stab 7 \hspace*{-.6em}\\
        \hline
        0 &\hspace*{-.6em} \cellcolor[gray]{.8}1 \stab   \stab   \stab   \stab   \hspace*{-.6em}\\
        \hline
        1 &\hspace*{-.6em}   \stab\cellcolor[gray]{.8}$\ell$\stab   \stab   \stab   \hspace*{-.6em}\\
        \hline
        2 &\hspace*{-.6em}   \stab   \stab\cellcolor[gray]{.8}$\ell$\stab   \stab   \hspace*{-.6em}\\
        \hline
        5 &\hspace*{-.6em}   \stab   \stab   \stab\cellcolor[gray]{.8}$\ell$\stab   \hspace*{-.6em}\\
        \hline
        7 &\hspace*{-.6em}   \stab   \stab   \stab   \stab\cellcolor[gray]{.8}$\ell$\hspace*{-.6em}\\
        \hline
      \end{tabular}};
    \node[factor] (fact4) [right=of fact3]{};
    \node[factor] (fact5) [right=of fact4]{};
    \draw[->] (fact1) -- (fact2);
    \draw[->] (fact2) -- (fact3);
    \draw[->] (fact3) -- (fact4);
    \draw[->] (fact4) -- (fact5);
  \end{tikzpicture}\ ,
\end{align}
\begin{align}
  \begin{tikzpicture}
    [>=angle 90,factor/.style={rectangle, draw, inner sep=0pt, minimum size=1cm}
    ,baseline=(fact1)]
    \node (state) {$|\delta\rangle\!\rangle \in$};
    \node[factor] (fact1) [right=of state]{};
    \node[factor] (fact2) [right=of fact1]{};
    \node (fact3) [right=of fact2]{
      \scriptsize\begin{tabular}{c|c|c|c|c|c|}
          &\hspace*{-.6em} 0 \stab 1 \stab 2 \stab 5 \stab 7 \hspace*{-.6em}\\
        \hline
        0 &\hspace*{-.6em} 1 \stab   \stab   \stab   \stab   \hspace*{-.6em}\\
        \hline
        1 &\hspace*{-.6em}   \stab 2 \stab   \stab   \stab   \hspace*{-.6em}\\
        \hline
        2 &\hspace*{-.6em}   \stab   \stab \cellcolor[gray]{.8}2 \stab   \stab   \hspace*{-.6em}\\
        \hline
        5 &\hspace*{-.6em}   \stab   \stab   \stab\cellcolor[gray]{.8}$r$\stab   \hspace*{-.6em}\\
        \hline
        7 &\hspace*{-.6em}   \stab   \stab   \stab   \stab\cellcolor[gray]{.8}$\ell$\hspace*{-.6em}\\
        \hline
      \end{tabular}};
    \node[factor] (fact4) [right=of fact3]{};
    \node[factor] (fact5) [right=of fact4]{};
    \draw[->] (fact1) -- (fact2);
    \draw[->] (fact2) -- (fact3);
    \draw[->] (fact3) -- (fact4);
    \draw[->] (fact4) -- (fact5);
  \end{tikzpicture}\ .
\end{align}
This predicts the following overlaps up to a number of constants that cannot be determined
directly in this manner
\begin{align*}
  &\langle\!\langle{\rm id}|q^{L_0+\bar{L}_0}|{\rm id}\rangle\!\rangle= \tau^2 \big(
\alpha_0\,  \chi_{\mathcal{W}(0)}(q)+ \alpha_1\,\chi_{\mathcal{W}(1)}(q)
  + \alpha_2\,\chi_{\mathcal{W}(2)}(q)+\alpha_5 \,\chi_{\mathcal{W}(5)}(q)
  + \alpha_7\,\chi_{\mathcal{W}(7)}(q) \big)\\
  &\langle\!\langle{\rm id}|q^{L_0+\bar{L}_0}|\mu\rangle\!\rangle=
 \tau \Bigl( \beta^\mu_0\,  \chi_{\mathcal{W}(0)}(q)
 +  \beta^\mu_1\, \chi_{\mathcal{W}(1)}(q)
 +\beta^\mu_2\,\chi_{\mathcal{W}(2)}(q)
 +\beta^\mu_5\,\chi_{\mathcal{W}(5)}(q)
 +\beta^\mu_7\,\chi_{\mathcal{W}(7)}(q) \Bigr) \\
  &\langle\!\langle{\rm id}|q^{L_0+\bar{L}_0}|\nu\rangle\!\rangle=
   \tau \Bigl( \beta^\nu_0\, \chi_{\mathcal{W}(0)}(q)
   +\beta^\nu_1\, \chi_{\mathcal{W}(1)}(q)
  +\beta^\nu_2\,\chi_{\mathcal{W}(2)}(q)
  +\beta^\nu_5\,\chi_{\mathcal{W}(5)}(q)
  +\beta^\nu_7\, \chi_{\mathcal{W}(7)}(q) \Bigr)\\
  &\langle\!\langle{\rm id}|q^{L_0+\bar{L}_0}|\delta\rangle\!\rangle=
   \tau \Bigl( \beta^\delta_2\, \chi_{\mathcal{W}(2)}(q)
   +\beta^\delta_5\, \chi_{\mathcal{W}(5)}(q)
  +\beta^\delta_7\, \chi_{\mathcal{W}(7)}(q) \Bigr) \\
  &\langle\!\langle{\rm id}|q^{L_0+\bar{L}_0}|\sigma_i\rangle\!\rangle=
  \chi_{\mathcal{W}(i)}(q)\\
  &\langle\!\langle \mu|q^{L_0+\bar{L}_0}|\mu\rangle\!\rangle=\alpha^\mu \,  \chi_{\mathcal{W}(0)}(q) \ ,
  \qquad
  \langle\!\langle \nu|q^{L_0+\bar{L}_0}|\nu\rangle\!\rangle=\alpha^\nu \, \chi_{\mathcal{W}(0)}(q) \\
  &\langle\!\langle \delta|q^{L_0+\bar{L}_0}|\delta\rangle\!\rangle=
  \alpha^\delta \, \chi_{\mathcal{W}(2)}(q)\\
  &\langle\!\langle \mu|q^{L_0+\bar{L}_0}|\nu\rangle\!\rangle=
  \gamma^{\mu\nu}_0\, \chi_{\mathcal{W}(0)}(q)
  +  \gamma^{\mu\nu}_1\,  \chi_{\mathcal{W}(1)}(q)
  +   \gamma^{\mu\nu}_2\,  \chi_{\mathcal{W}(2)}(q)
  +   \gamma^{\mu\nu}_5\, \chi_{\mathcal{W}(5)}(q)
  +   \gamma^{\mu\nu}_7\, \chi_{\mathcal{W}(7)}(q)\\
  &\langle\!\langle \mu|q^{L_0+\bar{L}_0}|\delta\rangle\!\rangle
  =  \gamma^{\mu\delta}_2\,  \chi_{\mathcal{W}(2)}(q)
   + \gamma^{\mu\delta}_7\, \chi_{\mathcal{W}(7)}(q)\\
   & \langle\!\langle \nu|q^{L_0+\bar{L}_0}|\delta\rangle\!\rangle
 = \gamma^{\nu\delta}_2\, \chi_{\mathcal{W}(2)}(q)
 +  \gamma^{\nu\delta}_5\, \chi_{\mathcal{W}(5)}(q) \ .
\end{align*}
Some of these constants can be fixed by rescaling the Ishibashi states appropriately. 
Note that $\langle\!\langle \delta|q^{L_0+\bar{L}_0}|\delta\rangle\!\rangle$ probably
vanishes ({\it i.e.}\ $\alpha^\delta=0$) because the two contributions from the $2$'s in the middle 
sector appear to cancel against each other, see \eqref{eq:ishbasis}. 

Comparing to \eqref{eq:Ishibashioverlaps}, we see that the Ishibashi states that enter the 
boundary state analysis can be identified with
\be
  |0+\rangle\!\rangle=|\mu\rangle\!\rangle
  \qquad \text{and} \qquad
  |0-\rangle\!\rangle=|\nu\rangle\!\rangle
\ee
provided that 
$\alpha^\mu=\alpha^\nu=\frac12$, $\gamma^{\mu\nu}_0=\frac{1}{2}$, and $\gamma^{\mu\nu}_j=1$ for $ j=1,2,5,7$.
Thus we can again obtain all the relevant Ishibashi states within $\Hc_{\text{bulk}}$. 

As before we can also ask whether there are additional Ishibashi states that could
contribute to consistent boundary states. By a straightforward calculation, using the
modular transformation properties of Appendix~\ref{sec:S-matrix}, one can check that neither
$|{\rm id}\rangle\!\rangle$ nor $|\delta\rangle\!\rangle$ can appear in boundary states
since $\langle\!\langle{\rm id}|q^{L_0+\bar{L}_0}|{\rm id}\rangle\!\rangle$
and $\langle\!\langle \mu|q^{L_0+\bar{L}_0}|\delta\rangle\!\rangle$ both lead to $\tilde\tau$ terms
in the open string loop diagrams.
On the other hand, a similar argument 
does not apply to the Ishibashi states $|\sigma_i\rangle\!\rangle$ since their relative overlaps, as well as 
their overlaps with $|\mu\rangle\!\rangle$ and $|\nu\rangle\!\rangle$ all vanish. Thus
these Ishibashi states do not contribute to the cylinder diagrams, and we cannot 
decide whether they appear in boundary states based on cylinder diagrams alone. 

Given that the subgroup $\Kdd$ of the Grothendieck group is $12$-dimensional
but only ten Ishibashi states are 
required for the description of the cylinder diagrams, we suspect that at least certain linear 
combinations of the $|\sigma_i\rangle\!\rangle$ also contribute 
to the boundary states. In fact, this is even
required in order for the bulk-boundary map
to be non-degenerate in the bulk entry. To see this 
we observe that the bulk states from the top level in $\Hc_0$ only contribute in overlaps
with the $|\sigma_i\rangle\!\rangle$ Ishibashi states, see in particular eq.~(\ref{sigsupp}). Thus 
if the $|\sigma_i\rangle\!\rangle$ Ishibashi did not appear in any of the boundary states, 
all bulk-boundary maps would vanish on these top states, in contradiction with the
assumed non-degeneracy. However, since the  $|\sigma_i\rangle\!\rangle$ Ishibashi states
do not contribute to any cylinder diagrams the above analysis does not determine their 
coefficients.

\section{Conclusion}\label{sec:conclusion}

In this paper we have made a proposal for the bulk space
for the `charge conjugation' modular invariant of the logarithmic
${\cal W}_{2,3}$ triplet model. 
The basic idea of our construction is to obtain the bulk space from a given
boundary condition as the largest space for which a suitable bulk-boundary map is 
non-degenerate in the bulk entry. Unlike the situation for the logarithmic 
${\cal W}_{1,p}$-models where the starting point of the analysis was taken to be 
the `identity brane' whose space of boundary fields just consists of 
$\Wc_{1,p}$, no such brane exists for ${\cal W}_{2,3}$. We have therefore
not been able to perform this analysis starting from an actual boundary condition. Instead
we have argued that for the purpose of identifying the `charge-conjugation' bulk spectrum we may
take the space of boundary fields to consist just of $\Wc^\ast$, 
the representation conjugate to $\Wc$.
With this assumption the analysis could then be performed in a similar
way to what was done for the ${\cal W}_{1,p}$-models. The resulting bulk space 
satisfies a number of non-trivial consistency conditions: the space of bulk states is self-conjugate (as has to be the case in order for the 
bulk $2$-point function to be non-degenerate), the partition function is modular invariant, 
and we could identify boundary states whose annulus partition functions reproduce
the results of \cite{Gaberdiel:2009ug}. None of these consistency conditions are
part of our ansatz, and we therefore regard these tests as very
good evidence for the correctness of our proposal.

The resulting bulk space could be written as a quotient of the direct sum of
tensor products of the projective covers of the irreducible representations. Furthermore,
the representatives of the quotient space could be taken to be described by the direct
sum of tensor products of each irreducible representation with its projective cover. 
The bulk spectrum has therefore exactly the same structure as for the ${\cal W}_{1,p}$-models 
\cite{Gaberdiel:2007jv}, or the supergroup models of \cite{Quella:2007hr}. It is plausible
that the `charge-conjugation' theory of a generic (logarithmic) conformal field theory
may therefore have this structure. 

On the other hand, there are also important differences to what was found for the 
${\cal W}_{1,p}$-models. One surprising property is that the bulk space does not 
contain $\Wc_{2,3}\otimes_{\mathbb C} \bar{\Wc}_{2,3}$ as a sub-representation, but
only as a sub-quotient. The other unusual feature is that there are non-trivial Ishibashi states
that do not contribute to annulus amplitudes. As a consequence, the boundary states of the
various boundary conditions could not be fixed uniquely by these amplitudes alone. It would
be interesting to see whether there are other (easily accessible) amplitudes that would allow
one to determine the boundary states uniquely. One of us has recently studied the action of 
topological defects on these boundary states \cite{defects}, but, unfortunately, this does not seem 
to constrain the boundary states any further. (However, it is encouraging that the
defect analysis also seems to work nicely for this theory.)

On a technical level, probably the most tricky part of our analysis was the description 
of the projective covers of all irreducible representations, see Appendix~\ref{sec:embedding} 
and in particular Appendix~\ref{app:P0-cover}. Some aspects of the corresponding composition 
series could be deduced from the conjectures of \cite{Rasmussen:2008ii}, 
while for the determination of the projective cover $\Pc(0)$ 
of the trivial representation $\Wc(0)$ we had to use other arguments, in particular the 
recent calculation of Zhu's algebra \cite{AM}. The fact that everything
fits together nicely makes us confident that the composition series of these projective covers
are indeed correct, but a direct confirmation, for example by using
the Coulomb gas description of \cite{Nagatomo:2009xp,AM1}, is still missing.
\smallskip

As far as we know, the proposal of this paper is the first example of an interesting\footnote{Here
`interesting' refers to the fact that the theory is not just the tensor product of 
rational theories whose total central charge adds up to zero.}
(seemingly) consistent bulk theory with a discrete spectrum at $c=0$
(there are earlier examples with continuous spectrum, see \cite{Saleur:2006tf}).
It would therefore be very  instructive to study its properties further. For example, it would be interesting 
to determine the relevant critical exponents, and to understand which correlators are
logarithmic, {\it etc}. In particular, this may give some insight into the structure of closely related
bulk theories, such as those describing polymers or percolation.

\section*{Acknowledgements}

We thank Alexei Davydov, J\"urgen Fuchs, Azat Gainutdinov, Volker Schomerus, Christoph Schwei\-gert 
and Akihiro Tsuchiya for useful conversations. 
The research of MRG and SW is partially supported by a grant from the Swiss
National Science Foundation.
The work of IR is supported by the German Science Foundation (DFG) 
within the Collaborative Research Center 676 ``Particles, Strings and the Early Universe''. 
MRG and IR thank the Erwin-Schr\"odinger-Institute in Vienna for hospitality during the workshop
``Quantum field theory in curved spacetimes and target-spaces'', where part of this work was completed.

\appendix

\section{\texorpdfstring{$\boldsymbol{\Wc_{2,3}}$}{W2,3}-representations}\label{app:W23-reps}

\subsection{Composition series}\label{sec:embedding}

Here we list the embedding diagrams for the projective covers $\Pc(0)$, 
$\Pc(1)={\cal R}^{(3)}(0,0,1,1)$, $\Pc(2)={\cal R}^{(3)}(0,0,2,2)$, $\Pc(5)={\cal R}^{(3)}(0,1,2,5)$ and 
$\Pc(7)={\cal R}^{(3)}(0,1,2,7)$. The structure of $\Pc(0)$ is explained in 
Appendix~\ref{app:P0-cover}, while the structure of the other projective covers was deduced from 
the conjectures of \cite{Rasmussen:2008ii} and from the analogy with quantum group representations
described in \cite[Sect.\,6.1]{Feigin:2006iv} and \cite[Sect.\,2.2]{Feigin:2006xa}.  
\begin{align}\label{eq:P0-diagram}
  \begin{tikzpicture}
    [baseline=(lev33),scale=1, vert/.style={inner sep=1,
      text centered,circle,anchor=base}]
    \node at (0,1) {$\Pc(0)$};
    \path
    (0,0) node[vert] (lev11) {0}
    ++(-0.5,-1) node[vert] (lev21) {1} edge[-] (lev11)
    ++(1,0) node[vert] (lev22) {2} edge[-] (lev11)
    ++(1.5,-1) node[vert] (lev35) {7} edge[-] (lev21) edge[-] (lev22)
    ++(-1,0) node[vert] (lev34) {7} edge[-] (lev21) edge[-] (lev22)
    ++(-1,0) node[vert] (lev33) {0} edge[-] (lev21) edge[-] (lev22)
    ++(-1,0) node[vert] (lev32) {5} edge[-] (lev21) edge[-] (lev22)
    ++(-1,0) node[vert] (lev31) {5} edge[-] (lev21) edge[-] (lev22)
    ++(1.5,-1) node[vert] (lev41) {2} edge[-] (lev31) edge[-] (lev32) edge[-] (lev33) edge[-] (lev34) edge[-] (lev35)
    ++(1,0) node[vert] (lev42) {1} edge[-] (lev31) edge[-] (lev32) edge[-] (lev33) edge[-] (lev34) edge[-] (lev35)
    ++(-0.5,-1) node[vert] (lev51) {0} edge[-] (lev41) edge[-] (lev42);
    \path
    (-3,1) node[vert] {Level}
    ++(0,-1) node[vert] {0}
    ++(0,-1) node[vert] {1}
    ++(0,-1) node[vert] {2}
    ++(0,-1) node[vert] {3}
    ++(0,-1) node[vert] {4};
  \end{tikzpicture}
\end{align}
\begin{align}
  \begin{tikzpicture}
    [vert/.style={inner sep=1, text centered,circle,anchor=base},baseline=(B)]
    \node at (0,1) {$\Pc(1)$};
    \path
    (0,0) node[vert] (A) {1}
    ++(-2,-1) node[vert] (B) {7} edge (A)
    ++(1,0) node[vert] (C) {7} edge (A)
    ++(1,0) node[vert] (D) {0} edge (A)
    ++(1,0) node[vert] (E) {5} edge (A)
    ++(1,0) node[vert] (F) {5} edge (A)
    ++(0.5,-1) node[vert] (G) {1} edge (F) edge (E) edge (D)
    ++(-1,0) node[vert] (H) {2} edge (F) edge (C)
    ++(-1,0) node[vert] (I) {2} edge (F) edge (D) edge (B)
    ++(-1,0) node[vert] (J) {2} edge (C) edge (D) edge (E)
    ++(-1,0) node[vert] (K) {2} edge (B) edge (E)
    ++(-1,0) node[vert] (L) {1} edge (B) edge (C) edge (D)
    ++(0.5,-1) node[vert] (M) {5} edge (L) edge (K) edge (J)
    ++(1,0) node[vert] (N) {5} edge (L) edge (I) edge (H)
    ++(1,0) node[vert] (O) {0} edge (L) edge (J) edge (I) edge (G)
    ++(1,0) node[vert] (P) {7} edge (K) edge (I) edge (G)
    ++(1,0) node[vert] (Q) {7} edge (J) edge (H) edge (G)
    ++(-2,-1) node[vert] (R) {1} edge (M) edge (N) edge (O) edge (P) edge (Q);
        \path
    (-3.5,1) node[vert] {Level}
    ++(0,-1) node[vert] {0}
    ++(0,-1) node[vert] {1}
    ++(0,-1) node[vert] {2}
    ++(0,-1) node[vert] {3}
    ++(0,-1) node[vert] {4};
  \end{tikzpicture}\qquad
  \begin{tikzpicture}
    [vert/.style={inner sep=1, text centered,circle,anchor=base},baseline=(B)]
    \node at (0,1) {$\Pc(2)$};
    \path
    (0,0) node[vert] (A) {2}
    ++(-2,-1) node[vert] (B) {5} edge (A)
    ++(1,0) node[vert] (C) {5} edge (A)
    ++(1,0) node[vert] (D) {0} edge (A)
    ++(1,0) node[vert] (E) {7} edge (A)
    ++(1,0) node[vert] (F) {7} edge (A)
    ++(0.5,-1) node[vert] (G) {2} edge (F) edge (E) edge (D)
    ++(-1,0) node[vert] (H) {1} edge (F) edge (C)
    ++(-1,0) node[vert] (I) {1} edge (F) edge (D) edge (B)
    ++(-1,0) node[vert] (J) {1} edge (C) edge (D) edge (E)
    ++(-1,0) node[vert] (K) {1} edge (B) edge (E)
    ++(-1,0) node[vert] (L) {2} edge (B) edge (C) edge (D)
    ++(0.5,-1) node[vert] (M) {7} edge (L) edge (K) edge (J)
    ++(1,0) node[vert] (N) {7} edge (L) edge (I) edge (H)
    ++(1,0) node[vert] (O) {0} edge (L) edge (J) edge (I) edge (G)
    ++(1,0) node[vert] (P) {5} edge (K) edge (I) edge (G)
    ++(1,0) node[vert] (Q) {5} edge (J) edge (H) edge (G)
    ++(-2,-1) node[vert] (R) {2} edge (M) edge (N) edge (O) edge (P) edge (Q);
  \end{tikzpicture}
\end{align}
\begin{align}\label{eq:P57-diagram}
  \begin{tikzpicture}
    [vert/.style={inner sep=1, text centered,circle,anchor=base},baseline=(B)]
    \node at (0,1) {$\Pc(5)$};
    \path
    (0,0) node[vert] (A) {5}
    ++(-2,-1) node[vert] (B) {2} edge (A)
    ++(0.8,0) node[vert] (C) {2} edge (A)
    ++(2.4,0) node[vert] (D) {1} edge (A)
    ++(0.8,0) node[vert] (E) {1} edge (A)
    ++(0.8,-1) node[vert] (F) {5} edge (E) edge (D)
    ++(-0.8,0) node[vert] (G) {7} edge (C) edge (E)
    ++(-0.8,0) node[vert] (H) {7} edge (C) edge (D)
    ++(-0.8,0) node[vert] (I) {0} edge (C) edge (E)
    ++(-0.8,0) node[vert] (J) {0} edge (B) edge (D)
    ++(-0.8,0) node[vert] (K) {7} edge (B) edge (E)
    ++(-0.8,0) node[vert] (L) {7} edge (B) edge (D)
    ++(-0.8,0) node[vert] (M) {5} edge (B) edge (C)
    ++(0.8,-1) node[vert] (N) {1} edge (M) edge (L) edge (J) edge (H)
    ++(0.8,0) node[vert] (O) {1} edge (M) edge (K) edge (I) edge (G)
    ++(2.4,0) node[vert] (P) {2} edge (L) edge (K) edge (J) edge (F)
    ++(0.8,0) node[vert] (Q) {2} edge (I) edge (H) edge (G) edge (F)
    ++(-2,-1) node[vert] (R) {5} edge (N) edge (O) edge (P) edge (Q);
            \path
    (-3.6,1) node[vert] {Level}
    ++(0,-1) node[vert] {0}
    ++(0,-1) node[vert] {1}
    ++(0,-1) node[vert] {2}
    ++(0,-1) node[vert] {3}
    ++(0,-1) node[vert] {4};
  \end{tikzpicture}\
  \begin{tikzpicture}
    [vert/.style={inner sep=1, text centered,circle,anchor=base},baseline=(B)]
    \node at (0,1) {$\Pc(7)$};
    \path
    (0,0) node[vert] (A) {7}
    ++(-2,-1) node[vert] (B) {1} edge (A)
    ++(0.8,0) node[vert] (C) {1} edge (A)
    ++(2.4,0) node[vert] (D) {2} edge (A)
    ++(0.8,0) node[vert] (E) {2} edge (A)
    ++(0.8,-1) node[vert] (F) {7} edge (E) edge (D)
    ++(-0.8,0) node[vert] (G) {5} edge (C) edge (E)
    ++(-0.8,0) node[vert] (H) {5} edge (C) edge (D)
    ++(-0.8,0) node[vert] (I) {0} edge (C) edge (E)
    ++(-0.8,0) node[vert] (J) {0} edge (B) edge (D)
    ++(-0.8,0) node[vert] (K) {5} edge (B) edge (E)
    ++(-0.8,0) node[vert] (L) {5} edge (B) edge (D)
    ++(-0.8,0) node[vert] (M) {7} edge (B) edge (C)
    ++(0.8,-1) node[vert] (N) {2} edge (M) edge (L) edge (J) edge (H)
    ++(0.8,0) node[vert] (O) {2} edge (M) edge (K) edge (I) edge (G)
    ++(2.4,0) node[vert] (P) {1} edge (L) edge (K) edge (J) edge (F)
    ++(0.8,0) node[vert] (Q) {1} edge (I) edge (H) edge (G) edge (F)
    ++(-2,-1) node[vert] (R) {7} edge (N) edge (O) edge (P) edge (Q);
  \end{tikzpicture}
\end{align}
Looking at the embedding diagrams from top to bottom, we say the 
representations at the top are at level zero, the representations in the second line from the top
are at level one, {\it etc}. (This is indicated in the above diagrams by the small number written
on the left.) 

Let us explain the meaning of the lines in these diagrams by taking $\Pc(2)$ as an example.
The embedding diagram for $\Pc(2)$ states that one can 
choose a direct sum decomposition of $\Pc(2)$ as a {\em vector space} as
\begin{equation}\label{a.4}
\Pc(2) \cong 2 \Wc(0) \oplus 4 \Wc(1) \oplus 4 \Wc(2) \oplus 4 \Wc(5) \oplus 4 \Wc(7) \ .
\end{equation}
Furthermore, the action of a mode in $\Wc_{2,3}$ on a state in 
$\Wc(h)$ gives a vector in the direct sum of the given $\Wc(h)$ and all 
summands joined to it by downward directed lines (in one or more steps). In particular,
these diagrams are therefore in general 
{\em not} unique. For example, there must also exist a different direct sum decomposition of $\Pc(2)$ 
as in (\ref{a.4}) such that 
$\Wc(0)$ at level 1 is joined only to one $\Wc(1)$ and one $\Wc(2)$ at level $2$. 
(This has to be possible for there to be an intertwiner 
$\Pc(0) \rightarrow \Pc(2)$ whose image contains $\Wc(0) \subset \Pc(2)$ at level $1$.)

For completeness we also give the embedding diagrams for the other indecomposable
representations that appear in our analysis:
\begin{align}
  \begin{tabular}[t]{cc}
\begin{tikzpicture}
  [vert/.style={inner sep=1, text centered,circle,anchor=base},minimum width=0.5cm,baseline=(A)]
  \path
  (0,0) node[vert] (A) {$a$}
  ++(-1,-1) node[vert] (B) {$b$} edge (A)
  ++(2,0) node[vert] (C) {$b$} edge (A)
  ++(-1,-1) node[vert] (D) {$a$} edge (B) edge (C);
\end{tikzpicture}
&
    $
    \begin{array}[t]{rlrl}
      \mathcal{R}^{(2)}(\frac13,\frac13):& a=\frac13,\ b=\frac{10}{3}&      \mathcal{R}^{(2)}(2,7):& a=7,\ b=2\\
      \mathcal{R}^{(2)}(\frac13,\frac{10}{3}):& a=\frac{10}{3},\ b=\frac13&      \mathcal{R}^{(2)}(2,5):& a=5,\ b=2\\
      \mathcal{R}^{(2)}(\frac58,\frac58):& a=\frac58,\ b=\frac{21}{8}&      \mathcal{R}^{(2)}(1,7):& a=7,\ b=1\\
      \mathcal{R}^{(2)}(\frac58,\frac{21}{8}):& a=\frac{21}{8},\ b=\frac{5}{8}&      \mathcal{R}^{(2)}(1,5):& a=5,\ b=1\\
      \mathcal{R}^{(2)}(\frac{1}{8},\frac{1}{8}):& a=\frac{1}{8},\ b=\frac{33}{8}\\
      \mathcal{R}^{(2)}(\frac{1}{8},\frac{33}{8}):& a=\frac{33}{8},\ b=\frac{1}{8}\\
    \end{array}
    $
  \end{tabular}
\end{align}
and
\begin{align}
  \begin{tabular}[t]{cc}
\begin{tikzpicture}
  [vert/.style={inner sep=1, text centered,circle,anchor=base,minimum width=0.5cm},baseline=(A)]
  \path
  (0,0) node[vert] (A) {$a$}
  ++(-1,-1) node[vert] (B) {$b$} edge (A)
  ++(1,0) node[vert] (0) {$0$} edge (A)
  ++(1,0) node[vert] (C) {$b$} edge (A)
  ++(-1,-1) node[vert] (D) {$a$} edge (B) edge (C) edge (0);
\end{tikzpicture}
    &
    $
    \begin{array}[t]{cc}
      \mathcal{R}^{(2)}(0,2)_7:& a=2,\ b=7\\
      \mathcal{R}^{(2)}(0,2)_5:& a=2,\ b=5\\
      \mathcal{R}^{(2)}(0,1)_7:& a=1,\ b=7\\
      \mathcal{R}^{(2)}(0,1)_5:& a=1,\ b=5
    \end{array}
    $
  \end{tabular}
\end{align}
For the analysis of the Ishibashi states we also need to introduce some notation for the 
various intertwiners. We say an intertwiner $e:\Pc(a)\rightarrow \Pc(b)$ 
has degree $l$ if it adds $l$ to the level of the states, {\it i.e.}\ if it maps the representation
$\Wc(a)$ at level zero to a representation at level $l$ in $\Pc(b)$. 
We denote such an intertwiner by $e^{a\rightarrow b}_{l}$. 
Sometimes, this information does not specify an intertwiner uniquely; for example, there are two
intertwiners of degree $1$ corresponding to $e^{1\rightarrow 5}_{1}$. Whenever an intertwiner 
is not uniquely specified in this manner, we introduce an additional label,
{\it i.e.}\ we write $e^{1\rightarrow 5}_{1;\alpha}$ where $\alpha$ takes two values. In the
present context, it is natural to take $\alpha\in\{L,R\}$, where $e^{1\rightarrow 5}_{1;R}$
is the intertwiner that maps the circled nodes in the following diagram to one another,
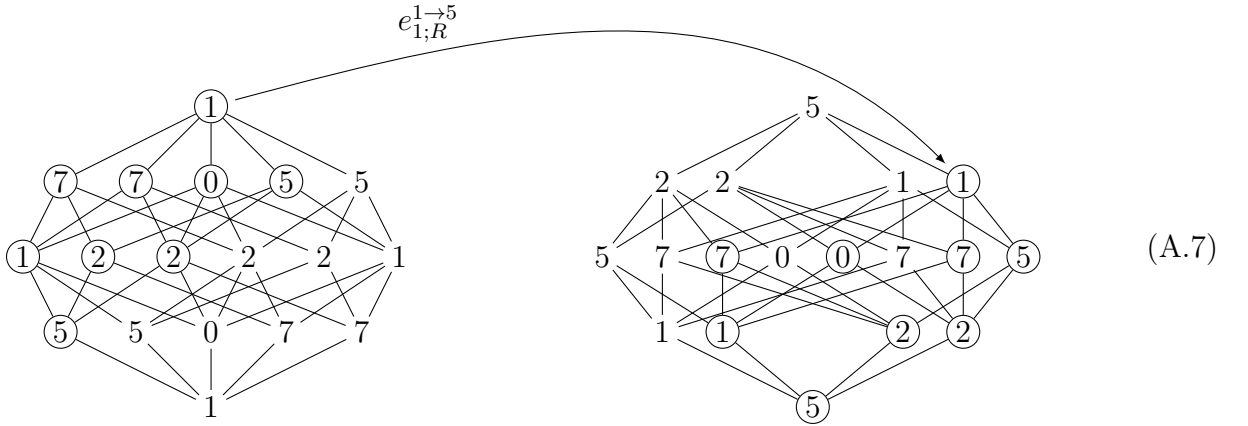
\begin{align}
  \begin{tikzpicture}
    [vert/.style={inner sep=1, text centered,circle,anchor=base},baseline=(GL)]
    \path
    (0,0) node[vert,draw] (AL) {1}
    ++(-2,-1) node[vert,draw] (BL) {7} edge (AL)
    ++(1,0) node[vert,draw] (CL) {7} edge (AL)
    ++(1,0) node[vert,draw] (DL) {0} edge (AL)
    ++(1,0) node[vert,draw] (EL) {5} edge (AL)
    ++(1,0) node[vert] (FL) {5} edge (AL)
    ++(0.5,-1) node[vert] (GL) {1} edge (FL) edge (EL) edge (DL)
    ++(-1,0) node[vert] (HL) {2} edge (FL) edge (CL)
    ++(-1,0) node[vert] (IL) {2} edge (FL) edge (DL) edge (BL)
    ++(-1,0) node[vert,draw] (JL) {2} edge (CL) edge (DL) edge (EL)
    ++(-1,0) node[vert,draw] (KL) {2} edge (BL) edge (EL)
    ++(-1,0) node[vert,draw] (LL) {1} edge (BL) edge (CL) edge (DL)
    ++(0.5,-1) node[vert,draw] (ML) {5} edge (LL) edge (KL) edge (JL)
    ++(1,0) node[vert] (NL) {5} edge (LL) edge (IL) edge (HL)
    ++(1,0) node[vert] (OL) {0} edge (LL) edge (JL) edge (IL) edge (GL)
    ++(1,0) node[vert] (PL) {7} edge (KL) edge (IL) edge (GL)
    ++(1,0) node[vert] (QL) {7} edge (JL) edge (HL) edge (GL)
    ++(-2,-1) node[vert] (RL) {1} edge (ML) edge (NL) edge (OL) edge (PL) edge (QL);
    \path
    (8,0) node[vert] (AR) {5}
    ++(-2,-1) node[vert] (BR) {2} edge (AR)
    ++(0.8,0) node[vert] (CR) {2} edge (AR)
    ++(2.4,0) node[vert] (DR) {1} edge (AR)
    ++(0.8,0) node[vert,draw] (ER) {1} edge (AR)
    ++(0.8,-1) node[vert,draw] (FR) {5} edge (ER) edge (DR)
    ++(-0.8,0) node[vert,draw] (GR) {7} edge (CR) edge (ER)
    ++(-0.8,0) node[vert] (HR) {7} edge (CR) edge (DR)
    ++(-0.8,0) node[vert,draw] (IR) {0} edge (CR) edge (ER)
    ++(-0.8,0) node[vert] (JR) {0} edge (BR) edge (DR)
    ++(-0.8,0) node[vert,draw] (KR) {7} edge (BR) edge (ER)
    ++(-0.8,0) node[vert] (LR) {7} edge (BR) edge (DR)
    ++(-0.8,0) node[vert] (MR) {5} edge (BR) edge (CR)
    ++(0.8,-1) node[vert] (NR) {1} edge (MR) edge (LR) edge (JR) edge (HR)
    ++(0.8,0) node[vert,draw] (OR) {1} edge (MR) edge (KR) edge (IR) edge (GR)
    ++(2.4,0) node[vert,draw] (PR) {2} edge (LR) edge (KR) edge (JR) edge (FR)
    ++(0.8,0) node[vert,draw] (QR) {2} edge (IR) edge (HR) edge (GR) edge (FR)
    ++(-2,-1) node[vert,draw] (RR) {5} edge (NR) edge (OR) edge (PR) edge (QR);
    \draw[->,shorten >=3pt,shorten <=3pt,>=latex] (AL) to [out=15,in=135] (ER) node [near start,above] {$e^{1\rightarrow 5}_{1;R}$};
  \end{tikzpicture}
\end{align}
and similarly for $e^{1\rightarrow 5}_{1;L}$. 
For the case of  the intertwiners from $\Pc(h)\rightarrow \Pc(h)$ with $h\in\{1,2,5,7\}$, it is 
convenient to  label the different intertwiners at level $2$ by the representations to which the image 
node at level $2$ is connected, {\it e.g.}\ the intertwiner $e^{1\rightarrow 1}_{2;5}$ maps the circled
nodes to one another, {\it etc.}
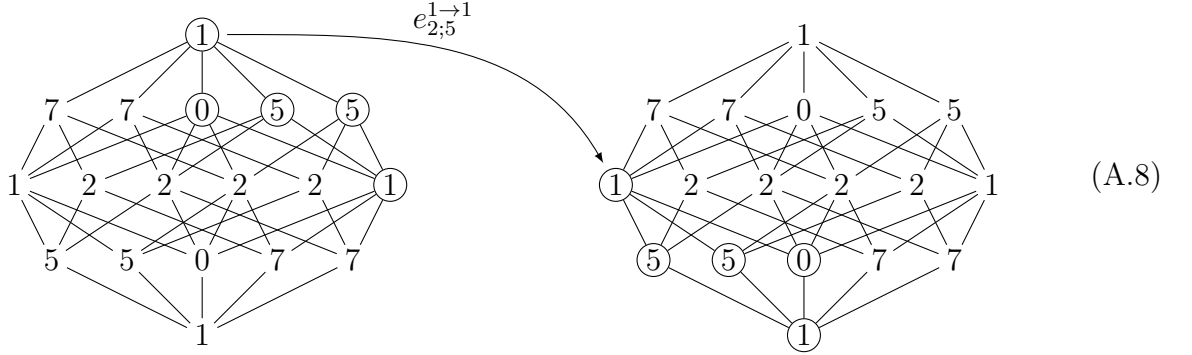
\begin{align}
  \begin{tikzpicture}
    [vert/.style={inner sep=1, text centered,circle,anchor=base},baseline=(Gl)]
    \path
    (0,0) node[vert,draw] (Al) {1}
    ++(-2,-1) node[vert] (Bl) {7} edge (Al)
    ++(1,0) node[vert] (Cl) {7} edge (Al)
    ++(1,0) node[vert,draw] (Dl) {0} edge (Al)
    ++(1,0) node[vert,draw] (El) {5} edge (Al)
    ++(1,0) node[vert,draw] (Fl) {5} edge (Al)
    ++(0.5,-1) node[vert,draw] (Gl) {1} edge (Fl) edge (El) edge (Dl)
    ++(-1,0) node[vert] (Hl) {2} edge (Fl) edge (Cl)
    ++(-1,0) node[vert] (Il) {2} edge (Fl) edge (Dl) edge (Bl)
    ++(-1,0) node[vert] (Jl) {2} edge (Cl) edge (Dl) edge (El)
    ++(-1,0) node[vert] (Kl) {2} edge (Bl) edge (El)
    ++(-1,0) node[vert] (Ll) {1} edge (Bl) edge (Cl) edge (Dl)
    ++(0.5,-1) node[vert] (Ml) {5} edge (Ll) edge (Kl) edge (Jl)
    ++(1,0) node[vert] (Nl) {5} edge (Ll) edge (Il) edge (Hl)
    ++(1,0) node[vert] (Ol) {0} edge (Ll) edge (Jl) edge (Il) edge (Gl)
    ++(1,0) node[vert] (Pl) {7} edge (Kl) edge (Il) edge (Gl)
    ++(1,0) node[vert] (Ql) {7} edge (Jl) edge (Hl) edge (Gl)
    ++(-2,-1) node[vert] (Rl) {1} edge (Ml) edge (Nl) edge (Ol) edge (Pl) edge  (Ql);
    \path
    (8,0) node[vert] (Ar) {1}
    ++(-2,-1) node[vert] (Br) {7} edge (Ar)
    ++(1,0) node[vert] (Cr) {7} edge (Ar)
    ++(1,0) node[vert] (Dr) {0} edge (Ar)
    ++(1,0) node[vert] (Er) {5} edge (Ar)
    ++(1,0) node[vert] (Fr) {5} edge (Ar)
    ++(0.5,-1) node[vert] (Gr) {1} edge (Fr) edge (Er) edge (Dr)
    ++(-1,0) node[vert] (Hr) {2} edge (Fr) edge (Cr)
    ++(-1,0) node[vert] (Ir) {2} edge (Fr) edge (Dr) edge (Br)
    ++(-1,0) node[vert] (Jr) {2} edge (Cr) edge (Dr) edge (Er)
    ++(-1,0) node[vert] (Kr) {2} edge (Br) edge (Er)
    ++(-1,0) node[vert,draw] (Lr) {1} edge (Br) edge (Cr) edge (Dr)
    ++(0.5,-1) node[vert,draw] (Mr) {5} edge (Lr) edge (Kr) edge (Jr)
    ++(1,0) node[vert,draw] (Nr) {5} edge (Lr) edge (Ir) edge (Hr)
    ++(1,0) node[vert,draw] (Or) {0} edge (Lr) edge (Jr) edge (Ir) edge (Gr)
    ++(1,0) node[vert] (Pr) {7} edge (Kr) edge (Ir) edge (Gr)
    ++(1,0) node[vert] (Qr) {7} edge (Jr) edge (Hr) edge (Gr)
    ++(-2,-1) node[vert,draw] (Rr) {1} edge (Mr) edge (Nr) edge (Or) edge (Pr) edge (Qr);
    \draw[->,shorten >=3pt,shorten <=3pt,>=latex] (Al) to [out=0,in=120] (Lr) node [midway,above] {$e^{1\rightarrow 1}_{2;5}$};
  \end{tikzpicture}
\end{align}
These diagrams then allow one to compute the composition of intertwiners. For example
we have 
$e^{1\rightarrow 5}_{1;R} \circ e^{1\rightarrow 1}_{2;5} = e^{1\rightarrow 5}_{3;R}$, where 
$e^{1\rightarrow 5}_{3;R}$ maps the top $\Wc(1)$ of $\Pc(1)$ to the `right' $\Wc(1)$ at level $3$ of $\Pc(5)$. 
By the same reasoning, $e^{1\rightarrow 7}_{1;L/R} \circ e^{1\rightarrow 1}_{2;5} = 0$.

\subsection{Modular transformations}\label{sec:S-matrix}

In this appendix we briefly review the modular $S$-matrix of the $\Wc_{2,3}$-model following
\cite{Semikhatov:2007qp}\footnote{There 
are two small typos in \cite{Semikhatov:2007qp} ({\tt arXiv v2}): in (2.5) there should be an 
additional $(-1)^{p p'}$ in front of $\chi_{s,s'}^-$, and in the expression for 
$\widetilde{{\cal S}}_{r,r';s,s'}^+(\tau)$, the summand $-2 i \pi p p' \tau$ should read $- i p p' \tau / (2\pi)$.}.
To keep the notation for the $S$-modular transformation of the 
characters of all irreducible 
representations compact, we will label each character by the position of its corresponding
irreducible representation in the Kac table \eqref{eq:W23-irreps}. Because each cell of the
Kac table contains two or three entries, we will denote the character of the right most
irreducible representation in the cell by $\chi_{(r,s,-)}$ and the second to right 
representations $\chi_{(r,s,+)}$. The two cells in the interior of the Kac table both contain
an additional weight 0 representation; this is just $\Wc(0)$ and we will continue to refer
to the character of $\Wc(0)$ by $\chi_{\Wc(0)}$.

The $S$-modular transformation of the characters is then, for $\varepsilon = \pm$,
\begin{align}
  \chi_{\Wc(0)}(-\tfrac{1}{\tau})&=\chi_{\Wc(0)}(\tau) \nonumber\\
  \chi_{(\rho,\sigma,\varepsilon)}(-\tfrac{1}{\tau})&= \widetilde\Sc_{\rho,\sigma}^\varepsilon(\tau)\chi_{\Wc(0)}(\tau) +
  \sum_{r=1}^2\sum_{s=1}^3 \varepsilon^{r}
  \Sc_{\rho,\sigma;r,s}(\tau) \Big(\chi_{(r,s,+)}(\tau)+(-1)^{\rho}\chi_{(r,s,-)}(\tau) \Big)
\end{align}
where $\rho \in \{1,2\}$, $\sigma \in \{1,2,3\}$ and the $\tau$ dependant $S$-matrix coefficients
are given by
\begin{align}
  \Sc_{\rho,\sigma;1,s}(\tau)&=\frac{1}{3\sqrt{3}}(-1)^{\rho s+\sigma}
  \left(\rho\cos\tfrac{\pi 3\rho}{2}-i\tau \sin\tfrac{\pi 3\rho}{2})
  \right) \,
  \left(\sigma\cos\tfrac{\pi 2\sigma s}{3}-i\tau(3{-}s)
    \sin\tfrac{\pi 2 \sigma s}{3}\right)\,,&s \in \{1,2\} \nonumber\\
  \Sc_{\rho,\sigma;1,3}(\tau)&=\frac{\sigma}{6\sqrt{3}}
  (-1)^{\sigma+\rho}\left(\rho \cos\tfrac{\pi 3\rho }{2}
    -i\tau \sin\tfrac{\pi 3\rho }{2}\right)\,, \nonumber\\
  \Sc_{\rho,\sigma;2,s}(\tau)&=\frac{\rho}{6\sqrt{3}}
  (-1)^{(s+1)\rho}\left( \sigma \cos\tfrac{\pi 2\sigma s}{3}
    -i\tau(3{-}s)\sin\tfrac{\pi 2\sigma s}{3}\right)\,,&s \in \{1,2\} \nonumber\\
  \Sc_{\rho,\sigma;2,3}&=\frac{\rho\sigma}{12\sqrt{3}} \nonumber
\end{align}
\begin{align}
  \widetilde\Sc_{\rho,\sigma}^+(\tau)&=(-1)^{\rho+\sigma}\frac{1}{36\sqrt{3}}
  \bigg(6\rho \sigma\cos\tfrac{\pi 3\rho}{2}\cos\tfrac{\pi 2 \sigma}{3}
-i3\rho\tau\cos\tfrac{\pi 3\rho}{2}\sin\tfrac{\pi 2\sigma}{3}
  +i2\sigma\tau\cos\tfrac{\pi 2\sigma}{3}\sin\tfrac{\pi 3\rho}{2} \nonumber\\
  &\hspace{15em}\left.\phantom{=}+\left(\frac{1}{2}\tau^2-\frac{i3\tau}{\pi}
    +\frac{4\sigma^2+9\rho^2}{2}\right)
  \sin\tfrac{\pi 3\rho}{2}\sin\tfrac{\pi 2\sigma}{3}\right) \nonumber\\
  \widetilde\Sc^-_{\rho,\sigma}&=-\widetilde\Sc_{\rho,\sigma}^+(\tau)
  +(-1)^{\rho+\sigma }\frac{1}{\sqrt{12}}
  \sin\tfrac{\pi 3 \rho}{2}\sin\tfrac{\pi 2\sigma}{3}\ .
\end{align}

\section{Projective cover of \texorpdfstring{$\boldsymbol{\Wc(0)}$}{W(0)}}
\label{app:P0-cover}

Given the results of \cite[Thm.\,3.24]{Huang:2007mj} it seems plausible that $\Wc(0)$ must
have a projective cover $\Pc(0)$. In this appendix we want to deduce the structure of $\Pc(0)$,
assuming its existence, and using input from a recent calculation of Zhu's algebra \cite{AM}. 
For the following it will be convenient to introduce the abbreviation 
\be
  \big\langle U , V \big\rangle = \dim\Hom(U,V) \ ,
\ee
where $U$ and $V$ are two $\Wc$-representations. As was argued in \cite[Sect.\,3.1]{Gaberdiel:2009ug}
we have isomorphisms $\pi_{U,V} : \Hom(U,V^*) \rightarrow \Hom(U \otimes_f V, \Wc^*)$ that are 
natural in $U$ and $V$ and thus
\be\label{eq:P0-fusion-calc-conditions}
\big\langle U , V^* \big\rangle = 
  \big\langle U \otimes_f V , \Wc^* \big\rangle  \ .
\ee
Finally, we note that  $\dim\Hom(\Pc(h),U)$ gives the number of occurrences of $\Wc(h)$ 
in the composition series of $U$, {\it i.e.}
\begin{equation}\label{eq:P0-fusion-calc-conditions1}
  \big\langle \Pc(h) , U \big\rangle = \big( \text{$\#$ of $\Wc(h)$ in comp.\ series of $U$}\big)
\end{equation}

\subsection{Fusion rules} \label{app:P0xP0}

Even without knowledge of the structure of $\Pc(0)$, we can
partially deduce its fusion rules. 

\subsubsection*{Fusion with $\Wc(0)$}

It was argued in \cite[Sect.\,2]{Gaberdiel:2009ug}, the fusion of $\Wc(0)$ with any 
$\Wc$-representation can only produce multiples of $\Wc(0)$. Thus 
$\Pc(0) \otimes_f \Wc(0) \cong n \Wc(0)$ for some $n \ge 0$.  Because of 
$\dim\Hom(\Pc(0),\Wc(0))=1$, it follows from \eqref{eq:P0-fusion-calc-conditions} and
\eqref{eq:P0-fusion-calc-conditions1} that 
\be
  1 
  = \big\langle \Pc(0) , \Wc(0) \big\rangle 
  = \big\langle \Pc(0) \otimes_f \Wc(0) , \Wc^* \big\rangle 
  = n \big\langle  \Wc(0) , \Wc^* \big\rangle 
  = n \ .
\ee  
Thus we conclude that 
\be
  \Pc(0) \otimes_f \Wc(0) \cong \Wc(0) \ .
\labl{eq:fus-P0-W0}

\subsubsection*{Fusion with representations $U$ that have a dual}

Given an object $U$ in a monoidal category $\Cc$, we call an object 
$U^\vee \in \Cc$ a {\em (left and right) dual} 
of $U$ iff there are intertwiners
\begin{align}
  \label{eq:duality-maps}
  b_U &: \one \rightarrow U \otimes_f U^\vee \ ,& 
  \tilde b_U &: \one \rightarrow U^\vee \otimes_f U \ ,\\
  d_U &: U^\vee \otimes_f U \rightarrow \one  \ ,&
  \tilde d_U &: U \otimes_f U^\vee \rightarrow \one  \ ,\nonumber
\end{align}
subject to the usual conditions, see {\it e.g.}\ \cite[Def.\,2.1.1]{BaKi-book}.
We want to deduce the fusion rules of $\Pc(0)$ with representations that have a dual. 
Our analysis is based on the following result 
(see {\it e.g.}\ \cite[p.\,441,\,Cor.\,1\,\&\,2]{Kazhdan:1994}).

\begin{Lem} \label{lem:fus-proj}
Let $\Cc$ be a monoidal category and let  $U \in \Cc$ have (left and right) dual $U^\vee$.
\\
(i) If $f : X \rightarrow Y$ is epi, so are $f \otimes \id_U$ and $\id_U \otimes f$.
\\
(ii) If $P \in \Cc$ is projective, so are $P \otimes U$ and $U \otimes P$.
\end{Lem}

\begin{proof}
Part (i): We will show that $f \otimes \id_U$ is epi, the other case follows analogously. 
We need to check that for any two morphisms $a,b : Y \otimes U \rightarrow Z$, 
$a \circ (f \otimes \id_U) = b \circ (f \otimes \id_U)$ implies $a=b$. For 
$k : Y \otimes U \rightarrow Z$ define $C(k) : Y \rightarrow Z \otimes U^\vee$ as 
(we do not write out associators and unit isomorphisms)
\be
  C(k) = (k \otimes \id_{U^\vee}) \circ (\id_Y \otimes b_U) \ .
\ee
From the properties of duality morphisms it is easy to see that 
$a \circ (f \otimes \id_U) = b \circ (f \otimes \id_U)$ implies $C(a) \circ f = C(b) \circ f$. 
By assumption, $f$ is epi, so that $C(a) = C(b)$. Again by properties of duality maps, 
this in turn implies $a=b$.

\medskip\noindent
Part (ii): We will show that $P \otimes U$ is projective, the other case follows analogously. 
We have to check that for any epi $f : X \rightarrow Y$ and any 
$h : P \otimes U \rightarrow Y$ there is a $h' : P \otimes U \rightarrow X$ such that $h = f \circ h'$. 
Note that since $U$ has left and right dual $U^\vee$, the left and right dual of 
$U^\vee$ can be chosen to be $U$. Applying part (i) to $U^\vee$ we conclude that 
$f \otimes \id_{U^\vee}$ is epi. Given the map $C(h) : P \rightarrow Y \otimes U^\vee$, since 
$P$ is projective there exists a map $q : P \rightarrow X \otimes U^\vee$ such that 
$C(h) = (f \otimes \id_{U^\vee}) \circ q$. The map
\be
  h' = (\id_X \otimes d_U) \circ (q \otimes \id_U)
\ee
satisfies $h = f \circ h'$, as can again easily be seen using the properties of duality maps.
\end{proof}

The map $C$ used in the above proof shows also that
\be
  \text{$U$ has a dual}
  \quad \Rightarrow \quad
  \big\langle A \otimes U , B \big\rangle = 
  \big\langle A , B \otimes U^\vee \big\rangle 
  \text{ for all $A,B$} \ .
\ee
Of the 35 indecomposable representations listed in \eqref{eq:W23-irreps} and \eqref{eq:indec-W-rep}, all representations not in grey boxes have a dual, as do $\Wc$ and $\Qc$. Let us collect these 28 representations in a set $\mathcal{D}$. By Lemma \ref{lem:fus-proj}, for any $U \in \mathcal{D}$, 
$\Pc(0) \otimes_f U$ is projective, that is, we can write
\be
  \Pc(0) \otimes_f U \cong \bigoplus_h n_h(U) \, \Pc(h) \ ,
\labl{eq:P0xU-direct-sum-1}
where the sum runs over the 13 values that $h$ takes for the irreducible representations $\Wc(h)$ according to \eqref{eq:W23-irreps}. The multiplicities $n_h(U)$ are obtained via
\be
  n_h(U) 
  = \big\langle \Pc(0) \otimes_f U , \Wc(h) \big\rangle 
  = \big\langle \Pc(0) , \Wc(h) \otimes_f U^\vee \big\rangle \ ,  
\labl{eq:P0xU-direct-sum-2}
where the right hand side can be computed form the fusion rules in 
\cite[App.\,A.4]{Gaberdiel:2009ug} and the composition series given 
Appendix~\ref{sec:embedding}. Since $\Wc(0) \otimes_f U^\vee = 0$ for all 
$U \in \mathcal{D}$, we conclude $n_0(U) = 0$. The other multiplicities are 
obtained case by case, for example
\be
  n_1(\Wc(\tfrac13)) 
  = \big\langle \Pc(0) , \Wc(1) \otimes_f \Wc(\tfrac13) \big\rangle   
  = \big\langle \Pc(0) , \Rc^{(2)}(0,1)_7 \big\rangle  = 1 \ .
\ee
Altogether, this gives
\begin{align}
  \label{eq:P0x(with-dual)}
  \Wc(\tfrac13)\otimes_{f}\Pc(0) &= \Pc(\tfrac{1}{3})\oplus\Pc(1)\\
  \Wc(\tfrac{10}3)\otimes_{f}\Pc(0) &= \Pc(\tfrac{10}{3})\oplus\Pc(5)\nonumber\\
  \Wc(\tfrac58)\otimes_{f}\Pc(0) &= \Pc(\tfrac{1}{8})\oplus\Pc(\tfrac{5}{8})\nonumber\\
  \Wc(\tfrac{33}8)\otimes_{f}\Pc(0) &= \Pc(\tfrac{33}{8})\oplus\Pc(\tfrac{21}{8})\nonumber\\
  \Wc(\tfrac18)\otimes_{f}\Pc(0) &= \Pc(\tfrac{1}{8})\oplus\Pc(\tfrac{5}{8})\oplus2\,\Wc(\tfrac{-1}{24})\oplus2\,\Wc(\tfrac{35}{24})\nonumber\\
  \Wc(\tfrac{21}8)\otimes_{f}\Pc(0) &= \Pc(\tfrac{33}{8})\oplus\Pc(\tfrac{21}{8})\oplus2\,\Wc(\tfrac{-1}{24})\oplus2\,\Wc(\tfrac{35}{24})\nonumber\\
  \Wc(\tfrac{-1}{24})\otimes_{f}\Pc(0) &= 2\,\Pc(\tfrac{1}{8})\oplus2\,\Pc(\tfrac{21}{8})\oplus2\,\Wc(\tfrac{-1}{24})\oplus2\,\Wc(\tfrac{35}{24})\nonumber\\
  \Wc(\tfrac{35}{24})\otimes_{f}\Pc(0) &= 2\,\Pc(\tfrac{1}{8})\oplus2\,\Pc(\tfrac{21}{8})\oplus2\,\Wc(\tfrac{-1}{24})\oplus2\,\Wc(\tfrac{35}{24})\nonumber\\
  \Rc^{(2)}(0,2)_7\otimes_{f}\Pc(0) &= \Pc(1)\oplus\Pc(2)\nonumber\\
  \Rc^{(2)}(2,7)\otimes_{f}\Pc(0) &= \Pc(5)\oplus\Pc(7)\nonumber\\
  \Rc^{(2)}(0,1)_5\otimes_{f}\Pc(0) &= 2\,\Pc(\tfrac{10}{3})\oplus2\,\Pc(\tfrac{1}{3})\oplus\Pc(1)\oplus\Pc(2)\nonumber\\
  \Rc^{(2)}(1,5)\otimes_{f}\Pc(0) &= 2\,\Pc(\tfrac{10}{3})\oplus2\,\Pc(\tfrac{1}{3})\oplus\Pc(5)\oplus\Pc(7)\nonumber\\
  \Rc^{(2)}(0,2)_5\otimes_{f}\Pc(0) &= 2\,\Pc(\tfrac{10}{3})\oplus\Pc(1)\oplus\Pc(2)\nonumber\\
  \Rc^{(2)}(2,5)\otimes_{f}\Pc(0) &= 2\,\Pc(\tfrac{10}{3})\oplus\Pc(5)\oplus\Pc(7)\nonumber\\
  \Rc^{(2)}(0,1)_7\otimes_{f}\Pc(0) &= 2\,\Pc(\tfrac{1}{3})\oplus\Pc(1)\oplus\Pc(2)\nonumber\\
  \Rc^{(2)}(1,7)\otimes_{f}\Pc(0) &= 2\,\Pc(\tfrac{1}{3})\oplus\Pc(5)\oplus\Pc(7)\nonumber\\
  \Pc(\tfrac{1}{3})\otimes_{f}\Pc(0) &= 2\,\Pc(\tfrac{10}{3})\oplus2\,\Pc(\tfrac{1}{3})\oplus2\,\Pc(1)\oplus2\,\Pc(5)\nonumber\\
  \Pc(\tfrac{10}{3})\otimes_{f}\Pc(0) &= 2\,\Pc(\tfrac{10}{3})\oplus2\,\Pc(\tfrac{1}{3})\oplus2\,\Pc(1)\oplus2\,\Pc(5)\nonumber\\
  \Pc(\tfrac{5}{8})\otimes_{f}\Pc(0) &= 2\,\Pc(\tfrac{1}{8})\oplus2\,\Pc(\tfrac{33}{8})\oplus2\,\Pc(\tfrac{21}{8})\oplus2\,\Pc(\tfrac{5}{8})
    \oplus4\,\Wc(\tfrac{-1}{24})\oplus4\,\Wc(\tfrac{35}{24})\nonumber\\
  \Pc(\tfrac{21}{8})\otimes_{f}\Pc(0) &= 2\,\Pc(\tfrac{1}{8})\oplus2\,\Pc(\tfrac{33}{8})\oplus2\,\Pc(\tfrac{21}{8})\oplus2\,\Pc(\tfrac{5}{8})
  \oplus4\,\Wc(\tfrac{-1}{24})\oplus4\,\Wc(\tfrac{35}{24})\nonumber\\
  \Pc(\tfrac{1}{8})\otimes_{f}\Pc(0) &= 2\,\Pc(\tfrac{1}{8})\oplus2\,\Pc(\tfrac{33}{8})\oplus2\,\Pc(\tfrac{21}{8})\oplus2\,\Pc(\tfrac{5}{8})
  \oplus4\,\Wc(\tfrac{-1}{24})\oplus4\,\Wc(\tfrac{35}{24})\nonumber\\
  \Pc(\tfrac{33}{8})\otimes_{f}\Pc(0) &= 2\,\Pc(\tfrac{1}{8})\oplus2\,\Pc(\tfrac{33}{8})\oplus2\,\Pc(\tfrac{21}{8})\oplus2\,\Pc(\tfrac{5}{8})
  \oplus4\,\Wc(\tfrac{-1}{24})\oplus4\,\Wc(\tfrac{35}{24})\nonumber\\
  \Pc(1)\otimes_{f}\Pc(0) &= 4\,\Pc(\tfrac{10}{3})\oplus4\,\Pc(\tfrac{1}{3})\oplus2\,\Pc(1)\oplus2\,\Pc(2)
  \oplus2\,\Pc(5)\oplus2\,\Pc(7)\nonumber\\
  \Pc(2)\otimes_{f}\Pc(0) &= 4\,\Pc(\tfrac{10}{3})\oplus4\,\Pc(\tfrac{1}{3})\oplus2\,\Pc(1)\oplus2\,\Pc(2)
  \oplus2\,\Pc(5)\oplus2\,\Pc(7)\nonumber\\
  \Pc(5)\otimes_{f}\Pc(0) &= 4\,\Pc(\tfrac{10}{3})\oplus4\,\Pc(\tfrac{1}{3})\oplus2\,\Pc(1)\oplus2\,\Pc(2)
  \oplus2\,\Pc(5)\oplus2\,\Pc(7)\nonumber\\
  \Pc(7)\otimes_{f}\Pc(0) &= 4\,\Pc(\tfrac{10}{3})\oplus4\,\Pc(\tfrac{1}{3})\oplus2\,\Pc(1)\oplus2\,\Pc(2)
  \oplus2\,\Pc(5)\oplus2\,\Pc(7)\nonumber\\
  \Wc\otimes_{f}\Pc(0) &= \Pc(0)\nonumber\\
  \Qc\otimes_{f}\Pc(0) &= \Pc(0)\oplus\Pc(\tfrac{1}{3})\nonumber
\end{align}
This analysis also leads to a consistency check of the assumption that 
$\Wc(0)$ has a projective cover. We have arrived at the above list without 
making use of associativity of the fusion rules. However, it turns out that 
\eqref{eq:fus-P0-W0} and \eqref{eq:P0x(with-dual)} extend the fusion rules in 
\cite[App.\,A.4]{Gaberdiel:2009ug} in an associative manner.

\subsubsection*{$\Pc(0)$ cannot have a dual}

Unfortunately, the fusion rules of $\Pc(0)$ with itself cannot be determined by this approach. Indeed,
as we shall show in the following, $\Pc(0) \otimes_f \Pc(0)$ cannot be isomorphic to a direct sum of 
projectives. Because of Lemma \ref{lem:fus-proj} this therefore implies that $\Pc(0)$ does not have a 
dual.

Suppose that $\Pc(0) \otimes_f \Pc(0) \cong \bigoplus_h n_h \Pc(h)$, where the sum runs over the 
13 values that $h$ takes for the irreducible representations $\Wc(h)$ according to \eqref{eq:W23-irreps}. 
Then
\be
  n_0 
  = \big\langle \Pc(0) \otimes_f \Pc(0) , \Wc(0) \big\rangle 
  = \big\langle \Pc(0) \otimes_f \Pc(0) \otimes_f \Wc(0) , \Wc^* \big\rangle 
  = 1 \ ,
\ee
by \eqref{eq:P0-fusion-calc-conditions} and \eqref{eq:fus-P0-W0}, and similarly, using also 
\eqref{eq:P0x(with-dual)},
\bea
  \big\langle \Pc(0) \otimes_f \Pc(0) , \Rc^{(2)}(0,2)_7 \big\rangle 
  = \big\langle \Pc(0) \otimes_f \Pc(0) \otimes_f \Rc^{(2)}(0,2)_7 , \Wc^* \big\rangle 
\enl
  = \big\langle 4 \Pc(1) \oplus 4 \Pc(2) \oplus 4 \Pc(5) \oplus 4 \Pc(7) \oplus 8 \Pc(\tfrac13) \oplus 8 \Pc(\tfrac13) , \Wc^*\big\rangle 
  =  4 \ .
\eear\ee
On the other hand,
\be
  \big\langle \Pc(0) \otimes_f \Pc(0) , \Rc^{(2)}(0,2)_7 \big\rangle 
  = 2 n_2 + 2 n_7 + 1 \ ,
\ee
which is always odd. Thus $\Pc(0) \otimes_f \Pc(0)$ cannot be isomorphic to a 
direct sum of projectives. 

\subsection{Composition series}

For any representation $R$ we have 
\begin{equation}
\langle \Pc(h) , R \rangle = \langle \Pc(h) , R^* \rangle 
\end{equation}
since $R$ and $R^\ast$ have the same composition factors. Thus it follows that 
\be
  \langle \Pc(h) , \Pc(0) \rangle
  = \langle \Pc(h) , \Pc(0)^* \rangle
  = \langle \Pc(0) , \Pc(h)^* \rangle
  = \langle \Pc(0) , \Pc(h) \rangle \ .
\ee
We deduce that the composition series of $\Pc(0)$ only contains 
$\Wc(h)$ with $h \in \{0,1,2,5,7\}$ and for 
$h \in \{1,2,5,7\}$, the multiplicity is $2$. 

Now consider the composition series of $\Pc(h)$ for $h = 1,2$ as given in 
Appendix~\ref{sec:embedding}. Let us denote by $N_h$ the 
representation generated by $\Wc(0)$ at level 1 in $\Pc(h)$; for the two cases $h=1,2$, it has the 
composition series
\begin{align}
  N_2 &: 
  \Wc(0) \rightarrow a\Wc(1)\oplus b\Wc(2) \rightarrow \Wc(0)\oplus 2\Wc(5)\oplus 2\Wc(7) \rightarrow \Wc(2)
\nonumber\\
  N_1 &: 
  \Wc(0) \rightarrow c\Wc(1)\oplus d\Wc(2) \rightarrow \Wc(0)\oplus 2\Wc(5)\oplus 2\Wc(7) \rightarrow \Wc(1)
\end{align}
for some $a,b,c,d \in \Zb_{\ge 0}$. 
Independent of $a,b,c,d$, $N_1$ and $N_2$ have $\Wc(1)$ and $\Wc(2)$ at level 3.
Thus $\Pc(0)$ can at most have one copy of $\Wc(1)$ and $\Wc(2)$ at level 1. Thus
$a=b=c=d=1$ (as they have to be at least 1). 
But then $\Wc(0)$ in $N_2$ gets mapped to a linear combination of two of the $\Wc(1)$ at level 
2 in $\Pc(2)$. Since each generates a different set of $\Wc(5)$ and $\Wc(7)$ at level 3, these 
have to appear with multiplicity $2$. This fixes the composition series of $\Pc(0)$ up to the 
appearances of $\Wc(0)$. Since $\Wc(0)$ does not allow non-trivial extensions by itself, the most 
general ansatz is 
\be
  \Wc(0) \rightarrow \Wc(1)\oplus \Wc(2) \rightarrow m\Wc(0)\oplus 2\Wc(5)\oplus 2\Wc(7) 
  \rightarrow \Wc(1)\oplus \Wc(2)
  \rightarrow n \Wc(0) 
\ee
for some $m \ge 1$ and $n \ge 0$. In order to have maps $\Pc(0) \rightarrow N_h$ 
with $h=1,2$ we need $m=1$. Since $\Rc^{(2)}(0,2)_7$ is self-dual, we have
\be
  \big\langle \Rc^{(2)}(0,2)_7, \Pc(0) \big\rangle 
             = \big\langle \Wc , \Pc(0) \otimes_f \Rc^{(2)}(0,2)_7 \big\rangle   
  = \big\langle \Wc , \Pc(1) \oplus  \Pc(2) \big\rangle = 1 \ .
\ee
Hence $\Wc(2)$ at level 3 of $\Pc(0)$ can be joined to at most one $\Wc(0)$. 
Repeating the argument with $\Rc^{(2)}(0,1)_7$ gives the same result for $\Wc(1)$ at level 3. 
Thus $n \in \{0,1,2\}$. But $n=2$ is inconsistent with 
$n+m+1 = \langle \Pc(0) , \Pc(0) \rangle = \langle \Pc(0)^* , \Pc(0)^* \rangle = 3$ 
(the identity, and mapping each of the top $\Wc(0)$ to the bottom $\Wc(0)$). The remaining 
two  possibilities for the composition series of $\Pc(0)$ are therefore
\begin{eqnarray}
\text{(a):} & \qquad  &   \Wc(0) \rightarrow \Wc(1)\oplus \Wc(2) \rightarrow \Wc(0)\oplus 2\Wc(5)\oplus 2\Wc(7) 
  \rightarrow \Wc(1)\oplus \Wc(2)  \nonumber  \\
\text{(b):} & \qquad  & 
\Wc(0) \rightarrow \Wc(1)\oplus \Wc(2) \rightarrow \Wc(0)\oplus 2\Wc(5)\oplus 2\Wc(7) 
  \rightarrow \Wc(1)\oplus \Wc(2)
  \rightarrow \Wc(0) \ . \nonumber
\end{eqnarray}
Very recently, in \cite{AM} Zhu's algebra was determined for the $\Wc_{2,3}$-model. It was found (see 
Proposition~5.1) that Zhu's algebra possesses a representation for which $L_0$ has a Jordan 
block of rank three at $h=0$. Because of Zhu's theorem \cite{Zhu} this implies
that $\Wc_{2,3}$ has a highest weight representation with the same property. Since the projective
covers have surjections to all highest weight representations, $\Pc(0)$ must then 
have a Jordan block of at least this size.
As a consequence, we deduce that the composition series (b) is correct, while (a) does not 
satisfy this condition. 

We should mention in passing that the other findings of \cite{AM}  also fit perfectly with
our conjectured embedding diagrams for the other projective covers $\Pc(h)$. In particular,
according to \cite{AM} there are rank $2$ Jordan blocks for $h=1,2,\tfrac{1}{8},\tfrac{5}{8},\tfrac{1}{3}$, while for
the remaining values of $h$ (namely $h=5,7,\tfrac{33}{8},\tfrac{21}{8},\tfrac{10}{3}$ as well as 
$h=\tfrac{-1}{24},\tfrac{35}{24}$) there are no Jordan blocks. On the face of it this seems to
disagree with our proposed structure for say $\Pc(1)$, since one would guess that the
top `1' should be part of a Jordan block of rank $3$. However, one has to recall that for
the analysis of Zhu's algebra, only the proper highest weight states contribute, {\it i.e.}\
those states that are annihilated by all positive modes, see for example 
\cite{Gaberdiel:1998fs} for a physicist's explanation of this. Given our embedding diagram of 
$\Pc(1)$, it is clear that the top `1' of $\Pc(1)$ is not highest weight since there is an arrow to 
the `0' at level 1. On the other hand, of the two `1's at level 2, there will be one linear
combination that will not map to the `0' at level 3, and that is thus highest weight. It follows that 
on the level of highest weight states $\Pc(1)$ only involves a Jordan block of rank $2$ for
$h=1$, in perfect agreement with the analysis of \cite{AM}. The other cases work similarly.

\small

\newcommand\arxiv[2]      {\href{http://arXiv.org/abs/#1}{\tt #2}}
\newcommand\doi[2]        {\href{http://dx.doi.org/#1}{#2}}
\newcommand\httpurl[2]    {\href{http://#1}{#2}}

\end{document}